\documentclass[11pt]{article}
\usepackage[a4paper]{geometry}
\usepackage{array}
\usepackage{float}
\usepackage{booktabs}
\usepackage{mathtools}
\usepackage{multirow}
\usepackage{booktabs}
\usepackage{makecell}
\usepackage{amsmath}
\usepackage{amsthm}
\usepackage{enumitem}
\usepackage{amssymb}
\usepackage{graphicx}
\usepackage{xcolor}  
\usepackage[authoryear]{natbib}

\usepackage{varioref}
\usepackage{hyperref}

\makeatletter

\providecommand{\algorithmname}{Algorithm}

\usepackage{bm}
\usepackage{thmtools} 
\usepackage{algorithm}
\floatname{algorithm}{\protect\algorithmname}
\usepackage{algorithmicx}
\usepackage{algpseudocode} 
\usepackage{authblk}
\usepackage{subcaption}

\usepackage[utf8]{inputenc}

\usepackage{cleveref}

\declaretheorem[name=Theorem]{theorem}

\declaretheorem[name=Lemma, sibling=theorem]{lemma}

\declaretheorem[name=Assumption]{assumption}
\declaretheorem[name=Definition]{definition}
\declaretheorem[name=Proposition]{proposition}

\declaretheorem[name=Remark, style=remark, unnumbered]{remark}

\newcommand{\cL}{\mathcal{L}}
\newcommand{\cO}{\mathcal{O}}

\newcommand{\E}{\mathbb{E}}

\title{Breaking the Dimensional Barrier: A Pontryagin-Guided Direct Policy Optimization for  Continuous-Time Multi-Asset Portfolio Choice}

\author[1]{Jeonggyu Huh}
\author[2]{Jaegi Jeon}
\author[3]{Hyeng Keun Koo}
\author[4]{Byung Hwa Lim\thanks{Corresponding Author: limbh@skku.edu}}

\affil[1]{\small Department of Mathematics, Sungkyunkwan University, Republic of Korea}
\affil[2]{\small Graduate School of Data Science, Chonnam National University, Republic of Korea}
\affil[3]{\small Department of Financial Engineering, Ajou University, Republic of Korea}
\affil[4]{\small Department of Fintech, SKK Business School, Sungkyunkwan University, Republic of Korea}

\makeatother

\begin{document}

\maketitle

\begin{abstract}
We introduce the Pontryagin-Guided Direct Policy Optimization (PG-DPO) framework for high-dimensional continuous-time portfolio choice. 
Our approach combines Pontryagin’s Maximum Principle (PMP) with backpropagation through time (BPTT) to directly inform neural network–based policy learning, enabling accurate recovery of both myopic and intertemporal hedging demands—an aspect often missed by existing methods. 
Building on this, we develop the Projected PG-DPO (P-PGDPO) variant, which achieves near-optimal policies with substantially improved efficiency. 
P-PGDPO leverages rapidly stabilizing costate estimates from BPTT and analytically projects them onto PMP’s first-order conditions, reducing training overhead while improving precision. 
Numerical experiments show that PG-DPO matches or exceeds the accuracy of Deep BSDE, while P-PGDPO delivers significantly higher precision and scalability. 
By explicitly incorporating time-to-maturity, our framework naturally applies to finite-horizon problems and captures horizon-dependent effects, with the long-horizon case emerging as a stationary special case.

\end{abstract}

\vspace{0.3cm} 

\newpage
\section{Introduction}\label{sec:intro}

The continuous-time consumption and portfolio choice problem has been a cornerstone of financial economics, offering fundamental insights into asset pricing and long-term investment behavior \citep{merton1973intertemporal,campbell1999consumption,campbell2001should}. 
Despite its theoretical elegance, the practical implementation of this framework has remained exceptionally challenging. 
Two obstacles are particularly salient. 
First, portfolio choice problems involving many risky assets and state variables are subject to the curse of dimensionality \citep{bellman1966dynamic}. 
Traditional dynamic programming (DP) approaches, built on solving the Hamilton–Jacobi–Bellman (HJB) equation, readily become intractable as dimensionality increases.
This is one reason why Markowitz’s static mean–variance framework still remains dominant in practice, despite its limitations and the greater realism offered by dynamic models.\footnote{Cochrane wrote: ``Merton’s theory is also devilishly hard to implement, which surely helps to account for its disuse in practice. \citep[p.4]{cochrane2022portfolios}" The appeal of the mean-variance model lies in its tractability, as it reduces portfolio choice to a single-period quadratic optimization problem that remains computationally feasible even with many assets. However, it ignores intertemporal hedging motives, relies on a restrictive single-period formulation, and assumes normally distributed returns so that only means and variances matter. Moreover, it disregards the dynamic evolution of investment opportunities and practical frictions such as transaction costs or portfolio constraints, and suffers from a lack of time consistency in multi-period settings.} 
Consequently, most continuous-time models in the literature have been confined to highly stylized, low-dimensional settings.
Second, even when solutions are available, they often fail to recover the structural decomposition of the optimal policy. A well-established theoretical benchmark, dating back to the seminal work of \citet{merton1973intertemporal}, demonstrates that the optimal portfolio can be decomposed into two critical components: the \emph{myopic demand} and the \emph{intertemporal hedging demand}. The myopic component represents the demand for a mean-variance efficient portfolio given current investment opportunities, while the intertemporal component hedges against expected future changes in those opportunities driven by evolving economic conditions. Accurate identification of both components is crucial for evaluating investor behavior and designing long-term investment strategies.


Recent advances in machine learning have introduced new PDE-based solvers, including Physics-Informed Neural Networks (PINNs), the Deep BSDE method \citep{han2018solving}, and Deep Galerkin Networks (DGNs). 
While these approaches help mitigate the curse of dimensionality, they focus on approximating the value function and are not designed to recover the full structure of the optimal portfolio, particularly the intertemporal hedging demand. 
Moreover, contributions such as \citet{duarte2024machine} restrict attention to infinite-horizon problems, where the optimal solution is stationary and time-independent. 
Although suitable for certain theoretical contexts, such methods are ill-suited for practical applications such as retirement planning, endowment management, or finite-maturity products, where horizon effects and time-varying hedging demands play a central role.

To address these limitations, we propose a new framework, \emph{Pontryagin-Guided Direct Policy Optimization (PG-DPO)}. 
Our approach integrates Pontryagin’s Maximum Principle (PMP) with policy networks trained via backpropagation through time (BPTT), enabling scalable and structurally accurate solutions to high-dimensional continuous-time portfolio choice problems.

PMP characterizes optimal controls through a coupled system of forward and backward equations. 
The forward equation governs the evolution of wealth and state variables that drive asset returns, which we simulate using Monte Carlo methods. 
This allows direct computation of expected utility and removes the need to estimate the value function over future periods—a core but computationally intensive step in DP-based approaches. 
Accordingly, our framework relies solely on policy networks, without requiring a value network. 
The backward equation governs the dynamics of the \emph{costates} (or adjoint variables), corresponding to the first-order partial derivatives of the value function with respect to the state variables.

Modern machine learning frameworks such as PyTorch allow these costates and their derivatives—including second-order partial derivatives of the value function—to be computed automatically via BPTT and automatic differentiation. 
We show that a single forward simulation, combined with BPTT, yields unbiased estimates of the costate processes required by PMP. 
As a result, PG-DPO avoids explicit solution of the backward equation, thereby sidestepping the major bottlenecks of dynamic programming. 
Both the costates and the value function emerge naturally as byproducts of the learning process.

A central innovation is the \emph{Projected PG-DPO (P-PGDPO)} variant. 
P-PGDPO builds on PG-DPO by substituting the estimated derivatives of the value function directly into the structural first-order conditions (FOC) of PMP. 
During training, costate estimates typically stabilize well before the policy networks are fully converged. 
P-PGDPO leverages this property by projecting the stabilized costates and their derivatives onto the analytical manifold defined by PMP’s first-order conditions. 

The algorithm proceeds in two stages. 
In Stage 1, the networks are trained until the costate estimates and their derivatives stabilize. 
In Stage 2, these estimates are substituted into the first-order conditions to construct the optimal policies. 
The resulting controls adhere closely to the structural form of the true solution, as dictated by PMP, and are recovered quasi-analytically at each time–state point. 
This requires minimal additional training while substantially reducing computational cost. 
By disentangling costate estimation from policy construction, P-PGDPO achieves significant gains in both accuracy and scalability.


We provide a theoretical justification for the P-PGDPO method in Theorem~\ref{thm:policy_gap}. When the objective function surface is relatively flat near the optimum, the first-order condition (FOC) violation for the policy obtained in the warm-up stage tends to be small. We can show that if this FOC deviation is small and the costate estimates obtained via BPTT are sufficiently accurate, then the control policy recovered in the second stage is close to the true optimal policy.

We demonstrate the numerical accuracy of our methods by comparing them to analytic benchmarks under affine dynamics, where the true solution is available from previous work \citep{kim1996dynamic,liu2007portfolio}. 
We further compare against the Deep BSDE approach, which is representative of existing techniques for solving high-dimensional HJB equations. 
Finally, we apply P-PGDPO to long-horizon investment problems and to models with non-affine return dynamics, thereby demonstrating its broader applicability.

Our numerical experiments in Section~\ref{sec:numerical_experiments} show that baseline PG-DPO produces solutions comparable to those of Deep BSDE in terms of mean squared error. 
However, the P-PGDPO variant yields significantly smaller errors. 
This improvement arises because PG-DPO, like other policy optimization algorithms, struggles to reliably discover the true optimum, whereas P-PGDPO overcomes this limitation by directly enforcing the structural first-order conditions of PMP. 
As a result, P-PGDPO not only scales to problems with many assets and state variables, but also accurately disentangles and reconstructs complex intertemporal hedging demands. 
Furthermore, it generates stable solutions in non-affine settings and converges smoothly to the affine benchmark as the nonlinear terms vanish. 
By contrast, the Deep BSDE approach—highly successful in other high-dimensional applications such as option pricing and hedging—fails to capture the full structure of the optimal policy, particularly the intertemporal hedging component.


 Conceptually, our approach is related to recent methods that parameterize control policies with neural networks and optimize the expected objective function directly through simulation and gradient ascent, thereby bypassing the need to solve the HJB equation \citep{zhang2019deep,kopeliovich2024portfolio}. 
Our PMP-based PG-DPO framework, however, is explicitly guided by the continuous-time optimality conditions of PMP. 
Unlike standard reinforcement learning algorithms such as Proximal Policy Optimization (PPO) or Trust Region Policy Optimization (TRPO), which rely on sampled policy gradient estimates, our method computes exact gradients via BPTT. 
This yields lower-variance updates and improved computational efficiency in high dimensions \citep{Kushner2003Stochastic,Borkar2008Stochastic}. 

This work also contributes to the literature on dynamic portfolio choice, which seeks analytical or numerical solutions. 
Most contributions adopt a discrete-time formulation, with the exception of affine continuous-time models. 
For example, \citet{campbell1999consumption,campbell2001should} employ log-linear approximations, while \citet{balduzzi1999transaction}, \citet{lynch2000predictability}, and \citet{lynch2001portfolio} use discrete-space approximations with transaction costs. 
\citet{garlappi2010solving} apply state-variable decomposition, and \citet{brandt2005simulation} propose a Monte Carlo projection of simulated marginal utility onto basis functions—an approach reminiscent of our forward-simulation step, though lacking the scalability and structural guidance of PMP. 

In continuous time, semi-analytical solutions are available under affine dynamics, where the HJB reduces to ODEs \citep{kim1996dynamic,liu2007portfolio,buraschi2010correlation}. 
For non-affine dynamics, no general solution exists. 

Parallel efforts apply reinforcement learning to continuous-time portfolio problems, including policy evaluation, policy gradients, $Q$-learning, and PPO/TRPO methods \citep{WZZ2020,JZ2022a,JZ2022b,JZ2023,ZTY2023}. 
These have been applied to dynamic mean-variance analysis \citep{WZ2020,DDJ2023,huang2024mean} and to Merton-type portfolio problems \citep{dai2023learning}. 
Such approaches learn return dynamics from data, but are typically restricted to few assets and factors. 
For example, \citet{GNKRS2023} propose the FaLPO method, which combines policy learning with model calibration but remains limited in scale. 
In a related line, \citet{cong2021alphaportfolio} introduce the AlphaPortfolio framework, which bypass preliminary estimation of return distributions and directly constructs portfolios using deep reinforcement learning, showing strong out-of-sample performance. By contrast, we assume the data-generating process is known and focus on solving the general continuous-time consumption and portfolio choice problem with many risky assets and multiple state variables. In a complementary direction, \citet{doh2025informationRL} show that reinforcement learning can recover the rational expectations equilibrium in information-based asset market models without relying on explicit market data estimation, providing a theoretical foundation for the use of reinforcement learning in equilibrium settings.

Other contributions closer to ours include \citet{davey2022deep}, who develop a BSDE-based method with portfolio constraints; \citet{duarte2024machine}, who propose a machine learning approach for infinite-horizon settings; and \citet{cheridito2025deep}, who combine generalized policy iteration with PINNs for finite-horizon DP. 
These methods, however, either require explicit BSDE solutions, restrict attention to stationary infinite-horizon problems, or fail to recover the structural decomposition of myopic and intertemporal hedging demands in high dimensions. 
Our approach, by contrast, recovers partial derivatives directly through automatic differentiation of costates, enabling accurate estimation of both the myopic and hedging components.\footnote{As Cochrane emphasizes, partial derivatives of the value function are critical for long-term portfolio choice yet notoriously difficult to obtain: 
\begin{quote}
“Merton’s theory is also devilishly hard to implement, which surely helps to account for its disuse in practice... What are the partial derivatives of the investor’s value function that guide state-variable investments? That is harder still...” \citep[pp.~4--5]{cochrane2022portfolios}.
\end{quote}} 
The aforementioned literature does not demonstrate accurate recovery of intertemporal hedging terms when both the number of assets and state variables are large.\footnote{To our knowledge, the existing literature has demonstrated the ability to solve problems involving at most a dozen risky assets and a dozen state variables, with two notable exceptions. \citet{davey2022deep} consider a setting with 20 assets, but under a constant investment opportunity set, i.e., without state variables influencing asset returns.
\citet{cheridito2025deep} address a case with 25 assets, although without providing evidence regarding the precision of the resulting solution. In contrast, our method substantially improves scalability, handling much larger values of \( n \) and comparably high values of \( k \), while preserving accuracy even in high-dimensional settings. Our companion paper demonstrates that our method is applicable to constrained portfolio optimization problems involving up to 1000 assets \citep{HJKL}. }

In the remainder of this paper, we begin with an introduction to the machine learning methods used for the PG-DPO framework in Section \ref{sec:machine_learning_basic}. Then we present the multi-asset Merton problem with exogenous states in Section~\ref{sec:Merton_multiasset}. Section~\ref{sec:BPTT_multiasset} introduces our Pontryagin-Guided Direct Policy Optimization (PG-DPO) algorithm, its Projected PG-DPO (P-PGDPO) extension, and provides a theoretical justification for the effectiveness of the P-PGDPO approach under specified regularity conditions. Numerical experiments presented in Section~\ref{sec:numerical_experiments} confirm the framework's scalability and near-optimal performance within this testbed, illustrating its potential for high-dimensional applications. Finally, Section~\ref{sec:conclusion} concludes and discusses potential future directions of research.



\section{Machine Learning Basics}\label{sec:machine_learning_basic}

Our proposed Pontryagin-Guided Direct Policy Optimization (PG-DPO) framework is built upon fundamental concepts from machine learning. Before presenting our approach, this section introduces the computational methods and reinforcement learning principles that underpin our methodology. We begin with computational graphs and the algorithm of backpropagation through time (BPTT), which form the basis for efficient gradient calculation in sequential settings. We then discuss reinforcement learning and policy-based methods, highlighting their connection to dynamic portfolio choice and clarifying the distinctions between our framework and conventional RL approaches. Since our method assumes that the dynamics of the financial market are known, it can be classified as a form of model-based reinforcement learning. 

\subsection{Computational Graphs and Backpropagation Through Time}

A computational graph provides a structured way to represent complex mathematical problems as a network of simple operations. In the portfolio optimization problem, the graph encodes the entire process, from the initial wealth and market states to the final objective function $J$, through sequential transformations of wealth and state variables. This representation is not merely a visualization but a computational tool: it enables us to systematically apply the chain rule of calculus across the entire problem.

Once we have this graph, we need an efficient way to understand how small changes in our decisions (the portfolio weights, parameterized by $\theta$) affect the final outcome. This is achieved through backpropagation, which works by starting from the final outcome and moving backward through the graph, systematically applying the chain rule at each step to calculate the contribution of each parameter to the final result. It is an algorithm that efficiently computes the gradient, or the sensitivity of the final output with respect to every parameter in the graph, denoted as $\nabla_{\theta} J$. 

When applied to decision problems evolving over time, such as dynamic portfolio choice, the method is known as \textit{Backpropagation Through Time (BPTT)}. A key advantage of our model-based approach is that the dynamics of the environment, such as the stochastic differential equations (SDEs) for asset prices, are explicitly known. This allows us to compute exact analytical gradients for a given sample path through the simulation, rather than relying on estimation. BPTT is the core mechanism that enables our model to learn from simulated market trajectories by calculating explicit gradients. For a given trajectory, it computes the gradient of the final utility $J$ with respect to the policy parameters $\theta$ by propagating sensitivities backward through time. The influence of wealth in any intermediate time step, $X_{t+1}$, on the final result is found by applying the chain rule sequentially from the terminal state $X_T$:
\begin{equation*}
\frac{\partial J}{\partial X_{t+1}} = \frac{\partial J}{\partial X_T} \frac{\partial X_T}{\partial X_{T-1}} \frac{\partial X_{T-1}}{\partial X_{T-2}} \cdots \frac{\partial X_{t+2}}{\partial X_{t+1}} = \frac{\partial J}{\partial X_T} \prod_{k=t+1}^{T-1} \frac{\partial X_{k+1}}{\partial X_k}.
\end{equation*}
This term represents the gradient "flowing back" from the end of the simulation. The total gradient with respect to the policy parameters $\theta$ is then calculated by summing up the contributions from the policy $\pi_t$ at every time step, weighted by this backward propagated sensitivity:
\begin{equation*}
\frac{\partial J}{\partial \theta} = \sum_{t=0}^{T-1} \frac{\partial J}{\partial X_{t+1}} \frac{\partial X_{t+1}}{\partial \pi_t} \frac{\partial \pi_t}{\partial \theta}
\end{equation*}
Since the dynamics of the financial market are explicitly modeled, our framework can compute analytical gradients along each sample path. This contrasts with model-free methods, which must rely on repeated trial-and-error estimates. In our method, BPTT thus plays a dual role: it provides the gradients necessary for optimizing the policy parameters, and at the same time, it generates pathwise costate estimates that are central to Pontryagin’s Maximum Principle, which will be introduced in Section \ref{sec:PMP_multiasset_state_costate}. 

\subsection{Reinforcement Learning and Policy-Based Methods}

Reinforcement Learning (RL) provides a general framework in which an agent learns to make decisions by interacting with an environment. At each point in time, the agent observes the current state, chooses an action, and receives a reward, with the goal of maximizing cumulative rewards over time. The mapping of this paradigm to dynamic portfolio choice is direct. The agent corresponds to an investor, the environment to a stochastic financial market, the state to the current wealth and economic factors, the action to portfolio allocation and consumption, and the reward to the utility derived from consumption and terminal wealth. 

A central distinction in RL lies between value-based and policy-based methods. Value-based approaches estimate value functions and derive actions indirectly, whereas policy-based methods directly optimize a parameterized policy $\pi_\theta$. In our setting, the policy is represented by a neural network with parameters $\theta$, which maps the current state $s_t$ (e.g., market conditions, current wealth) to an action $a_t=\pi_\theta(s_t)$ (e.g., a specific portfolio allocation). The optimization problem reduces to finding the parameter set $\theta^*$ that  maximize the expected total reward (our objective function $J$):
\begin{equation*}
\theta^* = \arg\max_{\theta} \mathbb{E}[J(\theta)].
\end{equation*}
Parameter updates follow the stochastic gradient ascent rule, 
\begin{equation*}
\theta_{k+1} = \theta_k + \alpha \mathbb{E}\left[
\frac{\partial J(\theta_k)}{\partial \theta}\right]
\end{equation*}
where $\alpha$ is a learning rate and gradients are computed through BPTT.

While our PG-DPO framework is conceptually rooted in this class of algorithms, it departs in important ways. First, it is a model-based RL method. Unlike model-free approaches such as Proximal Policy Optimization or actor-critic schemes, which require repeated trial-and-error exploration to estimate gradients, as seen in \citet{dai2023learning} that tackle the Merton problem, our framework exploits the explicitly known dynamics of financial markets. This allows us to compute exact analytical gradients through BPTT, leading to both greater efficiency and higher numerical accuracy.

Second, our approach is explicitly guided by Pontryagin’s Maximum Principle (PMP). Standard RL methods often converge to policies that deliver nearly optimal utility values but may diverge substantially from the structural form of the true solution. In other domains, such as Atari, Go, or autonomous driving, such structural fidelity is less critical, since achieving a near-optimal utility value alone is often sufficient for strong performance. In dynamic portfolio choice, however, the decomposition of the optimal strategy into myopic demand and intertemporal hedging demand is both theoretically central and practically important. By leveraging PMP, our framework ensures that both components are faithfully captured, yielding policies that are structurally consistent with financial theory.

A further distinction arises in comparison with \citet{duarte2024machine}, whose method employs an actor–critic scheme in which the critic is trained by minimizing the HJB residual in a PINN-like fashion.\footnote{The HJB residual is commonly defined to be the cumulative squared deviation of the left-hand side of the HJB equation from 0.}. However, this approach is designed exclusively for infinite-horizon problems, where the optimal solution is stationary and time-independent. Although appropriate for certain theoretical settings, this restriction makes it less suitable for most practical applications, such as retirement planning, endowment management, or finite-maturity products. By contrast, our PG-DPO framework is explicitly tailored to the finite-horizon setting, incorporating time-to-maturity as a state variable. This enables the algorithm to capture horizon effects, including the time-varying structure of hedging demand and the transition of optimal strategies as maturity approaches. Moreover, since the long-horizon case can be treated as a special instance of finite-horizon formulation, our framework naturally extends to settings such as those analyzed in Section \ref{sec:long_horizon_robustness}, where the horizon is sufficiently long to yield stationary policies. This flexibility is in sharp contrast to the approach of \citet{duarte2024machine}, which is limited to the infinite horizon case and cannot easily accommodate horizon-dependent dynamics that are crucial in practice.

In sum, our method combines the principles of reinforcement learning with the structural discipline of PMP, while extending applicability to finite-horizon continuous-time finance. This integration ensures that learned policies are not only utility-maximizing but also theoretically well-structured, marking a significant departure from existing RL approaches.

\section{The Model and Pontryagin's Maximum Principle}
\label{sec:Merton_multiasset}


\subsection{General Multi-Asset Model}\label{sec:market_model_with_state}

We consider the continuous-time portfolio choice problem introduced by \citet{merton1973intertemporal}, in which a representative investor allocates wealth between  risky assets and a risk-free asset, while continuously consuming over time.\footnote{While we focus on Brownian motion as the source of return risk, the methods can be extended to accommodate jumps. Furthermore, if a discrete-time model is viewed as an approximation of a continuous-time model, our approach can be adapted accordingly.}

The stochastic investment opportunities are characterized by an exogenous \( k \)-dimensional \emph{state process} \( \mathbf{Y}_t =(Y_{1,t},\ldots, Y_{k,t})^\top \in \mathbb{R}^k \), where notation $^\top$ denotes the transpose of a vector (or a matrix), and each component \( Y_{i,t} \) evolves according to an Itô diffusion:
\[
dY_{i,t} = \mu_{Y,i}\bigl(t,\mathbf{Y}_t\bigr)\,dt + \sigma_{Y,i}\bigl(t,\mathbf{Y}_t\bigr)\,dW^Y_{i,t}, \quad i=1,\dots,k,
\]
where the Brownian motions \( W^Y_{i,t} \) are potentially correlated, satisfying
$\operatorname{Cov}\bigl(dW^Y_{i,t},\, dW^Y_{j,t}\bigr) = \phi_{ij}\,dt,$
with constant correlations \( \phi_{ij} \). For instance, \(\mathbf{Y}_t\) could include macroeconomic indices, volatility factors, exchange rates, etc. 

We consider a financial market consisting of \( n \) risky assets and one risk-free asset. The risk-free asset earns an instantaneous rate \( r(t,\mathbf{Y}_t) \), which depends on time and the state variable. The price of the \( i \)-th risky asset, \( S_{i,t} \), evolves according to
\[
dS_{i,t} = S_{i,t}\left( \mu_i(t,\mathbf{Y}_t)\,dt + \sigma_i(t,\mathbf{Y}_t)\,dW^X_{i,t} \right), \quad i = 1,\dots,n,
\]
where \( \mu_i(t,\mathbf{Y}_t) \) is the expected return and \( \sigma_i(t,\mathbf{Y}_t) \) the volatility. Brownian motions \( W^X_{i,t} \) have correlations that satisfy $
\operatorname{Cov}(dW^X_{i,t},\, dW^X_{j,t}) = \psi_{ij}\,dt$, and $\operatorname{Cov}(dW^X_{i,t},\, dW^Y_{j,t}) = \rho_{ij}\,dt,$
with constant correlations \( \psi_{ij} \) and \( \rho_{ij} \).
The investor selects an \( n \)-dimensional \emph{portfolio weight} vector \( \boldsymbol{\pi}_t = (\pi_{1,t}, \ldots, \pi_{n,t})^\top \), where \( \pi_{i,t} \) denotes the proportion of current wealth \( X_t \) invested in the \( i \)-th risky asset at time \( t \).

To clarify the structure of risk in the economy, we define the combined \((n+k)\)-dimensional Brownian motion vector
\(
\mathbf{W}_t = \begin{pmatrix} \mathbf{W}^X_t \\ \mathbf{W}^Y_t \end{pmatrix}.
\) where $\mathbf{W}_t^{X}$ and $\mathbf{W}_t^{Y}$ are Brownian vectors for the risky assets and state variables.
Then, the instantaneous covariance matrix of its increments \( d\mathbf{W}_t \) is denoted by \( \boldsymbol{\Omega} \,dt \), with
\(
\operatorname{Cov}(d\mathbf{W}_t) = \boldsymbol{\Omega}\,dt
=
\begin{pmatrix}
 \Psi & \boldsymbol{\rho} \\
 \boldsymbol{\rho}^\top & \Phi
\end{pmatrix} \,dt,
\)
where \( \Psi = (\psi_{ij}) \in \mathbb{R}^{n \times n} \) is the \textcolor{teal}{correlation} matrix of the asset-specific shocks \( d\mathbf{W}^X_t \), \( \Phi = (\phi_{ij}) \in \mathbb{R}^{k \times k} \) is the correlation matrix of the state shocks \( d\mathbf{W}^Y_t \), and \( \boldsymbol{\rho} = (\rho_{ij}) \in \mathbb{R}^{n \times k} \) captures the cross-correlation between the two sources of risk. 

Throughout the paper, we consider a complete probability space \( \left( \Omega, \mathcal{F}, \mathbb{P} \right) \) that supports a standard \((n + k)\)-dimensional Brownian motion \( \mathbf{W}_t \) and its natural filtration \( \{\mathcal{F}_t\}_{t \geq 0} \), augmented to satisfy the \emph{usual conditions} such as right-continuity and completeness. Then the investor's wealth dynamics is given by 
\begin{equation}\label{eq:wealth}
dX_t
= \left[\,X_t\left(r(t,\mathbf{Y}_t)+
    \boldsymbol{\pi}_t^\top\bigl(\boldsymbol{\mu}(t,\mathbf{Y}_t)
    - r(t,\mathbf{Y}_t)\mathbf{1}\bigr)\right) - C_t\,\right]dt
+ X_t\,\boldsymbol{\pi}_t^\top\boldsymbol{\sigma}(t,\mathbf{Y}_t)\,d\mathbf{W}^X_t,
\end{equation}
where $\mathbf{\mu}(t,\mathbf{Y}_t)$ is the vector of expected returns,  and \( \mathbf{1} \in \mathbb{R}^n \) is a vector of ones.\footnote{For the existence and uniqueness of the wealth process, we impose some regularity conditions for the market parameter in Online Appendix \ref{sec:regularity conditions}}. 
The weight on the risk-free asset is $1-\mathbf{1}^\top\boldsymbol{\pi}_t$, and we define 
$\boldsymbol{\sigma}(t,\mathbf{y}) := \mathrm{diag}\!\big(\sigma_1(t,\mathbf{y}),\ldots,\sigma_n(t,\mathbf{y})\big)\in\mathbb{R}^{n\times n}$.
In order to preclude bankruptcy and ensure solvency, 
we impose the non-negative wealth constraint 
\begin{equation}\label{constraint:non-negative_wealth}
    X_t \ge 0\quad \mbox{for every}\ t\in [0,T].
\end{equation}

The investor's preferences are represented by the following expected utility function:
\begin{align*}
    {\cal U} =  \mathbb{E}_{t,x,\mathbf{y}}\!\left[
\int_t^T e^{-\delta (s-t)}\, U\bigl(C_s\bigr)\,ds
\;+\;
\kappa\,e^{-\delta (T-t)}\, U\bigl(X_T\bigr)
\right], \label{eq:utility_function}
\end{align*}
where \( \delta > 0 \) is the subjective discount rate, and \( \kappa \ge 0 \) governs the strength of the bequest motive, \( \mathbb{E}_{t,x,\mathbf{y}}[\cdot] := \mathbb{E}[\cdot \mid X_t = x,\, \mathbf{Y}_t = \mathbf{y}] \) denotes the conditional expectation, and $U$ is a
continuously differentiable, strictly increasing and strictly concave utility function, which satisfies the following Inada conditions\footnote{The method in this paper can be easily extended to the case where $\lim_{x\to 0} U'(x)<\infty$ by considering the possibility of the corner solution $x=0$ in the first-order condition. Thus, the method covers the large class of utility functions such as the Hyperbolic Absolute Risk Aversion (HARA) family:
\begin{equation*}
  U(x)=
  \begin{cases}
    \dfrac{(a+bx)^{\,1-\gamma}}{1-\gamma}, & \gamma\neq1,\\[6pt]
    \log(a+bx),                             & \gamma=1,
  \end{cases}
  \qquad
  b > 0,\; \gamma > 0,\; a \in \mathbb{R},
\label{eq:hara_utility}
\end{equation*}
defined in the half-line \( \mathcal{D}_U := \{x > x_{\min} \} \), with \( x_{\min} := -a/b \).  }:
\begin{equation*}\label{condition:Inada}
    \lim_{x\to \infty} U'(x)=0, \quad \lim_{x\to 0} U'(x)=\infty.
\end{equation*}


We consider control pairs $(\boldsymbol{\pi}, C)$ that are both economically sensible and mathematically tractable and denote the admissible set by $\mathcal{A}(x,\mathbf{y})$.\footnote{Specifically, we define the admissible set $\mathcal{A}(x,\mathbf{y})$ in Online Appendix \ref{sec:regularity conditions}.} Then the investor’s optimization problem involves dynamically choosing both the portfolio weight vector \( \boldsymbol{\pi}_t \in \mathbb{R}^n \) and the consumption rate \( C_t \), conditional on the current state \( \mathbf{Y}_t \) and wealth \( X_t \). In particular, the value function \( V(t, x, \mathbf{y}) \) is defined as the maximum expected utility achievable when starting from wealth \( x \) and state \( \mathbf{y} \) at time \( t \):
\begin{align}
    V(t, x, \mathbf{y}) = \sup_{(\boldsymbol{\pi}_s, C_s) \in {\mathcal{A}(x,\mathbf{y})}} \mathbb{E}_{t,x,\mathbf{y}}\!\left[
\int_t^T e^{-\delta (s-t)}\, U\bigl(C_s\bigr)\,ds
\;+\;
\kappa\,e^{-\delta (T-t)}\, U\bigl(X_T\bigr)
\right], \label{eq:value_function}
\end{align}
subject to wealth dynamics  \eqref{eq:wealth} and constraint \eqref{constraint:non-negative_wealth}.

\subsection{Pontryagin's Maximum Principle}\label{sec:PMP_multiasset_state_costate}

We introduce Pontryagin's Maximum Principle (PMP) for solving the continuous-time portfolio choice problem formulated in \eqref{eq:value_function}. In general, PMP provides necessary optimality conditions through a system of forward–backward stochastic differential equations (FBSDEs), coupling the dynamics of the state variables with those of the associated costate (adjoint) processes.

In our problem, the system state consists of the investor's wealth \( X_t \) and the exogenous state variables \( \mathbf{Y}_t \in \mathbb{R}^k \). Specifically, for a given optimal control pair \( (\boldsymbol{\pi}_t^*, C_t^*) \), the wealth dynamics and state processes  are governed by the forward SDEs
\begin{align}
dX_t^*
&=
\left[
X_t^*\,\bigl(r(t,\mathbf{Y}_t) + \boldsymbol{\pi}_t^{*\top}\bigl(\boldsymbol{\mu}(t,\mathbf{Y}_t) - r(t,\mathbf{Y}_t)\mathbf{1}\bigr)\bigr) - C_t^*
\right] dt
+ X_t^*\, \boldsymbol{\pi}_t^{*\top} \boldsymbol{\sigma}(t,\mathbf{Y}_t)\, d\mathbf{W}_t^X,
\label{eq:X_forward_optimal}
\\[6pt]
d\mathbf{Y}_t
&=
\boldsymbol{\mu}_Y(t,\mathbf{Y}_t)\,dt + \boldsymbol{\sigma}_Y(t,\mathbf{Y}_t)\,d\mathbf{W}_t^Y,
\label{eq:Y_forward}
\end{align}
with initial conditions \( X_0^* = x_0 > x_{\min} \) and \( \mathbf{Y}_0 = \mathbf{y}_0 \in \mathbb{R}^k \). The corresponding costate processes represent the marginal values with respect to the state variables. In particular, we have two costate processes which are the scalar process, \( \lambda_t := V_x(t, X_t, \mathbf{Y}_t) \), corresponding to the sensitivity to wealth, and the vector process \( \boldsymbol{\eta}_t := V_{\mathbf{y}}(t, X_t, \mathbf{Y}_t) \in \mathbb{R}^k \), capturing the sensitivity to the exogenous state. 


To characterize optimal controls, PMP introduces a Hamiltonian function that incorporates both the instantaneous utility and the effects of the controls on the system's drift and diffusion.  Following the standard stochastic control formulation \citep[e.g.,][]{yong2012stochastic}, the Hamiltonian \( \mathcal{H} \) is defined by
\begin{equation*}\label{eq:H_complete_revised}
\begin{aligned}
\mathcal{H} :=
e^{-\delta t}\, U(C)
+ \lambda \left[ x \bigl(r + \boldsymbol{\pi}^\top (\boldsymbol{\mu} - r\mathbf{1})\bigr) - C \right]
+ \boldsymbol{\eta}^\top \boldsymbol{\mu}_Y 
+ \mathbf{Z}^{\lambda X} \cdot (x\, \boldsymbol{\sigma}^\top \boldsymbol{\pi})
+ (\mathbf{Z}^{\eta Y})^\top \boldsymbol{\sigma}_Y,
\end{aligned}
\end{equation*}
where \( \lambda \) and \( \boldsymbol{\eta} \) are the costate processes,
 \( \mathbf{Z}^{\lambda X} \), \( \mathbf{Z}^{\eta X} \), and \( \mathbf{Z}^{\eta Y} \) represent the diffusion coefficients of the BSDEs associated with \( \lambda_t \) and \( \boldsymbol{\eta}_t \),
 and all functions are evaluated at \( (t, x, \mathbf{y}, \boldsymbol{\pi}, C) \), though we suppress this for simplicity of notation. 

Complementing the forward dynamics, the optimal costate processes \( (\lambda_t^*, \boldsymbol{\eta}_t^*) \) satisfy a system of BSDEs, driven by the partial derivatives of the Hamiltonian \( \mathcal{H} \) with respect to the state variables:
\begin{align}
-d\lambda_t^*
&=
\mathcal{H}_x\bigl(t, X_t^*, \mathbf{Y}_t, \boldsymbol{\pi}_t^*, C_t^*, \lambda_t^*, \boldsymbol{\eta}_t^*, \mathbf{Z}_t^* \bigr)\,dt
-
\mathbf{Z}_t^{\lambda X*}\,d\mathbf{W}_t^X
-
\mathbf{Z}_t^{\lambda Y*}\,d\mathbf{W}_t^Y,
\label{eq:lambda_BSDE}
\\[6pt]
-d\boldsymbol{\eta}_t^*
&=
\mathcal{H}_{\mathbf{y}}\bigl(t, X_t^*, \mathbf{Y}_t, \boldsymbol{\pi}_t^*, C_t^*, \lambda_t^*, \boldsymbol{\eta}_t^*, \mathbf{Z}_t^* \bigr)\,dt
-
\mathbf{Z}_t^{\eta X*}\,d\mathbf{W}_t^X
-
\mathbf{Z}_t^{\eta Y*}\,d\mathbf{W}_t^Y.
\nonumber\label{eq:eta_BSDE}
\end{align}
Here, \( \mathcal{H}_x := \partial \mathcal{H} / \partial x \) and \( \mathcal{H}_{\mathbf{y}} := \nabla_{\mathbf{y}} \mathcal{H} \) denote the partial derivatives of the Hamiltonian with respect to wealth and the state vector, respectively. The terms \( \mathbf{Z}_t^{\lambda X} \in \mathbb{R}^{1 \times n} \), \( \mathbf{Z}_t^{\lambda Y} \in \mathbb{R}^{1 \times k} \), \( \mathbf{Z}_t^{\eta X} \in \mathbb{R}^{k \times n} \), and \( \mathbf{Z}_t^{\eta Y} \in \mathbb{R}^{k \times k} \) are adapted processes\footnote{These terms arise from the martingale representation theorem applied to the backward dynamics of the costate variables. See \citet{karatzas1991brownian} for a detailed treatment of the martingale representation theorem.} and collectively denoted by \( \mathbf{Z}_t \). Note that \( \mathcal{H}_x \) and \( \mathcal{H}_{\mathbf{y}} \) generally depend not only on the state, costate, and control variables, but also on the \(\mathbf{Z}_t\) terms. This reflects the coupling between the forward and backward systems, which arises from the second-order sensitivity of the value function. The terminal conditions for the BSDEs are determined by the structure of the objective function:
\[
\lambda_T^* = \kappa\,e^{-\delta T}\,U'\bigl(X_T^*\bigr), \qquad \boldsymbol{\eta}_T^* = \mathbf{0}.
\]
The terminal condition $\eta_T=0$ arises because the terminal utility depends only on $X_T$ 
and not directly on the state $Y_T$.



The core requirement of PMP is that the optimal pair \( (\boldsymbol{\pi}_t^*, C_t^*) \) maximizes the Hamiltonian \( \mathcal{H} \) \eqref{eq:H_complete_revised} pointwise in time, almost surely, given the optimal state and costate processes \( (X_t^*, \mathbf{Y}_t, \lambda_t^*, \boldsymbol{\eta}_t^*, \mathbf{Z}_t^*) \).   We derive the following FOC for the consumption rate $C_t$\footnote{The FOC is valid because of the Inada condition \eqref{condition:Inada}. If $\lim_{x\to 0} U'(x)<\infty$, we need to consider also the possibility of a corner solution,  $C_t^*=0.$} 
\begin{equation*}
e^{-\delta t}\,U'(C_t^*) = \lambda_t^*.
\label{eq:foc_C}
\end{equation*}
This implies that optimal consumption is the inverse function of marginal utility, \( C_t^* = (U')^{-1}(e^{\delta t} \lambda_t^*) \). For the portfolio weights \( \boldsymbol{\pi}_t \), we differentiate \( \mathcal{H} \) with respect to \( \boldsymbol{\pi} \), accounting for both the drift and diffusion terms that depend on the control. 
Setting the gradient \( \nabla_{\boldsymbol{\pi}} \mathcal{H} = \mathbf{0} \) and assuming $X_t^*\ne 0$, we obtain
\begin{equation}
\lambda_t^* \bigl( \boldsymbol{\mu}(t,\mathbf{Y}_t) - r(t,\mathbf{Y}_t)\, \mathbf{1} \bigr)
+ \boldsymbol{\sigma}(t,\mathbf{Y}_t)\, (\mathbf{Z}_t^{\lambda X *})^\top
= \mathbf{0}.
\label{eq:foc_pi}
\end{equation}
This condition implicitly defines the optimal portfolio \( \boldsymbol{\pi}_t^* \) through its dependence on the costate \( \lambda_t^* \) and the martingale component \( \mathbf{Z}_t^{\lambda X *} \), which captures the sensitivity of the marginal value of wealth to risk sources.

To obtain an explicit expression for \( \boldsymbol{\pi}_t^* \), we need to relate the coefficient \( \mathbf{Z}^{\lambda X *} \) of the martingale component in BSDE \eqref{eq:lambda_BSDE} to the derivatives of the value function. This relationship can be derived by applying Itô’s lemma to \( \lambda_t = V_x(t, X_t, \mathbf{Y}_t) \) and matching terms with the dynamics of the costate process in \eqref{eq:lambda_BSDE}.\footnote{The full derivation is presented in Online Appendix~\ref{sec:ito_derivation}.} This procedure yields an expression for \( \mathbf{Z}^{\lambda X *} \) in terms of the second-order derivatives \( V_{xx} \), \( V_{x\mathbf{y}} \), the optimal controls \( \boldsymbol{\pi}_t^* \), and model parameters \( (\boldsymbol{\sigma}, \boldsymbol{\sigma}_Y, \Psi, \boldsymbol{\rho}) \). Substituting this into the first-order condition \eqref{eq:foc_pi} and solving the resulting equation for \( \boldsymbol{\pi}_t^* \) leads to the well-known optimal portfolio formula, originally derived by \citet{merton1973intertemporal}:

\begin{equation*}
\label{eq:pi_star}
\begin{aligned}
\boldsymbol{\pi}_t^{*}
&= -\frac{V_x}{X_t V_{xx}}\,
    \Sigma^{-1}(t,\mathbf{Y}_t)\,
    \bigl(\boldsymbol{\mu}(t,\mathbf{Y}_t)-r(t,\mathbf{Y}_t)\,\mathbf{1}\bigr) \\
&\qquad\qquad -\frac{1}{X_t V_{xx}}\,
    \Sigma^{-1}(t,\mathbf{Y}_t)\,
    \boldsymbol{\sigma}(t,\mathbf{Y}_t)\, \boldsymbol{\rho}\, \boldsymbol{\sigma}_Y(t,\mathbf{Y}_t)\,
    V_{x\mathbf{Y}}(t,X_t,\mathbf{Y}_t).
\end{aligned}
\end{equation*}

Here, \( \Sigma(t,\mathbf{Y}_t) := \boldsymbol{\sigma}(t,\mathbf{Y}_t)\, \Psi\, \boldsymbol{\sigma}(t,\mathbf{Y}_t)^\top \) is the instantaneous covariance matrix of risky asset returns. Notice that all the derivatives in the optimal portfolio can be represented by the costate processes and their derivatives. Specifically, we have \( V_x = \lambda_t^* \), \( V_{xx} = \partial_x \lambda_t^* \), and \( V_{x\mathbf{y}} = \nabla_{\mathbf{y}} \lambda_t^* \in \mathbb{R}^{k} \), and they are evaluated along the optimal trajectory \( (t, X_t^*, \mathbf{Y}_t) \).

This expression decomposes the optimal portfolio into two components: a \emph{myopic demand}, which optimizes the static trade-off between expected return and risk under fixed investment opportunities, and an \emph{intertemporal hedging demand}, which adjusts the portfolio to hedge against stochastic shifts in future investment opportunities. The hedging component is shaped by the correlation structure \( \boldsymbol{\rho} \) between asset returns and state innovations, and by the cross-derivative \( V_{x\mathbf{y}} \), which captures how the marginal value of wealth responds to changes in the economic state.

While both PMP and the Hamilton-Jacobi-Bellman (HJB) equation characterize optimality in continuous-time control problems, they differ fundamentally in formulation and computational implications. The HJB approach is based on dynamic programming, yielding a single  nonlinear PDE for the value function $V(t,x,\mathbf{y})$. Solving the HJB equation requires evaluating first- and second-order derivatives of $V$ over the entire time–state domain and enforcing the PDE residual—the deviation of the left-hand side of the equation from zero—to vanish. This provides both necessary and sufficient conditions for optimality under regularity but becomes computationally intractable when the state dimension is high due to the curse of dimensionality.

By contrast, PMP formulates optimality through a two-point boundary value problem in the form of FBSDEs. The forward SDEs describe the evolution of state variables, while the backward SDEs define the costate processes and their terminal conditions. Optimal controls are obtained by maximizing the Hamiltonian pointwise in time, given the current state and costate values. PMP provides first-order necessary conditions without requiring a global solution for the value function, and its formulation naturally supports stochastic simulation and pathwise gradient computation via BPTT.

For high-dimensional problems, this structural difference is critical. The HJB method must approximate a PDE over the full state grid or via global function approximation, whereas PMP can estimate costates along Monte Carlo trajectories, reducing the computational burden. In this paper, we exploit this property by embedding PMP's first-order conditions directly into the learning loop, avoiding explicit HJB residual minimization and enabling scalability. Moreover, while neural PDE solvers such as the Deep BSDE method approximate HJB solutions indirectly via simulated trajectories, our PMP-based approach bypasses the global PDE entirely, allowing more accurate recovery of structural components, particularly the intertemporal hedging term, in high-dimensional portfolio problems.

\subsubsection{Application: Multivariate OU Process}
\label{sec:multiasset_multifactor_KO}

We apply the PMP framework to a \emph{multi-factor} environment. Specifically, we  consider \( n \) risky assets whose excess returns are driven by a \( k \)-dimensional Ornstein--Uhlenbeck (OU) factor vector \( \mathbf{Y}_t \in \mathbb{R}^k \). Our PMP formulation remains valid even when the state variables are not OU processes, but we use the OU specification here because it admits semi-analytic benchmark solutions and facilitates clear comparison of numerical accuracy. 

The risk premium for asset \( i = 1, \dots, n \) is given by \( \sigma_i (\boldsymbol{\alpha}_i^\top \mathbf{Y}_t) \), where \( \boldsymbol{\alpha}_i \in \mathbb{R}^k \) is the vector of factor loadings for asset \( i \), and \( \sigma_i > 0 \) is the asset’s constant volatility. This formulation captures richer interactions between multiple sources of risk and asset returns.
The factor process \( \mathbf{Y}_t \) evolves according to 
\[
d\mathbf{Y}_t = \kappa_Y (\boldsymbol{\theta}_Y - \mathbf{Y}_t)\, dt + \boldsymbol{\sigma}_Y\, d\mathbf{W}_t^Y,
\]
where \( \kappa_Y \in \mathbb{R}^{k \times k} \) is the mean-reversion matrix, \( \boldsymbol{\theta}_Y \in \mathbb{R}^k \) is the long-run mean, and \( \boldsymbol{\sigma}_Y \in \mathbb{R}^{k \times k} \) is a diagonal matrix of factor-specific volatilities. The Brownian motion \( \mathbf{W}_t^Y \in \mathbb{R}^k \) driving the factors has instantaneous covariance \( \mathrm{Cov}(d\mathbf{W}_t^Y) = \Phi\,dt \). 

The price $S_t^i$ of each asset follows
\(
\frac{dS_t^i}{S_t^i} = \bigl(r + \sigma_i\, \boldsymbol{\alpha}_i^\top \mathbf{Y}_t \bigr)\, dt + \sigma_i\, dW_t^{X,i}, \quad i = 1, \dots, n,
\)
where \( \mathbf{W}_t^X \in \mathbb{R}^n \) is a Brownian motion with instantaneous covariance \( \mathrm{Cov}(d\mathbf{W}_t^X) = \Psi\, dt \). The asset and factor Brownian motions may be correlated,  that is, \( \mathrm{Cov}(dW_t^{X,i}, dW_t^{Y,j}) = \rho_{ij}\,dt \), with cross-correlation matrix \( \boldsymbol{\rho} \in \mathbb{R}^{n \times k} \).


PMP's first-order conditions yield
\[ C_t^* = (U')^{-1}(e^{\delta t} \lambda_t^*),\quad 
\boldsymbol{\pi}_t^* = -\;\frac{1}{X_t^*\, \partial_x \lambda_t^*} \, \Sigma^{-1} \left\{
\lambda_t^*\, \boldsymbol{\sigma} \mathbf{A} \mathbf{Y}_t + \boldsymbol{\sigma} \boldsymbol{\rho} \boldsymbol{\sigma}_Y \, \partial_{\mathbf{Y}} \lambda_t^*
\right\}, \label{eq:optimal_policies_FOC}
\]
where
 \( \boldsymbol{\sigma} = \operatorname{diag}(\sigma_1, \dots, \sigma_n) \in \mathbb{R}^{n \times n} \) is the volatility matrix,
 \( \Sigma = \boldsymbol{\sigma} \Psi \boldsymbol{\sigma}^\top  \in \mathbb{R}^{n \times n} \) is the return covariance matrix,
 \( \mathbf{A} = (\boldsymbol{\alpha}_1, \dots, \boldsymbol{\alpha}_n)^\top \in \mathbb{R}^{n \times k} \) is the factor loading matrix,
 \( \boldsymbol{\sigma}_Y \in \mathbb{R}^{k \times k} \) is the factor volatility matrix, and
 \( \partial_{\mathbf{Y}} \lambda_t^* \in \mathbb{R}^k \) is the gradient of the costate with respect to \( \mathbf{Y}_t \).

The optimal portfolio can be separated into two components. The first component is a \emph{myopic demand}, driven by the instantaneous excess returns \( \boldsymbol{\sigma} \mathbf{A} \mathbf{Y}_t \) and the risk-return trade-off under current conditions. The second one is an \emph{intertemporal hedging demand}, which partially offsets the effects of changes in investment opportunities, depending on the sensitivity of marginal utility to changes in \( \mathbf{Y}_t \) (\( \partial_{\mathbf{Y}} \lambda_t^* \)), and the correlation structure between asset and factor shocks (\( \boldsymbol{\rho} \), \( \Psi \), and \( \Phi \)). 

For the special case with no intermediate consumption (\( C_t \equiv 0 \)), the model admits analytic solutions via an exponential-quadratic value function ansatz, leading to a system of matrix Riccati ODEs. These closed-form (exponential–quadratic) solutions are standard; see \citet{kim1996dynamic} and \citet{liu2007portfolio}. We use them as analytical benchmarks.

\subsubsection{Application: Non-Affine Models}\label{sec:non-affine models}

We can extend the model to allow for non-affine asset return dynamics, providing a more flexible representation of the relationship between state variables and risk premia. While the OU factor process in Section \ref{sec:multiasset_multifactor_KO} offers analytical tractability, many realistic market environments exhibit nonlinear effects that cannot be captured by an affine specification. This subsection introduces such nonlinearities in a controlled manner, enabling us to test the robustness of our PMP-based methods in settings where closed-form solutions are not available.  

We retain the same $k$-dimensional OU process for the state factors, \(\mathbf{Y}_t = [Y_{1,t}, \dots, Y_{k,t}]^T\), where $\kappa_Y, \theta_Y,$ and $\sigma_Y$ are as defined in Section \ref{sec:multiasset_multifactor_KO}. The expected returns of risky assets are modified to 
$$ \boldsymbol{\mu}_S(t, \mathbf{Y}_t) = \text{diag}(\boldsymbol{\sigma}) (\boldsymbol{\alpha} \mathbf{Y}_{\beta,t}) + r\mathbf{1} = \text{diag}(\boldsymbol{\sigma}) \left( \boldsymbol{\alpha} (\mathbf{Y}_t + \boldsymbol{\beta} \mathbf{Y}_t^2) \right) + r\mathbf{1},
$$
where \(\boldsymbol{\alpha}\in\mathbb{R}^{n\times k}\) is the factor loading matrix, \(\boldsymbol{\beta}\) controls the strength and direction of the nonlinear impact of each state variable.\footnote{The multiplication and addition in $\boldsymbol{\alpha} (\mathbf{Y}_t + \boldsymbol{\beta} \mathbf{Y}_t^2)$ is understood as element-wise operations. In addition, $Y_t^2$ denotes the element-wise squaring of $Y_t$. Moreover, the parameter $\boldsymbol{\beta}$ is generated as \(\boldsymbol{\beta} = \beta_{\text{norm}} \cdot \mathbf{u}\), where \(\beta_{\text{norm}} = ||\boldsymbol{\beta}||_2\) sets the overall magnitude of nonlinearity and $\mathbf{u}$ is a randomly generated unit vector (fixed across experiments for consistent). } This specification increases the complexity of the costate dynamics and makes analytical solutions infeasible, motivating a purely numerical approach. 
When \(\boldsymbol{\beta}=0\), the model reduces exactly to the affine specification of Sections \ref{sec:multiasset_multifactor_KO}, allowing us to check that our numerical methods converge to the known affine benchmark.

From the perspective of PMP, the introduction of non-affine terms changes the functional form of the drift $\mu(t,\mathbf{Y}_t)$ in the Hamiltonian, but the overall structure of the first-order conditions remains identical. The optimal controls are still given by
\[ C_t^* = (U')^{-1}(e^{\delta t} \lambda_t^*),\quad 
\boldsymbol{\pi}_t^* = -\;\frac{1}{X_t^*\, \partial_x \lambda_t^*} \, \Sigma^{-1}(t, \mathbf{Y}_t) \left\{
\lambda_t^*\,(\mu(t, \mathbf{Y}_t)-r\mathbf{1}) + \boldsymbol{\sigma} (t, \mathbf{Y}_t)\boldsymbol{\rho} \boldsymbol{\sigma}_Y (t, \mathbf{Y}_t)\, \partial_{\mathbf{Y}} \lambda_t^*
\right\},
\]
where $\mu(t, \mathbf{Y}_t), \sigma(t, \mathbf{Y}_t)$, and $\Sigma(t, \mathbf{Y}_t)$ incorporate the nonlinear effects of the state variables. 

This non-affine formulation plays a dual role in our study. When $\boldsymbol{\beta}\rightarrow 0$, the PMP-based methods need to reproduce the affine benchmark results from Section \ref{sec:multiasset_multifactor_KO}, providing a direct accuracy check. In addition, for a larger $\boldsymbol{\beta}$, the problem departs from the affine setting, allowing us to evaluate the stability and accuracy of the algorithms when no closed-form solution exists. 

In Section \ref{sec:non_affine_robustness}, we examine numerical results of this non-affine model with various values of $\boldsymbol{\beta}$, examining how the learned policies deviate from the affine benchmark as the degree of nonlinearity increases, and assessing how effectively the PMP-based approach captures both the myopic and intertemporal hedging demands in this more general environment.

\subsection{PMP and Neural Networks}

While PMP provides a set of first-order optimality conditions through the coupled FBSDE system, solving these equations exactly remains challenging in high-dimensional portfolio choice problems. The costate processes $\lambda_t=V_x(t, X_t, \mathbf{Y_t})$ and $\mathbf{\eta}_t=V_\mathbf{y}(t, X_t, \mathbf{Y_t})$ depend on derivatives of the value function that are typically unavailable in closed form. Traditional numerical approaches attempt to discretize the state space and solve the resulting system of equations, but this approach quickly becomes infeasible as the number of state variables or assets increases.

We address this challenge by developing a novel \emph{gradient-based} numerical approach using deep learning. The main idea is to parameterize the investor’s policy pair \( (\boldsymbol{\pi}_t, C_t) \) with neural networks:
$$\boldsymbol{\pi}_t=\boldsymbol{\pi}_\theta(t, x, \mathbf{y}), \quad C_t=C_\phi(t, x, \mathbf{y}),$$
where $\theta$ and $\phi$ are the trainable parameters of the respective networks. The role of these networks is to approximate the functional form of the optimal control decisions implied by the PMP first-order conditions. 

From PMP, the optimal portfolio can be expressed in terms of costate and its partial derivatives as in equations \eqref{eq:foc_C} and \eqref{eq:pi_star}. While the network parameters $\theta$ and $\phi$ are updated using the stochastic gradient method, the costate variables $(\lambda_t, \mathbf{\eta}_t)$ and their derivatives can be computed directly from the simulation paths through the automatic differentiation function of deep learning platforms such as Pytorch. All sample paths thus provide unbiased estimates of the necessary PMP quantities, allowing us to bypass the need for a global solution of the value function.

Importantly, this structure allows us to separate the myopic demand and intertemporal hedging demand during training. In particular, the myopic term depends on the current risk premia and the local curvature $V_{xx}$, while the hedging term involves the cross-derivative $V_{x\mathbf{y}}$ and the correlation structure between asset returns and state shocks. Embedding this decomposition into network training enables explicit enforcement of the PMP conditions for both components, rather than relying solely on implicit utility maximization.

The objective function for training is the expected utility, 
\begin{align*}
J(\theta, \phi) := \mathbb{E}\!\left[
\int_0^T e^{-\delta s}\, U\bigl(C_\phi(s,X_s, \mathbf{Y}_s)\bigr)\,ds
\;+\;
\kappa\,e^{-\delta T}\, U\bigl(X_T\bigr)
\right],\label{eq:objective}
\end{align*}
which is estimated over simulated sample paths of $(X_t, Y_t)$ driven by the policy networks. Automatic differentiation through simulation via BPTT yields unbiased gradients of $J(\theta, \phi)$ with respect to all network parameters, while simultaneously producing the costate estimates needed for PMP projection steps in our algorithm. 

The integration of PMP with neural network parameterizations combines the theoretical structure of continuous-time optimal control with the function approximation power of deep learning. In the next section, we will build on this framework to construct the Pontryagin-Guided Direct Policy Optimization (PG-DPO) algorithm and its Projected variant (P-PGDPO), which leverage BPTT-derived costates to achieve scalable solutions even in high-dimensional portfolio choice problems.

\section{The Method}\label{sec:BPTT_multiasset}

We present a novel machine learning approach to solving the optimal portfolio choice problem with multiple assets and numerous factors. Unlike HJB-based methods that rely on solving a global PDE for the value function, PMP provides pointwise first-order conditions that can be estimated pathwise. This difference makes PMP particularly compatible with Monte Carlo simulation and BPTT, thereby avoiding the main scalability bottlenecks of HJB approaches. 

Let us define the \textit{extended objective function}, \( \widetilde{J}(\theta, \phi) \), as the expected value of the conditional utility \( J(t_0, x_0, \mathbf{y}_0; \theta, \phi) \), evaluated over initial states sampled from a given distribution \( \nu \) on a domain \( \mathcal{D} \). Specifically,
\begin{equation}
\label{eq:extended_value_function}
\widetilde{J}(\theta,\phi)
=
\mathbb{E}_{(t_{0},x_{0}, \mathbf{y}_0)\,\sim\,\nu}
\left[
 J(t_0, x_0, \mathbf{y}_0; \theta, \phi)
\right],
\end{equation}
where $J(t_0, x_0, \mathbf{y}_0;\theta,\phi)$ is defined by
\begin{equation*}
\label{eq:policy_value_function}
\begin{split}
J(t_0, x_0, \mathbf{y}_0; \theta, \phi) = \mathbb{E}^{\theta, \phi}\Bigg[ & \int_{t_{0}}^{T} e^{-\delta (u-t_0)}\,U\bigl(C_{\phi}(u,X_u, \mathbf{Y}_u)\bigr)\,du \\
& + \kappa\,e^{-\delta (T-t_0)}\,U\bigl(X_T\bigr)
\;\Bigg| \; X_{t_0}=x_0,\, \mathbf{Y}_{t_0}=\mathbf{y}_0 \Bigg].
\end{split}
\end{equation*}
Here, \( \mathbb{E}^{\theta, \phi}[\cdot \,|\, \cdot] \) denotes the expectation under the dynamics induced by policy \( (C_\phi(t,X_t,\mathbf{Y}_t)\) and \( \boldsymbol{\pi}_\theta(t,X_t,\mathbf{Y}_t)) \). Notice that the extended value function can have an initial value at any time $t<T$. By doing this, the learning results of the neural network parameters are robust across a range of initial conditions, which can significantly reduce exploration costs and lead to rapid convergence of the algorithm. 

To update the parameter $(\theta, \phi)$, we apply the BPTT to Monte Carlo simulations of system trajectories initialized from $\nu$. 
A key feature of this approach is that BPTT not only provides unbiased gradient estimates \( \nabla \widetilde{J} \) for policy optimization, but also computes pathwise estimates of Pontryagin costate variables \( (\lambda_k^{(i)}, \boldsymbol{\eta}_k^{(i)}) \), capturing the marginal value sensitivities as prescribed by PMP. Thus, we refer to our approach as \emph{Pontryagin-Guided Direct Policy Optimization (PG-DPO)}. In this section, we first present the baseline PG-DPO algorithm, which is guided by the PMP. Then, we extend the approach to explicitly exploit costate information and stabilize the costate estimates by Monte Carlo averaging. The second approach will be referred to as \emph{Projected PG-DPO (P-PGDPO)}.

\subsection{Baseline PG-DPO}
\label{sec:pgdpo_baseline_detailed}

The baseline PG-DPO algorithm aims to maximize the extended objective function \(\widetilde{J}(\theta, \phi)\) defined in Eq.~\eqref{eq:extended_value_function} using mini-batch stochastic gradient ascent. The baseline PG-DPO directly parameterizes the control policies $\boldsymbol{\pi}_t$ and $C_t$ as neural networks and optimizes their parameters to maximize the extended objective function while satisfying the PMP-implied state-costate dynamics. The state dynamics follow the forward SDEs for wealth $X_t$ and factors $Y_t$, while the PMP-implied costates $\lambda_t=V_x$ and $\boldsymbol{\eta}_t=V_\mathbf{y}$ evolve according to the backward SDEs given by the Hamiltonian's adjoint equations. The policy networks $\boldsymbol{\pi}_\theta$ and $C_\phi$ map the current state $(t, X_t, Y_t)$ to the control variables, which are then used to simulate forward trajectories of $(X_t, Y_t)$. The key idea is to approximate the gradient \(\nabla \widetilde{J}\) by averaging gradients computed on individual sample paths, where each path starts from a randomly drawn initial state.

In each training iteration, Monte Carlo paths for the driving Brownian motions are first generated. These are used to simulate forward trajectories of $X_t$ and $Y_t$ by using the policies generated from the current policy networks. Starting from the terminal conditions $\lambda_T=\kappa e^{-\delta T}U'(X_T)$ and $\eta_T=0$, the costates $\lambda_t$ and $\eta_t$ are then computed backward along each path using PMP backward SDEs. The expected utility $J(\theta, \phi)$ is evaluated from these simulated paths, and gradients with respect to the network parameters are obtained by BPTT through both the forward state evolution and the backward costate computation. Network parameters are updated by stochastic gradient ascent using an optimizer such as Adam. Algorithm \ref{alg:pgdpo_base} describes our baseline PG-DPO method with a pseudo-code.\footnote{We provide the detailed algorithm for the baseline PG-DPO in Online Appendix \ref{app:PGDPO_algorithm}.}

\begin{algorithm}[t]
\caption{Baseline Pontryagin‑Guided Direct Policy Optimization (PG‑DPO)}
\label{alg:pgdpo_base}
\small 
\begin{algorithmic}[1]
 \Statex \textbf{Inputs:} Policy nets $\pi_\theta,\,C_\phi$; learning‑rate schedule $\{\alpha_j\}_{j=0}^{J_{\max}-1}$;
        max iters $J_{\max}$, batch size $M$, time steps $N$; initial‑state distribution $\nu$;
        utility $U$, discount $\delta$, terminal weight $\kappa$;
        market params $(\boldsymbol\mu,\boldsymbol\sigma,r,\boldsymbol\mu_Y,\boldsymbol\sigma_Y \text{ (diagonal diffusion matrix, } k\times k\text{)},\boldsymbol\Omega)$
 \Statex \textbf{Output:} optimized params $(\theta_\text{opt},\phi_\text{opt})$
 \State Initialize $\theta,\phi$
 \For{$j=0$ \textbf{to} $J_{\max}-1$}
  \State Draw mini‑batch $\{(t_0^{(i)},x_0^{(i)},\mathbf y_0^{(i)})\}_{i=1}^{M}\sim\nu$
  \State $\nabla_{\text{batch}}\gets\mathbf 0$
  \For{$i=1$ \textbf{to} $M$}
   \State $(X_0,\mathbf Y_0)\gets(x_0^{(i)},\mathbf y_0^{(i)})$,
          $\Delta t\gets(T-t_0^{(i)})/N$ 
   \For{$k=0$ \textbf{to} $N-1$}
    \State $t_k\gets t_0^{(i)}+k\Delta t$
    \State $\boldsymbol\pi_k\gets\boldsymbol\pi_\theta(t_k,X_k,\mathbf Y_k)$
    \State $C_k\gets C_\phi(t_k,X_k,\mathbf Y_k)$
    \State Sample $\Delta\mathbf W_k\sim\mathcal N(\mathbf 0,\boldsymbol\Omega\Delta t)$
    \State $\mathbf Y_{k+1}= \mathbf Y_k+\boldsymbol\mu_Y(t_k,\mathbf Y_k)\Delta t
            +\boldsymbol\sigma_Y(t_k,\mathbf Y_k)\,\Delta\mathbf W_k^{\text{Y}}$
    \State $\displaystyle
     X_{k+1}= X_k+\Bigl[X_k\bigl(r(t_k,\mathbf Y_k)
     +\boldsymbol\pi_k^{\!\top}(\boldsymbol\mu(t_k,\mathbf Y_k)-r(t_k,\mathbf Y_k)\boldsymbol 1)\bigr)
     -C_k\Bigr]\Delta t
     +X_k\boldsymbol\pi_k^{\!\top}\boldsymbol\sigma(t_k,\mathbf Y_k)
      \,\Delta\mathbf W_k^{\text{X}}$
   \EndFor
   \State $J^{(i)}=\sum_{k=0}^{N-1} e^{-\delta (t_k - t_0^{(i)})}U(C_k)\Delta t 
          +\kappa e^{-\delta (T - t_0^{(i)})}U(X_N)$ 
   \State $g^{(i)}\gets\nabla_{(\theta,\phi)}J^{(i)}$
   \State $\nabla_{\text{batch}}\gets\nabla_{\text{batch}}+g^{(i)}$
  \EndFor
  \State $(\theta,\phi)\gets(\theta,\phi)+\dfrac{\alpha_j}{M}\nabla_{\text{batch}}$
 \EndFor
 \State \Return $(\theta,\phi)$
\end{algorithmic}
\end{algorithm}

The PG-DPO adheres to the fundamental principle of reinforcement learning: starting from a randomly chosen initial state, it seeks the optimal policy that maximizes the extended value function. We adopt the standard updating rule from backpropagation and take an average of the individual path gradients over the mini-batch. Even though we consider the finite-horizon continuous-time model with multiple assets and state variables, the randomly chosen initial state and the averaging of individual sample paths enable stable convergence. 

We now turn to the differences between our approach and existing neural network–based methods for solving continuous-time control problems. Among the most widely used are the physics-informed neural networks (PINNs) and the Deep BSDE method of \citet{han2016deep} and \citet{E2017Deep}. PINNs aim to directly solve the HJB equation, whereas the Deep BSDE method reformulates the HJB equation as a forward–backward stochastic differential system and trains two neural networks: one to approximate the value function and another to approximate its gradient. 

Although the Deep BSDE approach is powerful for approximating HJB solutions, it recovers the control policy only indirectly from the estimated gradient $Z_t$. As a result, the intertemporal hedging term—depending on the mixed derivative $V_{x\mathbf{y}}$—is not explicitly parameterized but instead inferred indirectly through the network approximation of $Z_t$. In practice, this often yields a weak identification of hedging demand, especially in high-dimensional settings. Even minor inaccuracies in $Z_t$ along these gradient directions can become magnified when constructing the policy, particularly when the return covariance matrix $\Sigma$ is ill-conditioned.

Our PG-DPO framework adopts a different strategy. Costates $(\lambda_t, \eta_t)$ and their derivatives are computed pathwise via BPTT, and the policy is then either constructed directly or projected onto the manifold defined by the PMP first-order conditions. This explicit structural enforcement ensures accurate recovery of both the myopic and intertemporal hedging components—even in cases where Deep BSDE captures the myopic demand well but fails to approximate the more complex hedging demand. As shown in Section \ref{sec:numerical_experiments}, this distinction is critical in high-dimensional portfolio optimization, where our approach consistently yields lower RMSE and more stable convergence.

\paragraph{Role of intermediate BPTT values and the ``Pontryagin-guided'' nature.}

A central insight of PG-DPO is that the intermediate quantities produced during BPTT, namely the pathwise estimates of the costates and their derivatives, can themselves be exploited as guiding signals for policy learning. This gives the method its "Pontryagin-guided" nature, that is, rather than allowing the policy networks to search blindly for an optimum, the training is explicitly steered by the PMP first-order conditions. 

By applying the chain rule backward from the final reward \(J^{(i)}\), BPTT inherently calculates the sensitivity of the reward of the intermediate state variables. These sensitivities are the pathwise unbiased estimates to the Pontryagin costates:
\begin{equation*}
    \lambda_k^{(i)}=\frac{\partial J^{(i)}}{\partial X_k^{(i)}},
    \qquad
    \boldsymbol{\eta}_k^{(i)}=\frac{\partial J^{(i)}}{\partial\mathbf Y_k^{(i)}}.
\end{equation*}
To connect this computational recurrence to the formal PMP framework, we  generalize it to an \(\mathcal{F}_{t_k}\)-adapted process \(\lambda_k = \mathbb{E}[\lambda_k^{(i)} \mid \mathcal{F}_{t_k}]\). This transition from a path-specific value to an adapted process reveals a deep structural correspondence, which we formalize below.

\begin{lemma}[Pathwise BPTT Recurrence]\label{lem:pathwise_bptt_rec_full}
Along any simulated trajectory $i$, the BPTT algorithm computes the backward recursion for the pathwise sensitivity $\lambda_k^{(i)}$ as:
\begin{align}\label{eq:pathwise_rec_full}
  \lambda_k^{(i)}&=\lambda_{k+1}^{(i)}
  +\lambda_{k+1}^{(i)}
     \bigl[r_k+\boldsymbol{\pi}_k^{\!\top}(\boldsymbol{\mu}_k-r_k\mathbf1)\bigr]\Delta t
  +\bigl[e^{-\delta(t_k-t_0)}U'(C_k^{(i)})-\lambda_{k+1}^{(i)}\bigr]
     \frac{\partial C_k^{(i)}}{\partial X_k^{(i)}}\,\Delta t \nonumber\\
  &\quad+\lambda_{k+1}^{(i)}\,
        \boldsymbol{\pi}_k^{\!\top}\boldsymbol{\sigma}_k\,
        \Delta\mathbf W_k^{X,(i)}.
\end{align}
\end{lemma}
Lemma~\ref{lem:pathwise_bptt_rec_full} shows that BPTT naturally recovers the discrete-time adjoint dynamics of PMP along each simulated path. In the limit $\Delta t \to 0$, these recursions converge to the continuous-time PMP adjoint equations, with an $O(\Delta t)$ discretization error. 
This establishes that the backward sensitivities obtained by BPTT align with the PMP costate dynamics.

\begin{lemma}[Quadratic-Covariation Drift]\label{lem:quad_drift_full}
Let $\lambda_{k+1}^{(i)}= \lambda_{k+1} +\widetilde{\mathbf Z}^{\lambda X}_k\Delta\mathbf W_k^{X,(i)} +\widetilde{\mathbf Z}^{\lambda Y}_k\Delta\mathbf W_k^{Y,(i)}$ be the martingale representation of $\lambda_{k+1}^{(i)}$, where $\widetilde{\mathbf Z}^{\lambda X}_k$ is a row vector. The conditional expectation of the stochastic term in Eq.~\eqref{eq:pathwise_rec_full} is:
\[
  \mathbb{E}\!\Bigl[
     \lambda_{k+1}^{(i)}(\boldsymbol{\pi}_k^{\!\top}\boldsymbol{\sigma}_k
     \,\Delta\mathbf W_k^{X,(i)})\Bigm|\mathcal F_{t_k}\Bigr]
  \;=\;
     \bigl(\widetilde{\mathbf Z}^{\lambda X}_k
           +\widetilde{\mathbf Z}^{\lambda Y}_k\boldsymbol{\rho}_{Y\!X}\bigr)
     (\boldsymbol{\sigma}_k^{\!\top}\boldsymbol{\pi}_k)\,\Delta t.
\]
\end{lemma}
Lemma~\ref{lem:quad_drift_full} establishes the martingale representation,\footnote{The Martingale Representation Theorem, cited from \citet{karatzas1991brownian}, is a continuous-time result. Its application to our discrete-time grid is justified by interpreting the discrete process as a sampling of an underlying continuous process, for which the representation holds on each interval $[t_k, t_{k+1}]$.} 
where $\lambda_{k+1}^{(i)}$ denotes the costate along the $(i)$-th path. 
The recursion decomposes into a deterministic drift term and a stochastic fluctuation term, with the quadratic covariation structure ensuring well-defined conditional expectations. 
Consequently, the randomness in the BPTT-based costate update arises solely from martingale components, which vanish on average, leaving the deterministic drift consistent with PMP. 
This explains why costate estimates converge stably in expectation.

\begin{theorem}[BPTT--PMP Correspondence]\label{thm:bptt_pmp_correspondence_full}
Define the \emph{effective} martingale component as $\mathbf Z^{\lambda X}_k:= \widetilde{\mathbf Z}^{\lambda X}_k + \widetilde{\mathbf Z}^{\lambda Y}_k\boldsymbol{\rho}_{Y\!X}$. The adapted costate process $(\lambda_k)_{k=0}^{N}$ satisfies the discrete-time backward stochastic difference equation:
\begin{align}\label{eq:disc_bsde_costate_full}
  \lambda_k
  &=\lambda_{k+1}
  +\Bigl[
      \lambda_{k+1}
        \bigl(r_k+\boldsymbol{\pi}_k^{\!\top}(\boldsymbol{\mu}_k-r_k\mathbf1)\bigr)
     +\mathbf Z^{\lambda X}_k(\boldsymbol{\sigma}_k^{\!\top}\boldsymbol{\pi}_k)
     +\bigl(e^{-\delta t_k}U'(C_k)-\lambda_{k+1}\bigr)
        \frac{\partial C_k}{\partial X_k}
    \Bigr]\Delta t \nonumber\\
  &\quad-\widetilde{\mathbf Z}^{\lambda X}_k\,\Delta\mathbf W^X_k
        -\widetilde{\mathbf Z}^{\lambda Y}_k\,\Delta\mathbf W^Y_k.
\end{align}
Furthermore, the drift term in the square brackets of Eq.~\eqref{eq:disc_bsde_costate_full} is identical to the partial derivative of the Hamiltonian, $\partial_x\mathcal H$, evaluated at the corresponding state and control variables.
\end{theorem}
Theorem \ref{thm:bptt_pmp_correspondence_full} shows that the costate process obtained through BPTT coincides with the PMP adjoint equations, thereby establishing the equivalence between BPTT recursions and the structural optimality conditions of PMP. Put differently, the recursion underlying BPTT is mathematically identical to the adjoint dynamics in PMP.

In sum, we establish the theoretical foundation of the baseline PG-DPO framework. The costates and their pathwise estimates, obtained via BPTT, serve as reliable guiding signals for learning control policies. By exploiting this Pontryagin-guided structure, PG-DPO achieves stable training with provable near-optimality guarantees. Thus, the training procedure is not merely heuristic but grounded in the PMP optimality conditions. While the baseline PG-DPO leverages this link implicitly to obtain effective policy gradients, our Projected PG-DPO variant, introduced in the next section, exploits it explicitly by using the BPTT-computed costates to construct policies aligned with financial theory.

\subsection{Projected PG-DPO}
\label{sec:pgdpo_twostage_detailed} 

While the baseline PG-DPO algorithm described in Section~\ref{sec:pgdpo_baseline_detailed} effectively optimizes the extended objective \( \widetilde{J} \) through end-to-end training of the policy networks \( (\boldsymbol{\pi}_\theta, C_\phi) \), it often requires many iterations for the networks to converge. The Projected PG-DPO (P-PGDPO) variant provides a more computationally efficient alternative, motivated by the empirical observation that the costate estimates---\( (\lambda_k^{(i)}, \boldsymbol{\eta}_k^{(i)}) \) and their derivatives (e.g., \( \partial_x \lambda_k^{(i)}, \partial_{\mathbf{Y}} \lambda_k^{(i)}) \)---tend to stabilize and achieve reasonable accuracy well before full convergence of the policy networks.

P-PGDPO leverages these early-stabilizing costate estimates to construct near-optimal controls by projecting them onto the manifold defined by PMP’s first-order conditions, thereby avoiding the need for fully converged networks to deploy effective policies. The procedure unfolds in two distinct stages, as described below.

\begin{enumerate}
 \item[Stage 1.] \textbf{Warm-Up Phase for Costate Estimation}

 First, we execute the \textbf{Baseline PG-DPO} (see Section~\ref{sec:pgdpo_baseline_detailed}) for a predetermined, relatively small number of iterations \( K_0 \), optimizing \( \widetilde{J} \). The purpose of this stage is \textbf{not} to obtain fully trained networks \((\boldsymbol{\pi}_\theta, C_\phi)\), but rather to run the BPTT process long enough for the mechanism generating the \textbf{costate estimates} 
\[
\lambda_k^{(i)} = \frac{\partial J^{(i)}(\dots)}{\partial X_k^{(i)}}, \quad 
\boldsymbol{\eta}_k^{(i)} = \frac{\partial J^{(i)}(\dots)}{\partial \mathbf{Y}_k^{(i)}}
\]
and their relevant state-derivatives (e.g., \(\partial_x \lambda_k^{(i)}, \partial_{\mathbf{Y}} \lambda_k^{(i)}\)) to stabilize. We denote the policy parameters at the end of this warm-up phase by \((\theta^*, \phi^*)\).

 \item[Stage 2.] \textbf{Analytic Control Deployment}

After completing the \(K_0\) warm-up iterations, the trained policy networks \((\boldsymbol{\pi}_{\theta^*}, C_{\phi^*})\) are typically \textbf{not used directly} to generate controls during deployment. Instead, for any given state \((t_{deploy}, X_k, \mathbf{Y}_k)\) encountered at deployment, we employ a Monte Carlo averaging procedure to obtain reliable costate estimates for use in the PMP formulas.
\end{enumerate}

Specifically, to compute the \textbf{stabilized costate estimates} at \((t_{deploy}, X_k, \mathbf{Y}_k)\), we proceed as follows: simulate \(M_{\text{MC}}\) forward paths starting from \((t_{deploy}, X_k, \mathbf{Y}_k)\) (or nearby states) up to the horizon \(T\), using the \emph{fixed} policy parameters \((\theta^*, \phi^*)\) obtained in Stage~1. For each simulated path \(j=1,\dots,M_{\text{MC}}\), calculate its realized reward 
\[
J^{(j)}(t_{deploy}, X_k, \mathbf{Y}_k; \theta^*, \phi^*)
\]
(with discounting relative to \(t_{deploy}\), as in Algorithm~\ref{alg:pgdpo_two}), and apply BPTT to obtain the path-specific costate estimates \(\lambda_{deploy}^{(j)}, \partial_x \lambda_{deploy}^{(j)}, \partial_{\mathbf{Y}} \lambda_{deploy}^{(j)}\) at the initial node \(t_{deploy}\). The final stabilized estimates are then computed as Monte Carlo averages:
\begin{equation*}
  \lambda_{deploy} \approx \frac{1}{M_{\text{MC}}} \sum_{j=1}^{M_{\text{MC}}} \lambda_{deploy}^{(j)}, \quad
  \partial_x \lambda_{deploy} \approx \frac{1}{M_{\text{MC}}} \sum_{j=1}^{M_{\text{MC}}} \partial_x \lambda_{deploy}^{(j)}, \quad
  \partial_{\mathbf{Y}} \lambda_{deploy} \approx \frac{1}{M_{\text{MC}}} \sum_{j=1}^{M_{\text{MC}}} \partial_{\mathbf{Y}} \lambda_{deploy}^{(j)} .
\end{equation*}

This averaging step substantially reduces the variance inherent in single-path BPTT estimates from Stage~1 and provides more robust inputs for the PMP formulas.

 \begin{algorithm}[t!]
\caption{Projected PG‑DPO (P-PGDPO): Warm‑up Training \& Deployment}
\label{alg:pgdpo_two} 

\begin{algorithmic}[1]
\Statex \textbf{Stage 1: Warm‑up Training}
\Require Inputs of Algorithm~\ref{alg:pgdpo_base}; warm‑up iterations $K_0$
\Ensure Stabilized parameters $(\theta^*,\phi^*)$
\State Run Algorithm~\ref{alg:pgdpo_base} for $K_0$ iterations to obtain $(\theta^*,\phi^*)$

\Statex
\Statex \textbf{Stage 2: Deployment at state $(t,X,\boldsymbol Y)$}
\Require State $(t,X,\boldsymbol Y)$; parameters $(\theta^*,\phi^*)$;
        Monte‑Carlo rollouts $M_{\mathrm{MC}}$, steps $N^{\prime}$
\Ensure Near‑optimal controls $(\boldsymbol\pi^{\mathrm{PMP}},C^{\mathrm{PMP}})$

\State Initialize $\mathcal L,\mathcal L_x,\mathcal L_Y\gets\varnothing$ 
\For{$j=1$ \textbf{to} $M_{\mathrm{MC}}$}
    \State Simulate path from $(t,X,\boldsymbol Y)$ with $(\pi_{\theta^*},C_{\phi^*})$ using $N^{\prime}$ steps. Let $t_{0}'=t$.
    \State Compute realized reward for path $j$:
    \State $\qquad \displaystyle J^{(j)} \gets \sum_{n=0}^{N^{\prime}-1}e^{-\delta (t_n - t_{0}')}U\bigl(C_{t_n}^{(j)}\bigr)\,\Delta t^{\prime} + \kappa e^{-\delta (T - t_{0}')}U\bigl(X_T^{(j)}\bigr)$ 
    \State Evaluate $\lambda_t^{(j)}$, $\partial_x\lambda_t^{(j)}$, $\partial_{\boldsymbol Y}\lambda_t^{(j)}$ (costates at $t_{0}'=t$) via BPTT on $J^{(j)}$ and append to lists $\mathcal L,\mathcal L_x,\mathcal L_Y$.
\EndFor

\State Compute $\hat\lambda_t,\widehat{\partial_x\lambda}_t,\widehat{\partial_{\boldsymbol Y}\lambda}_t$ \Comment{E.g., by averaging estimates from lists $\mathcal L,\mathcal L_x,\mathcal L_Y$}
\State $C^{\mathrm{PMP}}\gets(U')^{-1}\!\bigl(e^{\delta t}\hat\lambda_t\bigr)$ \Comment{Using PMP FOC Eq.~\eqref{eq:foc_C}}
\State Compute $\boldsymbol\pi^{\mathrm{PMP}}$ via PMP FOC Eq.~ \eqref{eq:pi_star} \Comment{Using estimated costates/derivatives}
\State \Return $(\boldsymbol\pi^{\mathrm{PMP}},C^{\mathrm{PMP}})$
\end{algorithmic}
\end{algorithm}

Next, we substitute the \textbf{averaged costate estimates} 
\((\lambda_{deploy}, \partial_x \lambda_{deploy}, \partial_{\mathbf{Y}} \lambda_{deploy})\) 
directly into the analytical formulas derived from the PMP first-order conditions 
(Section~\ref{sec:PMP_multiasset_state_costate}). 

The optimal consumption at \(t_{deploy}\) is given by
\[
  C_{deploy}^{\mathrm{PMP}} 
    = \bigl(U'\bigr)^{-1}\!\bigl(e^{\delta t_{deploy}}\,\lambda_{deploy}\bigr).
\]

The optimal investment portfolio at \(t_{deploy}\) is
\begin{align*}
\boldsymbol{\pi}_{deploy}^{\mathrm{PMP}}
 &= -\frac{1}{X_k \,\partial_x \lambda_{deploy}}\,
    \Sigma^{-1}(t_{deploy}, \mathbf{Y}_k)
    \Biggl\{
      \lambda_{deploy}\,\bigl[\boldsymbol{\mu}(t_{deploy}, \mathbf{Y}_k)
      - r(t_{deploy}, \mathbf{Y}_k)\mathbf{1}\bigr] \nonumber \\[4pt]
 &\qquad\qquad
      + \bigl(\boldsymbol{\sigma}(t_{deploy}, \mathbf{Y}_k)\,
        \boldsymbol{\rho}\,
        \boldsymbol{\sigma}_Y(t_{deploy}, \mathbf{Y}_k)\bigr)\,
        (\partial_{\mathbf{Y}} \lambda_{deploy})
    \Biggr\}.
\end{align*}

This expression coincides with Eq.~\eqref{eq:pi_star}, where the averaged costate estimates \\
\((\lambda_{deploy}, \partial_x \lambda_{deploy}, \partial_{\mathbf{Y}} \lambda_{deploy})\) 
correspond in expectation to the true derivatives 
\((V_x, V_{xx}, V_{x\mathbf{Y}})\) 
at the state \((t_{deploy}, X_k, \mathbf{Y}_k)\).

Finally, the analytically computed controls 
\((C_{deploy}^{\mathrm{PMP}}, \boldsymbol{\pi}_{deploy}^{\mathrm{PMP}})\), 
based on these averaged costates, are applied as the control actions at 
\((t_{deploy}, X_k, \mathbf{Y}_k)\) during deployment.

In sum, the Projected PG-DPO (P-PGDPO) framework proceeds in two stages. 
Stage~1 consists of a short warm-up phase using the baseline PG-DPO procedure to obtain a preliminary policy \((\theta^*, \phi^*)\), during which BPTT produces reasonably stable costate estimates. 
In Stage~2, these estimates are refined via Monte Carlo averaging at deployment time and substituted into the analytical PMP formulas to compute the optimal controls. 
The overall procedure is summarized in Algorithm~\ref{alg:pgdpo_two}.

This two-stage strategy achieves significant computational savings while yielding high-quality, near-optimal controls that explicitly reflect the decomposition of optimal portfolios into myopic and intertemporal hedging components.
In particular, the Monte Carlo averaging in Stage~2 provides reliable inputs for the PMP formulas by mitigating the noise inherent in single-path estimates.

\paragraph{Rationale for the Projected PG-DPO (P-PGDPO) Approach.}

Here we provide a theoretical justification for P-PGDPO. 
The importance of this justification stems from a fundamental challenge facing many simulation-based optimization methods. 
In direct policy optimization, including standard reinforcement learning, a learned policy may achieve a near-optimal objective value (a small \emph{value gap}) while being structurally far from the true optimal policy (a large \emph{policy gap}). 
This discrepancy arises because the objective function is often locally ``flat'' near the optimum: even when the Pontryagin first-order conditions exhibit only small residuals and the resulting utility loss appears at the second-order level, the learned policy can still differ substantially from the true optimum.

Our framework addresses this issue by shifting the focus from the objective value to the satisfaction of Pontryagin's first-order conditions (FOCs). 
The effectiveness of this approach relies on a key property of the HJB equations common in finance: their parabolic nature. 
Parabolic PDEs, much like the heat equation, exhibit a ``smoothing'' effect, ensuring that a well-behaved value function also possesses well-behaved derivatives (costates). 
This is not merely an intuition but a rigorous result of regularity theory in Sobolev spaces, which guarantees that convergence in value entails convergence of its derivatives. 
It is precisely this regularity that allows us to trust the BPTT-derived costates and to reconstruct a robust policy via the PMP formulas.

Building on this expected regularity, we provide a more formal justification for P-PGDPO. We aim to bound the policy gap, $\|\widehat\pi-\pi^*\|$, which is the error between our P-PGDPO policy \(\widehat\pi\) and the true optimal policy \(\pi^*\). The bound depends on two primary sources of error. The first is the FOC residual, denoted by \(\varepsilon\), which measures how much the policy from the warm-up stage violates the optimality conditions of PMP. The second is the BPTT estimation error, \(\delta_{\mathrm{BPTT}}\), in the costate estimate \(\widehat\lambda\). This estimation error itself stems from two numerical sources: the time discretization of the SDEs (with time step \(\Delta t\)) and the Monte Carlo sampling used for the expectation (with batch size \(M\)). The following theorem encapsulates this relationship, with full technical details provided in Online Appendix~\ref{app:math_foundation_publishable}.

\begin{restatable}[Policy Gap Bound]{theorem}{thm:policy_gap}
\label{thm:policy_gap} 
\vspace{2pt}
Let $\pi^*$ be the true optimal policy and $\widehat\pi$ be the policy generated by the P-PGDPO algorithm using a time step $\Delta t$ and batch size $M$. Under the regularity conditions specified in Assumption~\ref{ass:baseline} in Online Appendix \ref{subsec:assumption_publishable}, the gap between the P-PGDPO policy and the true optimum is bounded as follows:
\[
     \|\widehat\pi-\pi^*\|_{L^{q,p}}
     \;\le\;
     C_{\text{tot}} \left( \varepsilon + \kappa_1 \Delta t + \frac{\kappa_2}{\sqrt{M}} \right)
\]
where:
\begin{itemize}[leftmargin=*,labelindent=5pt,itemsep=2pt]
    \item \(\varepsilon\) is the $L^{q,p}$-norm of the Pontryagin FOC residual from the warm-up policy.
    \item \(\kappa_1, \kappa_2\) are positive constants governing the BPTT estimation error, which arises from time discretization (error proportional to \(\Delta t\)) and Monte Carlo sampling (error proportional to \(1/\sqrt{M}\)).
    \item \(C_{\text{tot}}\) is a positive constant that depends on the model parameters but not on $\varepsilon$, $\Delta t$, or $M$.
\end{itemize}
\end{restatable}

The proof of Theorem~\ref{thm:policy_gap} formalizes the intuition outlined above. 
Heuristically, a small FOC-gap ($\varepsilon$) implies that the candidate policy is nearly consistent with the PDE's structural requirements. 
Owing to the smoothing property of uniformly parabolic PDEs—guaranteed by regularity theory in Sobolev spaces (e.g., \citet{dong2009parabolic})—this near-consistency extends beyond the value function itself to its derivatives. 
With the costates thereby shown to be close to their optimal counterparts, projecting them through the PMP map becomes a well-posed and stable operation, yielding a policy that is correspondingly close to the true optimum.

Our formal analysis shows that, under suitable regularity conditions—particularly uniform parabolicity—a small Pontryagin FOC gap ($\varepsilon$) from the warm-up stage guarantees that the P-PGDPO policy \(\widehat{\pi}\) is provably close to the true optimum \(\pi^*\), up to simulation-induced numerical errors. The resulting bound makes clear that achieving a highly accurate policy requires not only an effective warm-up (small \(\varepsilon\)) but also a sufficiently fine time discretization (\(\Delta t\)) and large batch size (\(M\)) to control BPTT estimation error. Moreover, when the objective function is relatively flat near the optimum, the FOC residual \(\varepsilon\) for the warm-up policy tends to be small, so that numerical errors become the dominant source of the final policy gap.
This provides a rigorous justification for the core idea of the P-PGDPO method: substituting numerically controlled and rapidly stabilizing costate estimates into the analytical PMP formulas yields a near-optimal policy.\footnote{Extending this proof framework beyond uniform parabolicity—to hypo-elliptic (\emph{H-class}, \citep{hormander1967hypoelliptic}) or boundary-degenerate (\emph{D-class}, \citep{oleinik2012second}) models, or to general non-affine systems—faces significant analytical challenges, since key assumptions (such as global \(L^p\) regularity or the costate floor) may fail. Despite these hurdles, the fundamental \emph{algorithmic rationale} of P-PGDPO remains compelling for broader applicability. We classify parabolic HJB PDEs relevant to financial markets in Online Appendix \ref{app:parabolic_financial_models}. } 

The central innovation is the decoupling of costate estimation from direct policy construction. 
Standard end-to-end methods conflate these two tasks, making it difficult to capture the fine structure of the policy in a ``flat'' optimization landscape. 
By contrast, P-PGDPO leverages BPTT for what it excels at—efficient gradient (costate) estimation—and then applies the robust analytical structure of PMP to deterministically map these estimates into a high-fidelity policy. 
This principle, reinforced by our numerical results, suggests that P-PGDPO is a practical and powerful tool for a broad class of continuous-time control problems, even when the strict assumptions required for the proof do not hold.

\section{Numerical Results}
\label{sec:numerical_experiments}

In this section, we present numerical experiments to assess both the performance and scalability of the proposed PG-DPO framework, considering its baseline form (PG-DPO) as well as the Projected variant (P-PGDPO). 
As outlined in Section~\ref{sec:multiasset_multifactor_KO}, the experiments consider a multi-asset, multi-factor portfolio optimization problem where the state factors follow Ornstein–Uhlenbeck (OU) processes, together with an extension to a non-affine specification.
While the methodology readily accommodates more general factor dynamics, we focus on OU processes to enable direct comparison with available analytical benchmarks.\footnote{Our evaluation differs from the reverse-engineering procedure of \citet{duarte2024machine}, where the model is designed so that the optimal policy is known by construction, permitting direct comparison in high-dimensional settings even without an analytic benchmark. 
While reverse-engineering is useful when no tractable benchmark exists, its accuracy assessment is conditional on the artificial structure imposed. 
By contrast, our experiments leverage an analytical benchmark from the original affine OU model, allowing RMSEs to be computed against the true optimal policy without altering the economic environment. 
This ensures that the accuracy measure is both exact and economically meaningful, though it requires the existence of such a tractable benchmark.} 

We consider the constant relative risk aversion utility function for experiments:
\begin{equation*}
    U(X) = \begin{cases}
        \dfrac{X^{1-\gamma}}{1-\gamma} & {\rm if}\  \gamma\ne 1\\
        \log X & {\rm if}\ \gamma=1,
    \end{cases}
\end{equation*}
where $\gamma>0$ is the coefficient of relative risk aversion.

We first examine the affine factor specification, deferring the non-affine case to Section~\ref{sec:non_affine_robustness}. 
To isolate the portfolio allocation problem—specifically the decomposition into myopic and intertemporal hedging demands, abstracting from consumption choice—we impose no intermediate consumption (i.e., \(C_t = 0\)), reducing the objective to 
\[
  J = \mathbb{E}[U(X_T)] .
\] 
This setup admits an analytical benchmark solution derived from the associated HJB equation,\footnote{The analytical benchmark follows the standard exponential–quadratic/Riccati formulation; see \citet{kim1996dynamic, liu2007portfolio}. We implement these benchmarks following their notation.} 
which serves as the ground truth for quantitative performance evaluation.

In addition to the analytic benchmark, we compare against a leading neural network–based baseline: the Deep BSDE method of \citet{E2017Deep} and \citet{han2018solving}. 
This approach reformulates the stochastic control problem as a system of FBSDEs and trains neural networks to approximate both the value function and its gradients. 
The approximations are then used to construct the control policy, typically by enforcing or approximating the HJB first-order conditions. 
Network parameters are learned by minimizing residuals in the BSDE dynamics and the terminal condition, evaluated under Monte Carlo simulation.

Throughout this section, we evaluate the portfolio strategies generated by PG-DPO, P-PGDPO, and Deep BSDE against the analytic benchmark. 
Performance is measured by the root mean squared error (RMSE) of portfolio weights on risky assets, computed across a range of problem dimensions.\footnote{Full implementation details—including model parameters, network architectures, training protocols, and evaluation methodology—are provided in Online Appendix~\ref{app:experimental_setup_details}.} 

We begin with a relatively short investment horizon of $T=1.5$ years. 
This choice enables a clean assessment of intrinsic accuracy and convergence properties, minimizing the confounding numerical issues that arise in long-horizon simulations. 
Establishing this baseline allows for a controlled comparison of methods before turning, in Section~\ref{sec:long_horizon_robustness}, to the long-horizon case with targeted algorithmic enhancements. 
Finally, Section~\ref{sec:non_affine_robustness} presents results for models with non-affine asset return dynamics.

\subsection{Single-Asset Models}\label{sec:single_asset_numerical}

We begin with the simplest setting: a single risky asset whose excess return is driven by a $k$-dimensional OU state process ($n=1$). 
This case provides a clear environment in which to examine the accuracy of each method without the numerical instability associated with the large-dimensional covariance matrix $\Sigma$.
In this case, the optimal portfolio is given by
\begin{equation*} 
\pi_t^* = -\frac{1}{\sigma^2 X_t^* \, \partial_x \lambda_t^*}
\left\{
\lambda_t^* \alpha \sigma Y_t
+ \sigma \rho \sigma_Y \, \partial_Y \lambda_t^*
\right\}, \label{eq:pi_single_factor}
\end{equation*}
where $Y_t$ is the dynamic risk-premium process
\[
dY_t = k_Y(\theta_Y - Y_t)\,dt + \sigma_Y \, dW_t^Y .
\]

Because the multivariate OU factor process yields an analytical benchmark, we can compute RMSE separately for the myopic and intertemporal hedging components. This decomposition is a key advantage of our PMP-based approach. 

Table~\ref{tab:rmse_n1_vary_k} reports the root mean squared errors (RMSEs) for $k \in \{1,5,10\}$ state factors, along with the iteration counts at which the minimum errors are achieved. 
Across all configurations, P-PGDPO achieves the lowest RMSE. 
For example, with $k=5$, P-PGDPO attains an RMSE of \(2.50 \times 10^{-4}\), consistently outperforming both the baseline PG-DPO \((3.71 \times 10^{-2})\) and Deep BSDE \((3.09 \times 10^{-2})\). 
The relative performance of baseline PG-DPO and Deep BSDE varies with $k$: for small dimensions ($k=1,5$), Deep BSDE converges faster and achieves lower RMSE than baseline PG-DPO, whereas at $k=10$, PG-DPO \((1.75 \times 10^{-2})\) surpasses Deep BSDE \((8.15 \times 10^{-2})\) in accuracy.\footnote{For $k=10$, Deep BSDE failed to train beyond initialization, so its lowest RMSE was attained at the first iteration. This reflects optimization difficulties in the high-dimensional BSDE loss. Figure~\ref{fig:rmse_1_10_updated} in Online Appendix~\ref{app:convergence_accuracy} illustrates the iteration-by-iteration policy RMSE comparison between methods.} 

These results suggest that both PMP-based and Deep BSDE approaches can mitigate the curse of dimensionality in $k$, but convergence speed and final precision depend critically on the interaction between factor dimension and algorithmic structure.\footnote{In this regime, Monte Carlo sampling error and function approximation quality become the primary computational bottlenecks, rather than an ill-conditioned covariance matrix $\Sigma$.}

\begin{table}[t!]
\centering
\caption{Summary of Policy Errors and  Iteration Counts when $n=1$}
\label{tab:rmse_n1_vary_k}
\begin{tabular}{@{}c c c c@{}}
\toprule
$k$ & Method & RMSE & Iterations at Min. \\
\midrule
\multirow{3}{*}{1}
  & PG-DPO & $3.260\times10^{-2}$ & 1000 \\
  & P-PGDPO         & $1.120\times10^{-3}$ & 800  \\
  & Deep BSDE         & $9.350\times10^{-3}$ & 3600 \\
\cmidrule(lr){1-4}
\multirow{3}{*}{5}
  & PG-DPO & $3.714\times10^{-2}$ & 1400 \\
  & P-PGDPO         & $2.500\times10^{-4}$ & 800  \\
  & Deep BSDE         & $3.094\times10^{-2}$ & 9400 \\
\cmidrule(lr){1-4}
\multirow{3}{*}{10}
  & PG-DPO & $1.751\times10^{-2}$ & 2800 \\
  & P-PGDPO         & $4.710\times10^{-4}$ & 1200 \\
  & Deep BSDE         & $8.146\times10^{-2}$ & 1    \\
\bottomrule
\end{tabular}
\end{table}

Notably, the structural decomposition of the PMP framework allows us to evaluate the myopic and hedging demands separately. 
In the single-asset case, decomposing the P-PGDPO policy shows that nearly all of the residual error originates from the hedging demand. 
This reflects the greater numerical sensitivity of the hedging term, which depends on cross-derivatives of the value function and is amplified by correlations between asset returns and state shocks. 
By contrast, the myopic demand is recovered with near machine precision across all experiments. 
The implications of hedging-demand errors in the multi-asset setting are examined in the next section.

\subsection{Multi-Asset Models}\label{sec:multi_asset_numerical}

We extend the analysis to settings with multiple risky assets ($n>1$). 
As $n$ increases, the optimization problem becomes more challenging due to the higher dimensionality of the control vector and the potential ill-conditioning of the covariance matrix $\Sigma$. 
We consider cases with $n \in \{1,10,50\}$ and factor dimensions $k \in \{2,10\}$, focusing on the affine OU factor model in order to retain analytic benchmark solutions.\footnote{To our knowledge, existing studies in portfolio management typically consider fewer than 30 risky assets when $k \geq 1$. 
For example, \citet{davey2022deep} analyze a model with 20 assets but without state variables influencing returns, while \citet{cheridito2025deep} study a 25-asset case without reporting the precision of the resulting solution.} 

Table~\ref{tab:rmse_vary_n} summarizes the minimum policy RMSEs and the iteration count at which each minimum is attained. 
All methods converge even with a large number of risky assets and state variables, and the resulting RMSE values remain small. 
Notably, increasing the number of state variables does not materially affect either the RMSEs or the iteration counts required for convergence. 

The numerical results are consistent with those from the single-asset model. 
In all cases, P-PGDPO delivers the lowest RMSE, demonstrating robustness to increased dimensionality. 
For example, with $n=50$ and $k=10$, P-PGDPO achieves an RMSE of $1.52\times 10^{-2}$, substantially outperforming both baseline PG-DPO \((1.07\times 10^{-1})\) and Deep BSDE \((8.20\times 10^{-2})\). 
While the relative performance of PG-DPO and Deep BSDE varies with $n$ and $k$, P-PGDPO remains clearly superior across all settings.\footnote{Figure~\ref{fig:rmse_plots_updated} in Online Appendix \ref{app:convergence_accuracy} report RMSE trajectories over iterations for different numbers of risky assets. P-PGDPO exhibits steady convergence as the number of iterations increases.}

The pronounced degradation with increasing $n$, even for small $k$, reflects a fundamental difference between scaling in the asset dimension and scaling in the state dimension. 
Increasing $n$ introduces control-space complexity and numerical fragility. 
The first-order condition for the optimal portfolio involves $\Sigma^{-1}$, and when $n$ is large, $\Sigma$ often becomes ill-conditioned. 
In such cases, even small estimation errors in the costates—particularly the cross-derivatives $V_{x\mathbf{y}}$—are amplified by the matrix inversion, significantly degrading accuracy. 
This effect is structural and cannot be fully offset by generic improvements in function approximation. 

By contrast, increasing $k$ primarily enlarges the state space over which the value function must be approximated. 
Although this raises the difficulty of the approximation task, it does not inherently suffer from the same numerical instability as the large-$n$ case, since $\Sigma$ remains fixed in size and conditioning. 
In other words, scaling in $k$ increases approximation complexity, whereas scaling in $n$ directly amplifies estimation errors through $\Sigma^{-1}$. 
As a result, the neural network–based methods studied here scale more gracefully with $k$ than with $n$.

\begin{table}[t!]

{\centering
\caption{{\centering Summary of Policy Errors and Iteration Counts.}}
\label{tab:rmse_vary_n}

\begin{tabular}{c |c |c c c| c c c}
\toprule
 \multicolumn{2}{c}{} & \multicolumn{3}{c}{$k=2$} & \multicolumn{3}{c}{$k=10$} \\
$n$ & Method & RMSE & Iterations &  HR(\%) & RMSE & Iterations & HR(\%) \\
\midrule
\multirow{3}{*}{1} 
  & PG-DPO & $4.770\times10^{-2}$ & 1200 &  0.1 & $1.751\times10^{-2}$ & 2800 & 0.1 \\
  & P-PGDPO         & $1.990\times10^{-4}$ & 800 &  0.1 & $4.710\times10^{-4}$ & 1200 & 0.1 \\
  & Deep BSDE         & $2.675\times10^{-3}$ & 3600 &  0.1 &$8.146\times10^{-2}$ & 1  & 0.1 \\
\cmidrule(lr){1-8}
\multirow{3}{*}{10}
  & PG-DPO & $1.145\times10^{-1}$ & 4000 &  4.2 & $4.996\times10^{-2}$ & 9800 & 0.9 \\
  & P-PGDPO         & $3.645\times10^{-3}$ & 4600 & 4.2 & $2.799\times10^{-3}$ & 4800 & 0.9 \\
  & Deep BSDE         & $3.835\times10^{-2}$ & 1   & 4.2 & $7.623\times10^{-2}$ & 200  & 0.9 \\
\cmidrule(lr){1-8}
\multirow{3}{*}{50}
  & PG-DPO & $3.875\times10^{-1}$ & 5600 & 11.8 & $1.070\times10^{-1}$ & 8800 & 5.7 \\
  & P-PGDPO         & $7.330\times10^{-2}$ & 9800 & 11.8 & $1.522\times10^{-2}$ & 10000 & 5.7 \\
  & Deep BSDE         & $6.340\times10^{-1}$ & 3800 & 11.8 & $8.199\times10^{-2}$ & 1200 & 5.7 \\
\bottomrule
\end{tabular}
}
\\{\small HR(\%) represents the percentage ratio of hedging demand in the total portfolio.}
\end{table}

Our numerical results confirm this distinction. 
For fixed $n$, increasing $k$ from $2$ to $10$ produces only modest changes in RMSE; in some cases, such as Deep BSDE with $n=50$, performance even improves at higher $k$. 
By contrast, for fixed $k$, raising $n$ from $1$ to $50$ causes a much larger degradation in accuracy across all methods, with the steepest decline observed for baseline PG-DPO.

We focus on the most challenging high-dimensional case with $n=50$ and $k=10$, presenting results for the first asset together with graphical illustrations. Figure~\ref{fig:n50k10_total_comparison_all_methods} compares the sum of myopic and intertemporal hedging demands for the first asset across analytic solution, Baseline PG-DPO, P-PGDPO, and Deep BSDE, evaluated at each method's optimal iteration. 
The plots are drawn against the first state variable $Y_1$ at a representative time to maturity $T-t$, with all other factors $Y_{j\neq 1}$ fixed at their long-term means $\theta_{Y,j}$.

\begin{figure}[t!]
 \centering
 \begin{subfigure}[b]{0.45\textwidth}
  \centering
  \includegraphics[width=\textwidth]{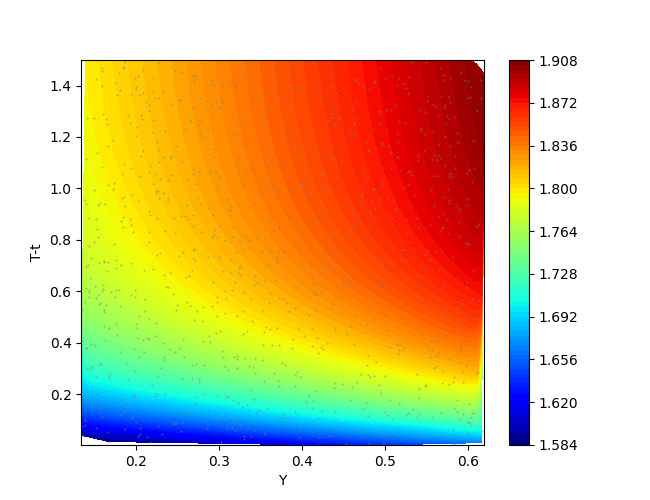}
  \caption{Analytic Policy ($\pi_{1}^{\text{sol}}$)}
  \label{fig:benchmark_policy_total_row}
 \end{subfigure}
 \hfill
 \centering
 \begin{subfigure}[b]{0.45\textwidth}
  \centering
  \includegraphics[width=\textwidth]{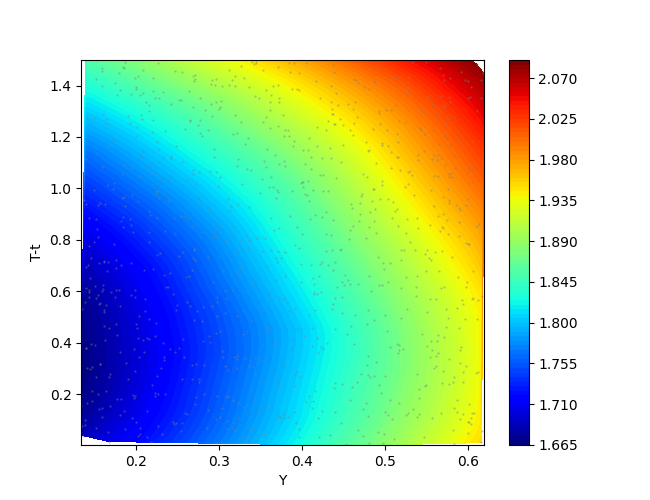}
  \caption{Baseline PG-DPO Policy ($\pi_{1,\theta}$)}
  \label{fig:n50k10_base_pol}
 \end{subfigure}
\vspace{1ex}
 \begin{subfigure}[b]{0.45\textwidth}
  \centering
  \includegraphics[width=\textwidth]{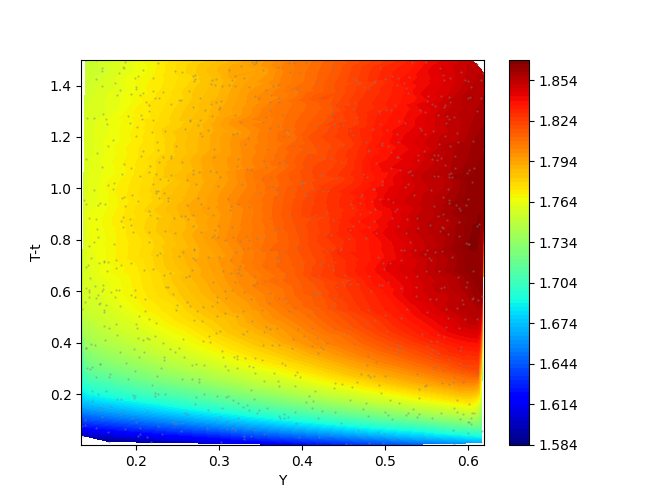}
  \caption{P-PGDPO Policy ($\pi_{1}^{\mathrm{PMP}}$)}
  \label{fig:n50k10_2pg_pol} 
 \end{subfigure}
 \hfill
 \begin{subfigure}[b]{0.45\textwidth}
  \centering
  \includegraphics[width=\textwidth]{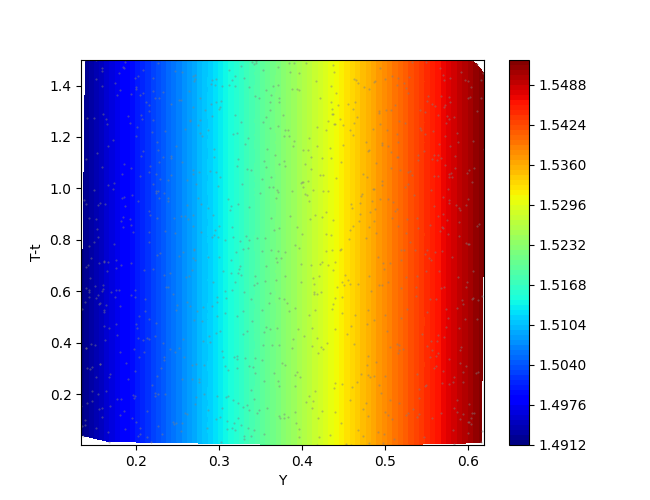}
  \caption{Deep BSDE Policy ($\pi_{1}^{\mathrm{BSDE}}$)}
  \label{fig:n50k10_bsde_pol}
 \end{subfigure}
 \caption{Comparison of total policies for the first asset ($i=1$) across Baseline PG-DPO, Projected PG-DPO (P-PGDPO), and Deep BSDE in the case where $n=50, k=10$. Results shown at (near) optimal iterations for each method (e.g., iteration 6400 for P-PGDPO as plotted), plotted against $Y_1$ at a representative time-to-maturity $T-t$. Other state factors $Y_{j \neq 1}$ are held at their long-term means $\theta_{Y,j}$. Errors use a symmetric logarithmic (symlog) scale. Different color bar ranges reflect varying policy/error magnitudes.}
 \label{fig:n50k10_total_comparison_all_methods}
\end{figure}

In this high-dimensional setting, the baseline PG-DPO produces policies that deviate substantially from the benchmark, resulting in large errors. 
This reflects the difficulty of directly approximating the optimal policy without structural guidance from PMP. 
By contrast, the P-PGDPO policy in the panel (c) of Figure \ref{fig:n50k10_total_comparison_all_methods} is visually indistinguishable from the benchmark, with residuals reduced to near zero except in small regions where hedging effects are most sensitive.\footnote{The visualizations of policy errors for the first asset across models are depicted in Figure \ref{fig:policy_error} in Online Appendix \ref{app:convergence_accuracy}.} 
Its errors are substantially smaller than those of the other methods, consistent with its superior RMSE of \(1.522 \times 10^{-2}\) for case \(n=50, k=10\) as reported in Table~\ref{tab:rmse_vary_n}. 
These results highlight the effectiveness of the P-PGDPO decoupled learning structure: costates are first estimated via BPTT and stabilized through Monte Carlo averaging, and the policy is then analytically constructed using the PMP first-order conditions.

The Deep BSDE results in Figure~\ref{fig:n50k10_total_comparison_all_methods} show a marked improvement over the baseline PG-DPO. 
Its learned policy qualitatively resembles the myopic component of the benchmark in Figure~\ref{fig:n50k10_component_decomp_new}, capturing the dependence on the state variable \(Y_1\). 
However, it fails to recover the more complex intertemporal hedging demands. 
The corresponding error plot exhibits large discrepancies in regions where the hedging component is critical, indicating that this element is largely absent from the Deep BSDE solution.\footnote{See Figure \ref{fig:policy_error} in Online Appendix \ref{app:convergence_accuracy}. This shortcoming reflects both architectural and training limitations. 
The Deep BSDE framework is primarily designed to enforce the terminal condition and maintain consistency between the forward process \(Y_t\) (approximating the value function) and the backward process \(Z_t\) (approximating its gradient). 
Intertemporal hedging, however, depends crucially on the mixed second-order derivative $V_{x\mathbf{y}}$, which is not explicitly supervised in the standard two-network Deep BSDE setup. 
Although the $Z_t$ network can in principle encode this information indirectly, the absence of direct supervision makes $V_{x\mathbf{y}}$ estimation noisy, leading Deep BSDE to underperform on the hedging component despite capturing the myopic part well. 

Moreover, the loss function penalizes deviations in \(Z_t\) through a quadratic martingale residual, which may place relatively little weight on certain gradient directions, particularly those tied to $V_{x\mathbf{y}}$. 
Small errors along these directions can nonetheless be magnified when constructing the control policy, especially through multiplication by the precision matrix \(\Sigma^{-1}\), which inherits the numerical instability of an ill-conditioned covariance matrix.

As a result, while the Deep BSDE method effectively learns the dominant myopic component of the policy, it consistently underperforms in capturing the subtler hedging demands essential for dynamic portfolio optimization.}

\begin{figure}[t!]
 \centering
 \begin{subfigure}[b]{0.45\textwidth}
  \centering
  \includegraphics[width=\textwidth]{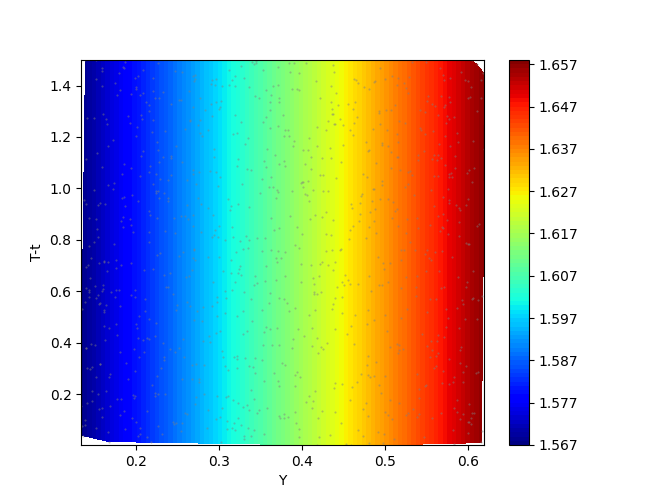}
  \caption{Analytic Myopic ($\pi_{1,0}^{\text{sol}}$)}
  \label{fig:benchmark_policy_myopic_row}
 \end{subfigure}
 \hfill
 \begin{subfigure}[b]{0.45\textwidth}
  \centering
  \includegraphics[width=\textwidth]{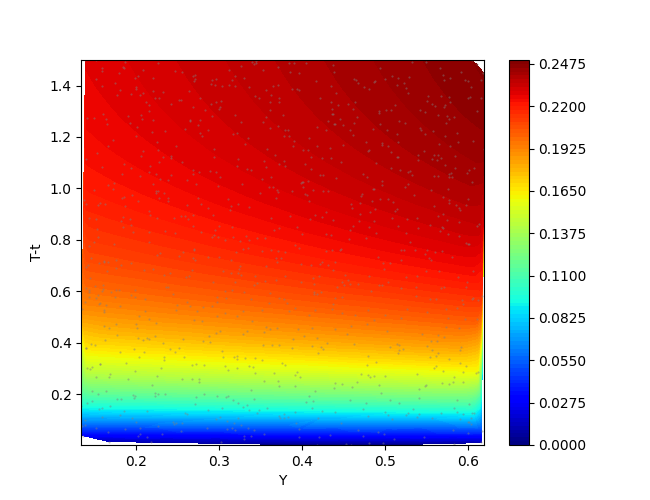}
  \caption{Analytic Hedging ($\pi_{1,1}^{\text{sol}}$)}
  \label{fig:benchmark_policy_hedging_row}
 \end{subfigure}
 \vspace{1ex}
 \begin{subfigure}[b]{0.45\textwidth}
  \centering
  \includegraphics[width=\textwidth]{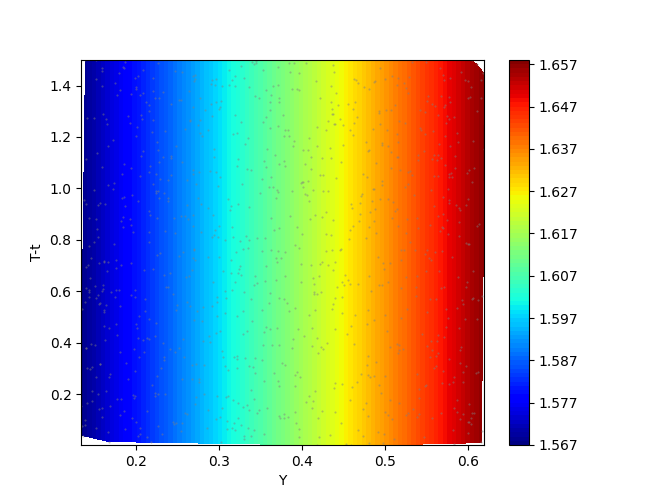}
  \caption{P-PGDPO Myopic ($\pi_{1,0}^{\mathrm{PMP}}$)}
  \label{fig:n50k10_myopic_pol_new}
 \end{subfigure}
 \hfill
 \begin{subfigure}[b]{0.45\textwidth}
  \centering
  \includegraphics[width=\textwidth]{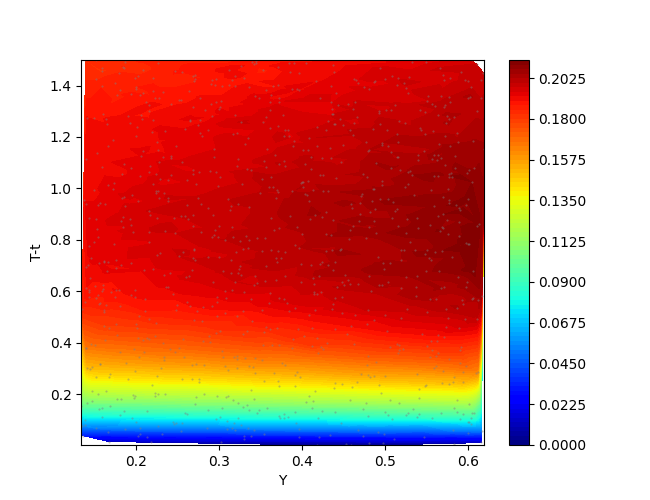}
  \caption{P-PGDPO Hedging ($\pi_{1,1}^{\mathrm{PMP}}$)}
  \label{fig:n50k10_hedging_pol_new}
 \end{subfigure}
 \caption{Decomposition of the Projected PG-DPO (P-PGDPO) policy for the first asset ($i=1$) in the $n=50, k=10$ case (results from iteration 6400 as plotted), plotted against $Y_1$ at a representative time-to-maturity $T-t$. Other state factors $Y_{j \neq 1}$ are held at their long-term means $\theta_{Y,j}$. Top row: analytic myopic and hedging components. Bottom row: myopic and hedging components of P-PGDPO. Note the different color bar ranges.}
 \label{fig:n50k10_component_decomp_new}
\end{figure}

We now examine the myopic and intertemporal hedging demands separately. 
For hedging demand, traditional dynamic programming methods become computationally intractable in high dimensions. 
Although we obtain numerical results for the intertemporal component, they exhibit complex dependencies on both the state variables and the remaining time to maturity.
The hedging ratio---defined as the Euclidean-norm ratio of the hedging demand to the total demand for risky assets---increases sharply with $n$ (Table~\ref{tab:rmse_vary_n}), making accurate estimation of hedging demand even more critical in high-dimensional problems. 
For instance, this percentage rises from 0.1 when $n=1$ to 12.1 when $n=50$. 

Consistent with this observation, Figure~\ref{fig:n50k10_component_decomp_new} decomposes the portfolio for the first asset into its myopic component and intertemporal hedging component, with the bottom row showing the corresponding P-PGDPO policy.
As in Figure~\ref{fig:n50k10_total_comparison_all_methods}, surfaces are plotted against the first state factor $Y_1$ at a representative time to maturity $T-t$, with all the other state factors $Y_{j \neq 1}$ fixed at their long-term means $\theta_{Y,j}$.

The myopic demand recovered by P-PGDPO is virtually identical to the analytic solution in Figure~\ref{fig:benchmark_policy_myopic_row}, with residual errors at the level of numerical precision \(O(10^{-7})\). 
This near-perfect reconstruction demonstrates that the costate-driven projection step in P-PGDPO reliably captures the static mean--variance structure of the optimal policy, even in high-dimensional settings. 

By contrast, the hedging component exhibits a richer nonlinear structure, reflecting portfolio adjustments needed to hedge against future changes in the investment opportunity set. 
Although P-PGDPO reproduces the overall shape and magnitude of this component, the error shows that most of the residual policy error originates here.\footnote{For the visualization of decomposition error, we provide the error map in Figure \ref{fig:decomposition error} in Online Appendix \ref{app:convergence_accuracy}.} 
This reflects the structural difficulty of estimating the cross-derivatives $V_{x\mathbf{y}}$ that drive the hedging term, together with the amplification of gradient noise when multiplied by the numerically unstable matrix $\Sigma^{-1}(\boldsymbol{\sigma}\,\boldsymbol{\rho}\,\boldsymbol{\sigma}_Y)$ in large-$n$ settings.

Indeed, the PMP solution for the hedging term involves the inverse covariance matrix:
\[
\pi^{\text{hedge}} 
  = -\frac{1}{X_t V_{xx}}\,
    \underbrace{\Sigma^{-1}}_{\text{ill-conditioned for large $n$}}\,
    (\boldsymbol{\sigma}\,\boldsymbol{\rho}\,\boldsymbol{\sigma}_Y)\,V_{x\mathbf{Y}}.
\]

Crucially, unlike Deep BSDE—which shows little evidence of capturing the hedging component at all (Figure~\ref{fig:n50k10_total_comparison_all_methods})—P-PGDPO’s two-stage design enables it to recover both components with high fidelity. 
By first stabilizing costate estimates via BPTT and then projecting them analytically through PMP’s first-order conditions, P-PGDPO enforces the structural form of the hedging term, a key factor behind its superior RMSE performance in high-dimensional portfolio optimization. 

In sum, the \(n=50, k=10\) visualizations clearly demonstrate that Projected PG-DPO can “break the dimensional barrier”: where baseline PG-DPO fails and Deep BSDE provides only a partial (myopic-focused) solution, P-PGDPO succeeds. 
The remaining inaccuracies in the hedging component highlight a pressing need to address the ill-conditioning of \(\Sigma\) in very high asset dimensions—potentially through factor structures, regularization, shrinkage techniques, or architecture-aware preconditioning.

\subsection{Long-Horizon Model}\label{sec:long_horizon_robustness}

We now turn to the long-horizon investment problem. 
While the results in the previous sections demonstrate that the PG-DPO framework performs well for shorter horizons ($T=1.5$ years), extending to long horizons poses significant challenges. 
In particular, the baseline PG-DPO algorithm is susceptible to compounding variance and bias over extended simulations. 
The error sources identified in Theorem~\ref{thm:policy_gap}—discretization error ($\propto \Delta t$) and Monte Carlo sampling error ($\propto 1/\sqrt{M}$)—become increasingly critical as the horizon lengthens.

To develop a robust framework capable of addressing long-horizon challenges, we introduce three key algorithmic enhancements aimed at improving both accuracy and stability. 

First, we adopt a residual learning scheme for hedging demand. 
The policy network is restructured to learn only the complex intertemporal hedging component, while the analytically known myopic demand is computed separately (see \citet{silver2018residual,johannink2019residual}). 
This separation leverages domain knowledge from financial economics: the myopic component is well understood, whereas accurate estimation of the hedging demand remains the primary challenge. 
Residual learning thus allows the network to dedicate its full capacity to the most difficult part of the problem. 

Second, we incorporate a control variate to suppress the variance inherent in long Monte Carlo rollouts, the most critical issue in long-term simulations. 
A key advantage of our model-based approach is the ability to seamlessly integrate classical variance-reduction techniques from the financial engineering literature \citep{glasserman2004monte}. 
By using the analytically tractable myopic policy as a control variate, we can significantly stabilize the learning process. 

Third, we employ Richardson extrapolation to reduce time-discretization error. 
The known structure of the underlying SDEs in our model-based setting permits the straightforward application of this classical numerical method. 
By combining simulations at time steps $\Delta t$ and $0.5\Delta t$, we achieve higher-order accuracy, which is crucial for maintaining fidelity over long horizons \citep{glasserman2004monte}.

Before applying these techniques to the long-horizon problem, we first validate their impact on the original experimental setup with $T=1.5$ years. 
This step ensures that the enhancements deliver genuine performance improvements. 
As shown in Table~\ref{tab:short_term_gains}, the new techniques yield a dramatic reduction in policy RMSE for both Baseline PG-DPO and P-PGDPO. 
Performance gains of up to $57.5\times$ confirm that these enhancements are not incremental tweaks but constitute a fundamental upgrade to the algorithm. 
This successful validation provides a strong foundation for tackling the more complex long-horizon investment problem.

\begin{table}[t!]
    \centering
    \caption{RMSE Improvement with Advanced Techniques for $T=1.5$ Horizon}
    \label{tab:short_term_gains}
    \begin{tabular}{cccccc}
        \toprule
        n & k & \makecell{New Baseline \\ RMSE} & \makecell{Baseline \\ Improvement} & \makecell{New P-PGDPO \\ RMSE} & \makecell{P-PGDPO \\ Improvement} \\
        \midrule
        1 & 1 & $5.10 \times 10^{-3}$ & $6.4 \times$ & $4.80 \times 10^{-5}$ & $23.3 \times$ \\
        1 & 2 & $8.29 \times 10^{-4}$ & $57.5 \times$ & $6.30 \times 10^{-5}$ & $3.16 \times$ \\
        1 & 5 & $1.04 \times 10^{-3}$ & $35.7 \times$ & $1.00 \times 10^{-5}$ & $25.0 \times$ \\
        1 & 10 & $5.54 \times 10^{-4}$ & $31.6 \times$ & $2.23 \times 10^{-4}$ & $2.11 \times$ \\
        10 & 2 & $1.49 \times 10^{-2}$ & $7.7 \times$ & $8.33 \times 10^{-4}$ & $4.38 \times$ \\
        10 & 10 & $5.07 \times 10^{-3}$ & $9.9 \times$ & $2.21 \times 10^{-3}$ & $1.27 \times$ \\
        50 & 2 & $8.07 \times 10^{-2}$ & $4.8 \times$ & $8.40 \times 10^{-3}$ & $8.73 \times$ \\
        50 & 10 & $3.38 \times 10^{-2}$ & $3.2 \times$ & $1.07 \times 10^{-2}$ & $1.42 \times$ \\
        \bottomrule
    \end{tabular}
    \vspace{1em}
    \parbox{\textwidth}{\small \textit{Note:} This table compares the minimum policy RMSE of the enhanced algorithms against the original results from Tables \ref{tab:rmse_n1_vary_k} and \ref{tab:rmse_vary_n}. The `Improvement' column indicates the factor by which the RMSE declines (e.g., Original RMSE / New RMSE).}
\end{table}

With a demonstrably more robust algorithm, we now address the long-horizon investment problem with a maturity of $T=20$ years. 
As the horizon extends toward infinity, the optimal policy is expected to become stationary, i.e., independent of time-to-maturity. 
Although such a stationary policy is structurally simpler for the network to learn, the extended simulations required to reach it amplify the challenges of variance and error control. 
Our enhanced framework is designed precisely for this scenario. 

\begin{table}[b!]
    \centering
    \caption{Policy Errors for Long‑Horizon Problem ($T=20$)}
    \label{tab:long_term_results}
    \begin{tabular}{cc cc cc}
        \toprule
        & & \multicolumn{2}{c}{Baseline PG-DPO} & \multicolumn{2}{c}{P-PGDPO} \\
        \cmidrule(lr){3-4} \cmidrule(lr){5-6}
        $n$ & $k$
            & \makecell{  RMSE}
            & \makecell{Iterations \\ at Min.}
            & \makecell{ RMSE}
            & \makecell{Iterations \\ at Min.} \\
        \midrule
         1 &  1  & 0.00169 &  400 & 0.000041 &  400 \\
         5 &  3  & 0.00876 &  200 & 0.000352 &    1 \\
        30 &  5  & 0.07220 & 1000 & 0.007030 &  200 \\
        50 & 10  & 0.05020 & 2600 & 0.003850 &  200 \\
        \bottomrule
    \end{tabular}
\end{table}

Table~\ref{tab:long_term_results} reports the final policy RMSE for the 20-year experiment. 
The results show remarkably low errors even in high dimensions: in the $n=50, k=10$ case, P-PGDPO achieves an error of just $3.85 \times 10^{-3}$. 
This error is even lower than in the 1.5-year experiment, a phenomenon attributable to the synergy between a more powerful algorithm and the convergence toward a simpler stationary target policy.


\begin{figure}[t!]
    \centering
    \hfill
    \begin{subfigure}[b]{0.43\textwidth}
        \centering
        \includegraphics[width=\textwidth]{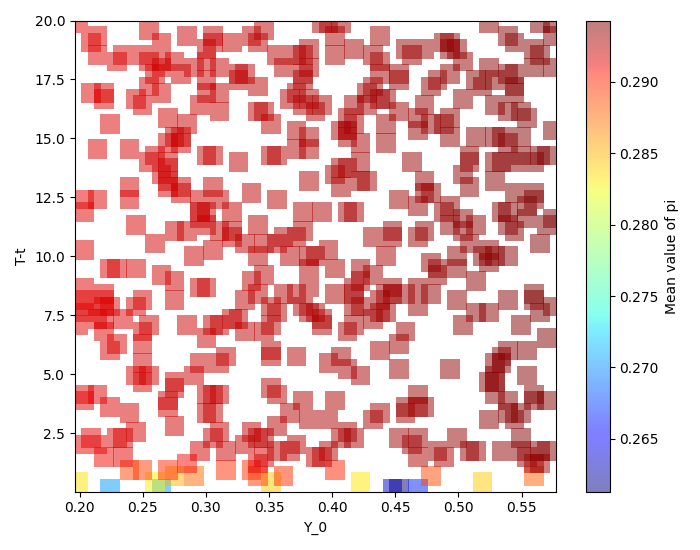}
        \caption{Analytic Solution (Asset 1)}
        \label{fig:long_term_sol_asset1}
    \end{subfigure}
    \hfill
    \begin{subfigure}[b]{0.43\textwidth}
        \centering
        \includegraphics[width=\textwidth]{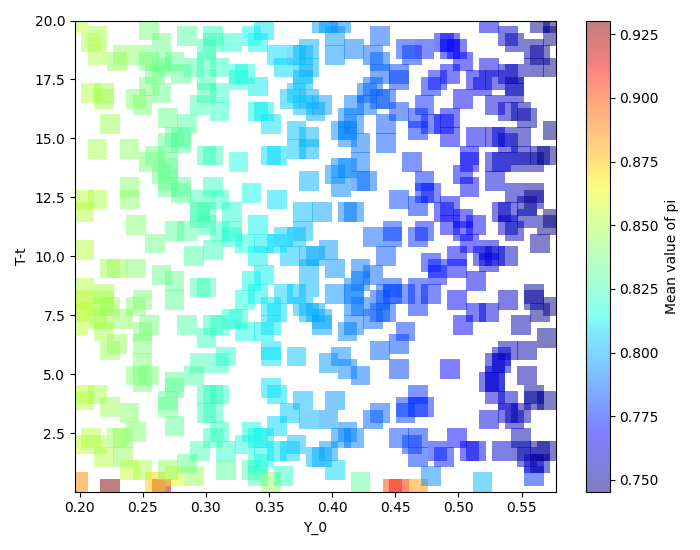}
        \caption{Analytic Solution (Asset 2)}
        \label{fig:long_term_sol_asset2}
    \end{subfigure}
\hfill \\
    \hfill
    \begin{subfigure}[b]{0.43\textwidth}
        \centering
        \includegraphics[width=\textwidth]{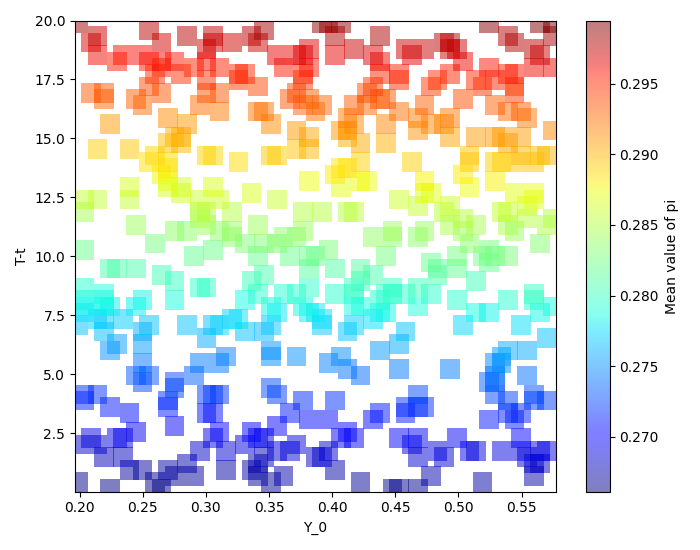}
        \caption{Baseline PG-DPO (Asset 1)}
        \label{fig:long_term_base_asset1}
    \end{subfigure}
    \hfill
    \begin{subfigure}[b]{0.43\textwidth}
        \centering
        \includegraphics[width=\textwidth]{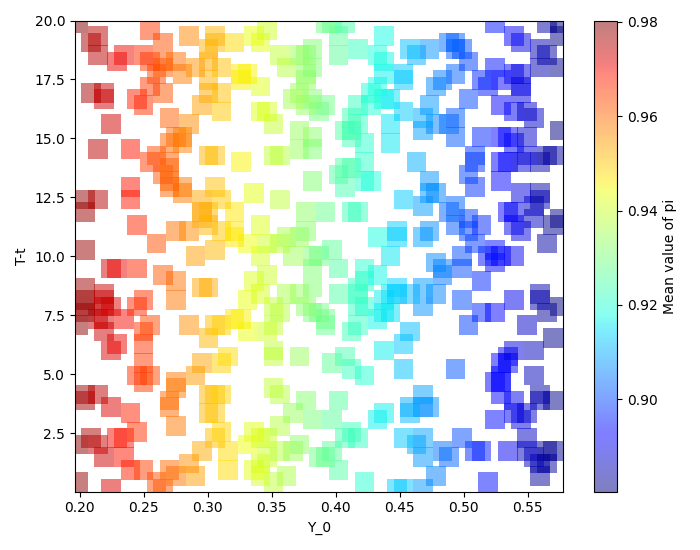}
        \caption{Baseline PG-DPO (Asset 2)}
        \label{fig:long_term_base_asset2}
    \end{subfigure}
    \hfill\\
    
    \hfill
    \begin{subfigure}[b]{0.43\textwidth}
        \centering
        \includegraphics[width=\textwidth]{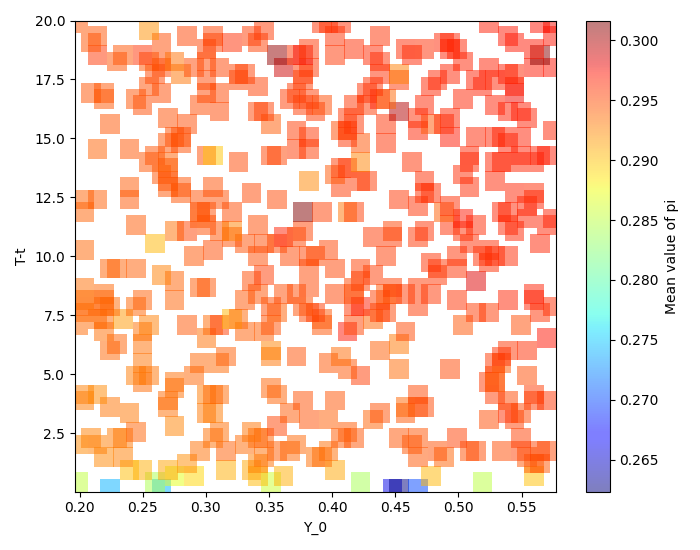}
        \caption{P-PGDPO (Asset 1)}
        \label{fig:long_term_ppgdpo_asset1}
    \end{subfigure}
    \hfill
    \begin{subfigure}[b]{0.43\textwidth}
        \centering
        \includegraphics[width=\textwidth]{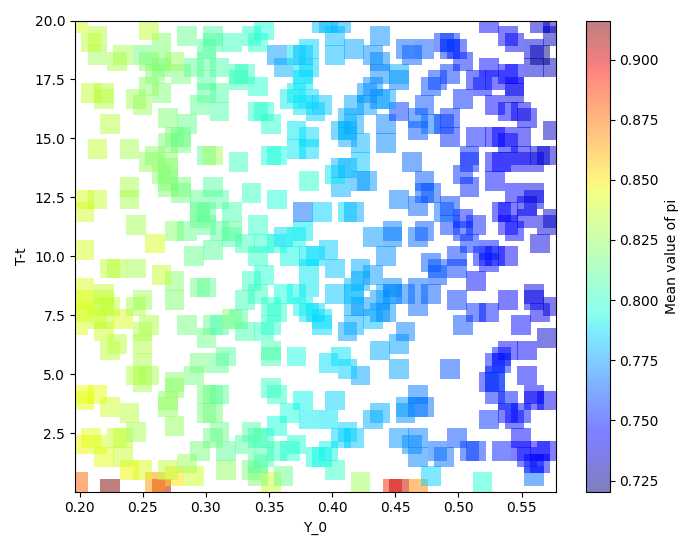}
        \caption{P-PGDPO (Asset 2)}
        \label{fig:long_term_ppgdpo_asset2}
    \end{subfigure}
    \hfill
    \caption{Comparison of learned policies for the 20-year horizon ($n=50, k=10$), plotted against the first state factor $Y_1$. Left column: Asset 1. Right column: Asset 2. The P-PGDPO results show a near-perfect match with the benchmark solution, whereas the Baseline PG-DPO fails to capture the correct policy structure.}
    \label{fig:long_term_plots}
\end{figure}

As shown in Figure~\ref{fig:long_term_plots}, the policies generated by P-PGDPO are visually indistinguishable from the benchmark, accurately capturing the intricate nonlinear dependencies on the state variables. 
By contrast, the baseline PG-DPO, even with enhanced techniques, fails to reproduce the correct structural form of the policy. 
This comparison illustrates the qualitative difference between a rough approximation—potentially yielding a low but misleading error metric—and a truly accurate approximation that captures both the value and the structural complexity of the optimal policy. 
These findings underscore the strength of the enhanced P-PGDPO framework in delivering robust, high-fidelity solutions to challenging long-horizon control problems.

\subsection{Non-affine Asset Return Dynamics}
\label{sec:non_affine_robustness}

We now assess the flexibility of our framework under non-affine model dynamics, as introduced in Section~\ref{sec:non-affine models}. 
In this specification, the $k$-dimensional OU state process is retained, but the expected returns of the risky assets include quadratic terms in the state variables:
\[
  \boldsymbol{\mu}_S(t, \mathbf{Y}_t) 
    = \text{diag}(\boldsymbol{\sigma}) \Big( \boldsymbol{\alpha} \big(\mathbf{Y}_t + \boldsymbol{\beta} \mathbf{Y}_t^2 \big) \Big) + r\mathbf{1},
\]
where the vector $\boldsymbol{\beta}$ governs the strength and direction of the nonlinearity, with overall magnitude measured by $\beta_{\text{norm}} = \|\boldsymbol{\beta}\|_2$. 
When $\beta_{\text{norm}} = 0$, the model collapses exactly to the affine setting of Section~\ref{sec:multiasset_multifactor_KO}, for which analytic solutions are available. 
As $\beta_{\text{norm}}$ increases, progressively stronger nonlinear effects are introduced, providing a controlled way to assess robustness in settings without an analytic benchmark.

To ensure a clear and direct comparison, we return to the $T=1.5$ year horizon and the foundational P-PGDPO algorithm. 
This choice allows us to attribute any changes in the optimal policy solely to the introduction of nonlinearity, isolating its effects from the advanced techniques introduced in the long-horizon study. 
All experiments in this subsection are conducted in the most challenging setting, with \(n=50\) assets and \(k=10\) state factors, using the P-PGDPO method. 

In the non-affine setting, the true optimal policy is unknown. 
To measure accuracy, we first identify in the affine case ($\beta_{\text{norm}}=0$) the iteration count that minimizes RMSE against the analytic benchmark solution from Section~\ref{sec:multi_asset_numerical}; we refer to this as the \emph{benchmark iteration}. 
For each $\beta_{\text{norm}}>0$, we then train the network for 10,000 iterations under the non-affine dynamics and compute the RMSE between the resulting policy and the affine benchmark policy at the benchmark iteration. 
This procedure is best interpreted as a deviation measure: it captures the policy changes induced by the nonlinear term relative to the affine baseline, rather than serving as an absolute accuracy metric. 
In this way, deviations reflect only the nonlinear structure, not stochastic optimization noise or premature convergence.

This approach differs from the reverse-engineering procedure of \citet{duarte2024machine}, where the environment is artificially constructed so that the optimal policies are known by design. 
Such a setup is useful for controlled evaluation of numerical methods, but its performance assessment is relative to the imposed functional form rather than to the optimal policy of an economically grounded model. 
By contrast, the affine-benchmark RMSE comparison employed here draws its reference from the same economic model in its limiting affine case, ensuring that the comparison reflects meaningful structural changes in the original portfolio optimization problem. 
Moreover, it allows us to verify continuity of the solution by confirming that the learned policy converges smoothly to the affine benchmark as $\beta_{\text{norm}}\!\to\!0$, and to attribute deviations at large $\beta_{\text{norm}}$ values directly to the nonlinearities in the model.


Table~\ref{tab:non_affine_rmse} reports the RMSE between the P-PGDPO policy learned under non-affine dynamics and the affine benchmark policy ($\beta_{\text{norm}}=0$), across varying degrees of nonlinearity. 
When $\beta_{\text{norm}}=0$, P-PGDPO achieves a low RMSE ($4.70\times 10^{-1}$), consistent with the high precision observed in the affine experiments of Section~\ref{sec:multi_asset_numerical}. 
As $\beta_{\text{norm}}$ increases, the RMSE grows monotonically, reflecting a systematic and stable adjustment of the learned policy to the nonlinear dynamics.

\begin{table}[h!]
    \centering
    \caption{Policy errors of P-PGDPO for the model with non-affine risk premia vs. affine benchmark ($n=50, k=10$). Results reported after 10,000 training iterations.}
    \label{tab:non_affine_rmse}
    \begin{tabular}{@{}cccccccc@{}}
    \toprule
    $\beta_{\text{norm}}$ & 0.0 & 0.1 & 0.5 & 1.0 & 2.0 & 3.0 & 4.0 \\
    \midrule
    RMSE (P-PGDPO) & 0.04701 & 0.05686 & 0.06755 & 0.11469 & 0.20120 & 0.27530 & 0.33570 \\
    \bottomrule
\end{tabular}
\end{table}

To complement the RMSE-based results, Figure~\ref{fig:nonaffine_robust_comparison} provides a visual comparison of the learned policies. 
The figure plots the policy for the first asset $\pi_1$ as a function of $Y_1$ and time-to-maturity $T-t$ for $\beta_{\text{norm}} \in \{0, 0.1, 1.0, 2.0\}$, with all other factors fixed at their long-run means $\theta_{Y,j}$. 
When $\beta_{\text{norm}}=0$, the learned surface matches the affine benchmark almost exactly, confirming that P-PGDPO reproduces the known solution in the absence of nonlinearity. 
As $\beta_{\text{norm}}$ increases, the surfaces deviate smoothly from the affine shape, exhibiting curvature and level shifts consistent with the quadratic state effects in $\mu(t, Y_t)$. 
Even at $\beta_{\text{norm}}=2.0$, the learned policy remains smooth and free of spurious oscillations, indicating that the algorithm preserves numerical stability under substantial nonlinearity. 
This visual progression corroborates the quantitative results in Table~\ref{tab:non_affine_rmse}, showing that the P-PGDPO framework consistently captures the effect of nonlinear terms on the risk premium.

\begin{figure}[t!]
    \centering 
    \hfill
    \begin{subfigure}[b]{0.45\textwidth}
        \includegraphics[width=\textwidth]{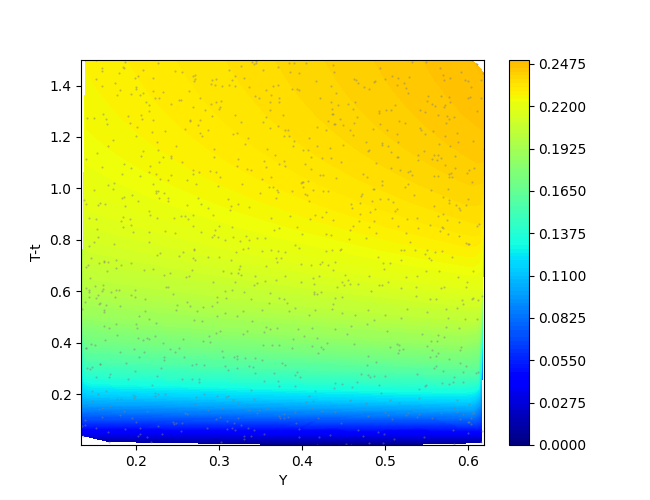}
        \caption{Analytical solution ($\beta_{\text{norm}}=0$).}
        \label{fig:nonaffine_robust_sol}
    \end{subfigure}
    \hfill
    \begin{subfigure}[b]{0.45\textwidth}
        \includegraphics[width=\textwidth]{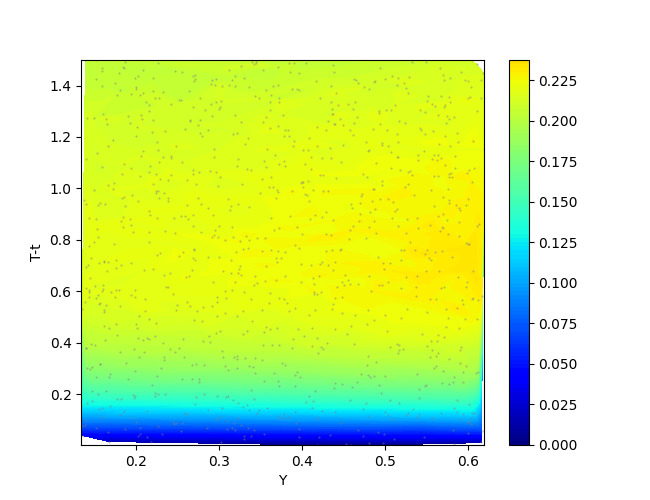}
        \caption{P-PGDPO policy ($\beta_{\text{norm}}=0.0$).}
        \label{fig:nonaffine_robust_b0_0}
    \end{subfigure}
    \hfill\\

    \hfill
    \begin{subfigure}[b]{0.45\textwidth}
        \includegraphics[width=\textwidth]{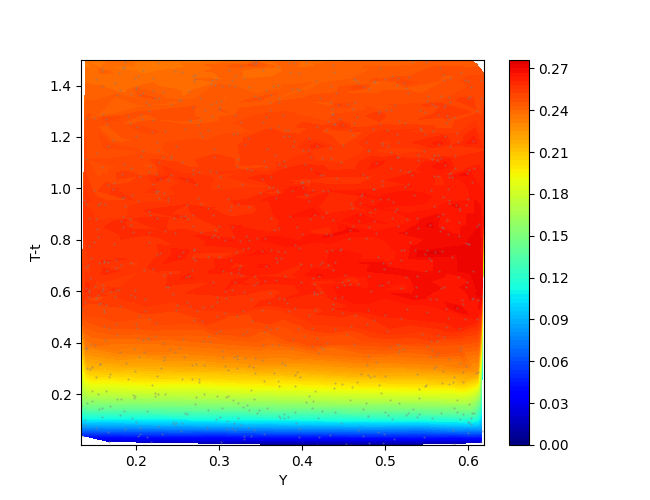}
        \caption{P-PGDPO policy ($\beta_{\text{norm}}=1.0$).}
        \label{fig:nonaffine_robust_b1_0}
    \end{subfigure}
    \hfill
    \begin{subfigure}[b]{0.45\textwidth}
        \includegraphics[width=\textwidth]{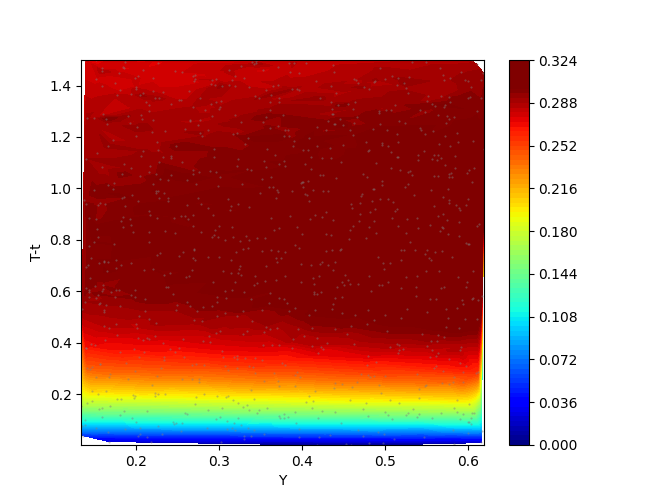}
        \caption{P-PGDPO policy ($\beta_{\text{norm}}=2.0$).}
        \label{fig:nonaffine_robust_b2_0}
    \end{subfigure}
    \hfill
    \caption{Visual comparison of optimal policies for the first asset ($\pi_1$) against the first state factor ($Y_1$) and time-to-maturity ($T-t$) in the $n=50, k=10$ setting. (a) shows the analytical solution for the affine model ($\beta_{\text{norm}}=0$). (b)-(e) display P-PGDPO learned policies for models with increasing strength of nonlinearity in asset risk premia, corresponding to $\beta_{\text{norm}}$ values of $0.0, 0.1, 1.0, 2.0$, respectively. Note the systematic change in policy structure and value ranges (indicated by colorbars) as $\beta_{\text{norm}}$ increases.}
    \label{fig:nonaffine_robust_comparison}
\end{figure}


The quadratic terms affect both the myopic and intertemporal hedging components of the optimal policy. 
Myopic demand shifts directly in response to the altered instantaneous risk premia $\sigma A(Y_t+\beta Y_t^2)$, while intertemporal hedging adjusts more subtly as costate gradients $\partial_Y \lambda_t^*$ respond to the nonlinear state dynamics. 
For small $\beta_{\text{norm}}$, these adjustments are modest, and the policy remains close to the affine benchmark. 
For larger $\beta_{\text{norm}}$, the deviations become more pronounced, reflecting genuine changes in optimal hedging behavior rather than artifacts of the numerical method. 

Overall, P-PGDPO exhibits high accuracy in the affine limit and adapts stably under increasing nonlinearity. 
Its smooth convergence to the affine benchmark as $\beta_{\text{norm}} \to 0$ validates the implementation, while systematic deviations at larger $\beta_{\text{norm}}$ highlight the robustness and interpretability of the method in settings without analytical solutions. 

\section{Conclusion}
\label{sec:conclusion}

We have proposed \emph{Pontryagin-Guided Direct Policy Optimization (PG-DPO)} for high-dimensional finite-horizon continuous-time consumption and portfolio choice. 
By integrating Pontryagin’s Maximum Principle (PMP) with backpropagation through time (BPTT), PG-DPO circumvents the curse of dimensionality that hampers DP- and PDE-based methods, directly learning policies that recover both myopic and intertemporal hedging components. 

A central innovation is the \emph{Projected PG-DPO (P-PGDPO)} algorithm. 
By exploiting the rapid stabilization of costate estimates, P-PGDPO projects them onto PMP’s first-order conditions, yielding policies that are both computationally efficient and structurally consistent. 
This design delivers substantial gains in accuracy and scalability. 

Numerical results confirm these strengths. 
Deep BSDE captures myopic demands but fails to recover hedging terms. 
In contrast, P-PGDPO achieves a policy RMSE of $1.52 \times 10^{-2}$ in a $n=50,k=10$ setting, with negligible error in the myopic component and accurate reconstruction of complex hedging demands. 
Theoretical guarantees, including our Policy Gap Bound, further establish near-optimality when costate estimates are accurate and FOC residuals are small. 
Finite-horizon problems fit naturally into this framework, with the long-horizon case emerging as a special case. 

The practical implications are significant. 
P-PGDPO’s ability to recover nonlinear hedging strategies in high dimensions makes it well-suited for institutional portfolios with many assets, ETFs, or derivatives. 
The principle—\emph{estimate costates first, then derive policy}—offers a blueprint for broader applications where PMP provides structure. 
Indeed, companion work applies it successfully to constrained portfolio optimization \citep{HJKL}. 

Future research should pursue: (i) regularization techniques for ill-conditioned covariance matrices; (ii) extensions to frictions such as transaction costs and taxes; (iii) richer preference specifications, including behavioral or sustainability criteria; and (iv) empirical validation through historical and simulated backtests. 

In sum, PG-DPO—and especially P-PGDPO—demonstrates how the synergy of deep learning and PMP’s structural guidance can expand the frontier of solvability in high-dimensional continuous-time financial control. 
It bridges optimal control and modern machine learning, providing both theoretical insight and a foundation for realistic, large-scale decision problems.

\section*{Acknowledgments}

This work was supported by the National Research Foundation of Korea (NRF) grant
funded by the Korea government (MSIT) (RS-2025-00562904).

\bibliographystyle{apalike}
\bibliography{gf_bib}

\newpage

\section*{Online Appendix for ``Breaking the Dimensional Barrier: A Pontryagin-Guided Direct Policy Optimization for Continuous-Time Multi-Asset Portfolio Choice"}

\begin{center}
    Jeonggyu Huh (Sungkyunkwan Univ.), Jaegi Jeon (Chonnam National Univ.), Hyeng Keun Koo (Ajou Univ.), and Byung Hwa Lim (Sungkyunkwan Univ.)
\end{center}

\appendix

%


\section{Regularity Conditions and Assumptions for a General Multi-Asset Model}\label{sec:regularity conditions}

We impose regularity conditions on the market parameters of the general multi-asset model in Section \ref{sec:market_model_with_state}.

\begin{assumption}[Market-parameter regularity]\label{ass:coeff}
Fix a finite horizon \( T > 0 \). The coefficients defining the state process \( \mathbf{Y}_t \) and influencing the wealth process \( X_t \) via Eq.~\eqref{eq:wealth}, namely
\(\boldsymbol{\mu}_{Y},\ \boldsymbol{\sigma}_{Y},\ \boldsymbol{\mu},\ \boldsymbol{\sigma},\ r\),
satisfy the following conditions:

\begin{enumerate}[label=\textup{(\alph*)}, leftmargin=3.2em]
\item \textbf{Measurability and continuity.}
    Each coefficient is jointly Borel-measurable in \( (t, \mathbf{y}) \) and continuous in \( \mathbf{y} \) for every fixed \( t \in [0,T] \).

\item \textbf{Global Lipschitz continuity.}
    There exists a constant \( K_{\mathrm{Lip}} > 0 \) such that for all
    \( t \in [0,T] \) and \( \mathbf{y}, \mathbf{y}' \in \mathbb{R}^k \),
    \begin{align*}
      &\lVert \boldsymbol{\mu}_{Y}(t,\mathbf{y}) - \boldsymbol{\mu}_{Y}(t,\mathbf{y}') \rVert
      + \lVert \boldsymbol{\sigma}_{Y}(t,\mathbf{y}) - \boldsymbol{\sigma}_{Y}(t,\mathbf{y}') \rVert \\
      &\quad + \lVert \boldsymbol{\mu}(t,\mathbf{y}) - \boldsymbol{\mu}(t,\mathbf{y}') \rVert
      + \lVert \boldsymbol{\sigma}(t,\mathbf{y}) - \boldsymbol{\sigma}(t,\mathbf{y}') \rVert
      \le K_{\mathrm{Lip}} \lVert \mathbf{y} - \mathbf{y}' \rVert.
    \end{align*}

\item \textbf{Linear growth.}
    There exists a constant \( K_{\mathrm{LG}} > 0 \) such that for all
    \( (t,\mathbf{y}) \in [0,T] \times \mathbb{R}^k \),
    \[
    \begin{aligned}
    &\lVert \boldsymbol{\mu}_{Y}(t,\mathbf{y}) \rVert^2
    + \lVert \boldsymbol{\sigma}_{Y}(t,\mathbf{y}) \rVert^2
    + \lVert \boldsymbol{\mu}(t,\mathbf{y}) \rVert^2
    + \lVert \boldsymbol{\sigma}(t,\mathbf{y}) \rVert^2 \\
    &\qquad \le K_{\mathrm{LG}}^2 \left( 1 + \lVert \mathbf{y} \rVert^2 \right), \qquad
    |r(t,\mathbf{y})| \le K_{\mathrm{LG}} \left( 1 + \lVert \mathbf{y} \rVert \right).
    \end{aligned}
    \]
\end{enumerate}
\end{assumption}

\noindent
Under Assumption~\ref{ass:coeff}, the SDE system for \( (\mathbf{Y}_t, X_t) \) (where \( X_t \) depends on the control processes \( \boldsymbol{\pi}_t \) and \( C_t \)) admits a unique strong solution for sufficiently well-behaved controls. For the existence and uniqueness of the solutions to SDEs, see chapters 5 and 6 in \cite{oksendal2013stochastic}.

Now we define the admissible set $\mathcal{A}( x,\mathbf{y})$ for our stochastic control problem. 

\begin{definition}
    [Admissible set]\label{ass:adm}
A pair \( (\boldsymbol{\pi}, C) \) is admissible if the following conditions hold:
\begin{enumerate}[label=\textup{(\alph*)},leftmargin=3.2em]
  \item \textbf{Adaptedness.}
        The process \( (\boldsymbol{\pi}_t, C_t)_{0 \le t \le T} \) is \( \{\mathcal{F}_t\} \)-progressively
        measurable, where \( \{\mathcal{F}_t\} \) is the filtration generated by
        the underlying Brownian motions.

  \item \textbf{Integrability.}
        \[
        \mathbb{E}\!\int_0^T
        \left( \lVert \boldsymbol{\pi}_t \rVert^{2} + C_t \right)\,dt \;<\; \infty.
        \]

  \item \textbf{Interior bounds.}
        There exist constants \( \Pi_{\max} > 0 \) and \( \xi \in (0,1) \) such that
        \[
        \lVert \boldsymbol{\pi}_t \rVert \;<\; \Pi_{\max}, \qquad
        0 \;\le\; C_t \;<\; \xi\,X_t,
        \qquad \forall t \in [0,T] \text{ a.s.}
        \]

  \item \textbf{Doadmissible wealth.}
        Let \( x_{\min} := -a/b \) (as defined for the HARA utility). Then
        \[
        X_t \;>\; x_{\min}, \qquad \forall t \in [0,T] \text{ a.s.}
        \]
\end{enumerate}
\end{definition}

\begin{assumption}[Risk covariance]\label{ass:cov}
The block covariance matrix
\[
\boldsymbol{\Omega} =
\begin{psmallmatrix}
\Psi & \rho \\[.2em] \rho^{\top} & \Phi
\end{psmallmatrix}
\]
is constant in time and uniformly positive definite: there exists \( \lambda_{\Omega} > 0 \) such that
\[
z^{\top} \boldsymbol{\Omega} z \ge \lambda_{\Omega} \lVert z \rVert^2, \qquad \forall z \in \mathbb{R}^{n+k}.
\]
As a consequence, both \( \Psi \) and \( \Phi \) are individually positive definite and thus invertible.
\end{assumption}

\noindent
The definition of an admissible set is standard in continuous-time portfolio choice models. Specifically, the doadmissible wealth condition in (d) ensures that the wealth process remains within the domain of the utility function, i.e., \( X_t > x_{\min} = -a/b \) almost surely, and also rules out arbitrage in this continuous-time setting \citep{dybvig1988nonnegative}.

Since \( U \) is \( C^2 \), strictly concave, and satisfies the Inada conditions on \( (x_{\min}, \infty) \), and because \( \delta > 0 \), we obtain the following integrability bounds:
\[
\mathbb{E}\left[e^{-\delta t}\,U(C_t)\right]
\;\le\;
e^{-\delta t}\,U\bigl(\xi\,\mathbb{E}[X_t]\bigr)
\;<\;\infty, \qquad
\mathbb{E}\left[U(X_T)\right] < \infty.
\]
Therefore, the value function $V(0,x,\mathbf{y})$ is well-defined for every admissible control pair \( (\boldsymbol{\pi}, C) \). We set all admissible pairs by the admissible set $\mathcal{A}(x,y)$.

\begin{assumption}[Uniform non-degeneracy]\label{ass:nondeg}

Let \( a(t,\mathbf{y}) := \boldsymbol{\sigma}(t,\mathbf{y})\, \Psi\, \boldsymbol{\sigma}(t,\mathbf{y})^{\top} \). There exists a constant \( \lambda > 0 \) such that
      \[
        \mathbf{z}^{\top} a(t,\mathbf{y})\, \mathbf{z}
        \;\ge\; \lambda\, \lVert \mathbf{z} \rVert^{2},
        \qquad
        \forall (t,\mathbf{y},\mathbf{z}) \in [0,T] \times \mathbb{R}^{k} \times \mathbb{R}^{n}.
      \]
\end{assumption}

\noindent
Assumption~\ref{ass:nondeg} ensures the diffusion matrix in the wealth dynamics is uniformly positive definite, which guarantees parabolic regularity of the associated Hamilton–Jacobi–Bellman (HJB) equation. 
We continue to work with the HARA utility function \( U \) defined in Eq.~\eqref{eq:hara_utility}, which is \( C^2 \), strictly increasing and concave on its domain \( (x_{\min}, \infty) \), and satisfies the Inada conditions. 
Together with the coefficient regularity in Assumption~\ref{ass:coeff}, these conditions yield the existence, uniqueness, and sufficient smoothness of the value function (e.g., \( V \in C^{1,2,2} \)); see \citet[Chapter 6 and 7]{oksendal2013stochastic}.

To ensure the well-posedness of the forward–backward stochastic differential equation (FBSDE) system and the applicability of Pontryagin’s Maximum Principle (PMP), additional smoothness conditions on the model coefficients are typically required.

\begin{assumption}[Additional regularity for PMP]\label{ass:pmp_reg}
In addition to Assumptions~\ref{ass:coeff}, \ref{ass:nondeg}, and an addmissible set in Definition \ref{ass:adm}, we assume:
\begin{enumerate}[label=\textup{(\alph*)},leftmargin=3.2em]
\item \textbf{\(C^1\)-regularity with respect to \(\mathbf{y}\).}
All coefficient functions—\( \boldsymbol{\mu}_{Y},\, \boldsymbol{\sigma}_{Y},\, \boldsymbol{\mu},\, \boldsymbol{\sigma},\, r \)—are continuously differentiable with respect to the state vector \( \mathbf{y} \). Moreover, their partial derivatives \( \partial_{\mathbf{y}} f \) are globally Lipschitz continuous in \( \mathbf{y} \) and satisfy linear growth bounds. These conditions ensure that the drift terms in the costate BSDEs, including expressions such as \( \mathcal{H}_{\mathbf{y}} \), satisfy the regularity requirements for the existence and uniqueness of solutions; see \citet[Chapter 2 and 6]{yong2012stochastic}.
\end{enumerate}
\end{assumption}

\section{Derivation of the Adjoint Diffusion Coefficient via It\^o's Lemma} 
\label{sec:ito_derivation}

We derive the diffusion coefficients \( \mathbf{Z}^{\lambda X} \) in Section \ref{sec:PMP_multiasset_state_costate}. In vector form, the state process satisfies
\[
d\mathbf{Y}_t = \boldsymbol{\mu}_Y\bigl(t,\mathbf{Y}_t\bigr)\,dt + \boldsymbol{\sigma}_Y\bigl(t,\mathbf{Y}_t\bigr)\,d\mathbf{W}^Y_t,
\]
where \( \boldsymbol{\mu}_Y(t,\mathbf{Y}_t) \in \mathbb{R}^k \) is the drift vector, \( \boldsymbol{\sigma}_Y(t,\mathbf{Y}_t) \in \mathbb{R}^{k \times k} \) is a diagonal matrix of diffusion coefficients. i.e.,
\[
\boldsymbol{\sigma}_Y(t,\mathbf{Y}_t) = \operatorname{diag}(\sigma_{Y,1}(t,\mathbf{Y}_t), \dots, \sigma_{Y,k}(t,\mathbf{Y}_t)),
\]
and \( \mathbf{W}^Y_t \in \mathbb{R}^k \) is a vector of correlated Brownian motions.

Moreover, let \( \boldsymbol{\mu}(t,\mathbf{Y}_t) \in \mathbb{R}^n \) denote the vector of expected returns, \( \boldsymbol{\sigma}(t,\mathbf{Y}_t) \in \mathbb{R}^{n \times n} \) the diagonal matrix of volatilities, and \( \mathbf{W}^X_t \in \mathbb{R}^n \) the Brownian motion vector for risky assets. Then the dynamics of the vector of risky asset prices \( \mathbf{S}_t \in \mathbb{R}^n \) can be compactly written as
\[
d\mathbf{S}_t = \operatorname{diag}(\mathbf{S}_t)\,\boldsymbol{\mu}(t,\mathbf{Y}_t)\,dt + \operatorname{diag}(\mathbf{S}_t)\,\boldsymbol{\sigma}(t,\mathbf{Y}_t)\,d\mathbf{W}^X_t.
\]

In the Pontryagin framework, the adjoint (or costate) process associated with the investor's wealth \(X_t\) is identified with the partial derivative of the value function \(V\). Concretely, we set \(\lambda_t = V_X\!\bigl(t,\,X_t,\,\mathbf{Y}_t\bigr),\) where \(\mathbf{Y}_t \in \mathbb{R}^k\) represents additional state variables and \(X_t\) is the investor's scalar wealth process. Our goal in this section is to derive the diffusion coefficient corresponding to the \(d\mathbf{W}^X_t\) term in the backward SDE for \(\lambda_t\), namely \(\mathbf{Z}_t^{\lambda X}\) from BSDE \eqref{eq:lambda_BSDE}, which plays a crucial role in determining the optimal portfolio via FOC \eqref{eq:foc_pi}. 

Recall that under an optimal control \(\bigl(\boldsymbol{\pi}_t^*, C_t^*\bigr)\), the wealth \(X_t^*\) evolves according to Eq.~\eqref{eq:X_forward_optimal} and Eq.~\eqref{eq:Y_forward}, and the \(k\)-dimensional state process \(\mathbf{Y}_t^*\) follows its own SDE, also presented in Section~\ref{sec:PMP_multiasset_state_costate}. The dynamics are driven by correlated Brownian motions \(\mathbf{W}_t^X \in \mathbb{R}^n\) and \(\mathbf{W}_t^Y \in \mathbb{R}^k\) with the joint covariance structure \(\operatorname{Cov}(d\mathbf{W}_t) = \begin{pmatrix} \Psi & \rho \\ \rho^\top & \Phi \end{pmatrix} dt.\)

Since \(\lambda_t = V_X\!\bigl(t,\,X_t,\,\mathbf{Y}_t\bigr),\) we apply Ito's multidimensional lemma to \(V_X\). The differential \(d\lambda_t\) can be expressed as:
\begin{align}
d\lambda_t
&=\; V_{Xt}\,dt \;+\; V_{XX}\,dX_t \;+\; V_{X\mathbf{Y}}\,d\mathbf{Y}_t \nonumber\\
&\quad +\; \tfrac{1}{2}\,V_{XXX}\,\bigl(dX_t\bigr)^2 \;+\; V_{XX\mathbf{Y}}\,(dX_t)\,(d\mathbf{Y}_t)^\top \;+\; \tfrac{1}{2}\, \text{Tr}\bigl(V_{X\mathbf{Y}\mathbf{Y}}\,(d\mathbf{Y}_t)(d\mathbf{Y}_t)^\top\bigr), 
\label{eq:ito_adjdiff_revised_flow} 
\end{align}
where subscripts denote partial derivatives (e.g., \(V_{X\mathbf{Y}}\) is the gradient w.r.t \(\mathbf{Y}\) of \(V_X\), \(V_{X\mathbf{Y}\mathbf{Y}}\) is the Hessian w.r.t \(\mathbf{Y}\) of \(V_X\)).

Our objective is to identify the coefficient of \(d\mathbf{W}_t^X\) in the expansion of \(d\lambda_t\). This coefficient is precisely the process \(\mathbf{Z}_t^{\lambda X}\) appearing in the BSDE \eqref{eq:lambda_BSDE}. Substituting the SDEs for \(dX_t\) and \(d\mathbf{Y}_t\) (from Section~\ref{sec:PMP_multiasset_state_costate}) into the Ito formula \eqref{eq:ito_adjdiff_revised_flow}, we observe that terms involving \(d\mathbf{W}_t^X\) arise directly from the diffusion part of the \(V_{XX}\,dX_t\) term, and indirectly from the \(V_{X\mathbf{Y}}\,d\mathbf{Y}_t\) term due to the correlation \(\rho\) between \(d\mathbf{W}_t^X\) and \(d\mathbf{W}_t^Y\). A careful application of Ito's lemma, considering the quadratic variation and covariation terms related to the covariance matrix \(\begin{pmatrix} \Psi & \rho \\ \rho^\top & \Phi \end{pmatrix}\), allows isolation of the full coefficient multiplying \(d\mathbf{W}_t^X\). Representing \(\mathbf{Z}_t^{\lambda X}\) as a \(1 \times n\) row vector, this process is found to be: 
\begin{equation}
\mathbf{Z}_t^{\lambda X} = \underbrace{(\partial_x \lambda_t^*) X_t^* \boldsymbol{\pi}_t^{*\top} \boldsymbol{\sigma}(t,\mathbf{Y}_t^*)}_{\text{Term from } V_{XX}dX_t} 
+ \underbrace{(\partial_{\mathbf{Y}} \lambda_t^*) \boldsymbol{\sigma}_Y(t,\mathbf{Y}_t^*) \rho^\top \Psi^{-1}}_{\text{Term from } V_{XY}dY_t \text{ via correlation}}. 
\label{eq:Z_lambda_X_derived_revised} 
\end{equation}
The first term reflects the direct impact of wealth volatility related to the portfolio choice \(\boldsymbol{\pi}_t^*\), scaled by the second derivative of the value function with respect to wealth, \(V_{XX}\) (or \(\partial_x \lambda_t^*\)). The second term captures the intertemporal hedging component arising from the correlation between asset returns and state variable movements, mediated by the cross-derivative of the value function \(V_{X\mathbf{Y}}\) (or \(\partial_{\mathbf{Y}}\lambda_t^*\)).

This derivation provides the explicit form of the adjoint diffusion coefficient \(\mathbf{Z}_t^{\lambda X}\), which is fundamental for understanding the structure of the optimal portfolio policy \eqref{eq:pi_star} and for constructing the Pontryagin adjoint equations.


\section{Proofs}

\subsection{Proof of Lemma \ref{lem:pathwise_bptt_rec_full}}

The recurrence relation for the pathwise costate $\lambda_k^{(i)} := \partial J^{(i)}/\partial X_k^{(i)}$ is derived by applying the chain rule. The total reward $J^{(i)}$ is decomposed into the utility from the current step $k$ and the total future reward $J_{k+1}^{(i)}$.
\[
    \lambda_k^{(i)} = \frac{\partial}{\partial X_k^{(i)}}\left( e^{-\delta(t_k-t_0)}U(C_k^{(i)})\Delta t \right) + \frac{\partial J_{k+1}^{(i)}}{\partial X_k^{(i)}}.
\]
The second term becomes $\lambda_{k+1}^{(i)} (\partial X_{k+1}^{(i)}/\partial X_k^{(i)})$. Differentiating the discrete wealth dynamics, which includes the consumption term $-C_k^{(i)}\Delta t$, with respect to $X_k^{(i)}$ yields:
\[
    \frac{\partial X_{k+1}^{(i)}}{\partial X_k^{(i)}} = 1 + \bigl(r_k+\boldsymbol{\pi}_k^{\top}(\boldsymbol{\mu}_k-r_k\mathbf{1})\bigr)\Delta t - \frac{\partial C_k^{(i)}}{\partial X_k^{(i)}}\Delta t + \boldsymbol{\pi}_k^{\top}\boldsymbol{\sigma}_k \Delta\mathbf{W}_k^{X,(i)}.
\]
Substituting these components and rearranging the terms for $\lambda_k^{(i)}$ yields the expression in Eq.~\eqref{eq:pathwise_rec_full}.

\subsection{Proof of Lemma \ref{lem:quad_drift_full}}

The product inside the expectation is an inner product of vectors. The term involving the adapted part, $\mathbb{E}[\lambda_{k+1}(\dots)|\mathcal{F}_{t_k}]$, is zero as it is the expectation of a martingale increment. The non-zero terms arise from the quadratic covariation of the raw martingale components of $\lambda_{k+1}^{(i)}$ with the wealth shock. Using the standard properties $\mathbb{E}[(\Delta\mathbf W_k^{X})(\Delta\mathbf W_k^{X})^\top|\mathcal{F}_{t_k}]=I\Delta t$ and $\mathbb{E}[(\Delta\mathbf W_k^{Y})(\Delta\mathbf W_k^{X})^\top|\mathcal{F}_{t_k}]=\boldsymbol{\rho}_{Y\!X}\Delta t$, where $\boldsymbol{\rho}_{Y\!X}$ is the cross-correlation matrix, yields the stated result.

\subsection{Proof of Theorem \ref{thm:bptt_pmp_correspondence_full}}

We begin by taking the conditional expectation $\mathbb{E}[\cdot\mid\mathcal{F}_{t_k}]$ of the entire pathwise recurrence in Eq.~\eqref{eq:pathwise_rec_full} from Lemma~\ref{lem:pathwise_bptt_rec_full}. The left-hand side becomes $\lambda_k$ by definition. On the right-hand side, the expectations of the drift terms are straightforward. The crucial step is handling the stochastic term $\mathbb{E}[\lambda_{k+1}^{(i)}(\dots)|\mathcal{F}_{t_k}]$, for which we substitute the result from Lemma~\ref{lem:quad_drift_full}. This transforms the stochastic term into the additional drift component involving the effective process $\mathbf{Z}^{\lambda X}_k$. After collecting all drift terms and separating the martingale components, we arrive at the discrete-time BSDE presented in Eq.~\eqref{eq:disc_bsde_costate_full}. The final step involves a direct algebraic comparison of this drift term with the definition of the Hamiltonian $\mathcal{H}$ from Section~\ref{sec:PMP_multiasset_state_costate}, which confirms their identity.

\subsection{Proof of Theorem~\ref{thm:policy_gap}}
\label{app:math_foundation_publishable}

This appendix provides a complete, self-contained analytic backbone for Theorem~\ref{thm:policy_gap}.
All proof sketches are expanded into detailed arguments, relying on parabolic maximal regularity (MR) with VMO (vanishing mean oscillation) coefficients and
the strong monotonicity of the PMP first-order condition (FOC) map.
The reflected SDE framework is used purely as a localization device; passing to expanding domains
recovers the original non-reflected problem \citep{karatzas1991brownian,lieberman1996second}.

\subsubsection{Table of Symbols}

\begin{center}
\begin{tabular}{cc}
\toprule
Symbol & Meaning \\ \midrule
$E_v,E_h,E_f$ & Gradient / coefficient / source errors \\
$\varepsilon$ & FOC residual norm \\
$\delta_{\mathrm{BPTT}}$ & Costate estimation error (Def.~\ref{def:bptt_error_publishable}) \\
$\underline\kappa$ & Strong monotonicity modulus of the FOC map (Lemma~\ref{lem:strong_monotonicity}) \\
$C_S,C_v''$ & PDE stability / coupled-error constants \\
$L_G, L_F, C_{\text{tot}}$ & FOC-map Lipschitz / final amplification constants \\
$\tau_\star$ & Contraction slab length (distinct from numerical step $\Delta t$) \\
\bottomrule
\end{tabular}
\end{center}

\subsubsection{Preliminaries and Notation}

Throughout, set the primary costate $\lambda := \partial_x V^{\pi,c}$.
For any admissible control $(\pi,c)$ the sub-value function $V^{\pi,c}$ solves
\[
  -\partial_t V^{\pi,c}
  \;=\; \cL_{\pi,c}[V^{\pi,c}] + f_{\pi,c},\qquad
  V^{\pi,c}(T,\cdot)=g,
\]
on $Q_T:=(0,T)\times\Omega$, with oblique-derivative boundary $\partial_{\mathbf n}V=0$ on $\partial\Omega$
matching a Skorokhod reflection used for localization; solutions on expanding domains converge
to the non-reflected solution \citep{karatzas1991brownian,lieberman1996second}.

Define the error quantities on $Q_T$ by
\[
  \begin{aligned}
    h &:= (a_{\pi,c}-a_{\ast},\,b_{\pi,c}-b_{\ast}), &\quad E_h &:= \|h\|_{L^{q,p}},\\
    E_f &:= \bigl\|e^{-\delta t}\bigl(U(c^{\ast})-U(c)\bigr)\bigr\|_{L^{q,p}}, &\quad
    E_v &:= \|\nabla V^{\pi,c}-\nabla V^{\ast}\|_{L^{q,p}},\\
    \varepsilon &:= \|\mathcal R_{\mathrm{FOC}}\|_{L^{q,p}}, &&
  \end{aligned}
\]

\paragraph{Convention on $Z_{\lambda X}$.}
We use
\[
  Z_{\lambda X} \;=\; (\partial_x\lambda)\,X\,\pi^{\top}\sigma \;+\; (\partial_Y\lambda)\,\sigma_Y\,\rho^{\top},
  \qquad
  \sigma\,Z_{\lambda X}^{\top}
  \;=\; X\bigl(-\partial_x\lambda\bigr)\,\Sigma\,\pi \;+\; \sigma\rho\sigma_Y^{\top}(\partial_Y\lambda)^{\top},
\]
consistent with the portfolio FOC (cf.\ main text Eq.~(11)); see \citet{yong2012stochastic} for PMP–HJB links.

\subsubsection{Minimal Assumptions}\label{subsec:assumption_publishable}

\begin{assumption}[Baseline well-posedness and regularity]\label{ass:baseline}
Fix $T>0$, a bounded $C^{1,1}$ domain $\Omega$, mixed-norm exponents $p>d+2$ and $q\ge 2$, and a compact cylinder
$\mathcal D=[0,T]\times[\underline x,\bar x]\times\{\|y\|\le R_Y\}$.
There exist constants $\underline\nu,\overline\nu,L_{\mathrm{coef}},C_{\mathrm{MR}},\eta_0>0$ such that:
\begin{enumerate}[leftmargin=1.5em] 
\item The Skorokhod-reflected SDE with inward normal $\mathbf n\in C^{1,\alpha}$ is well posed; expanding domains
converge to the non-reflected solution \citep{karatzas1991brownian}.
\item Parabolicity:
\begin{itemize}[leftmargin=2.5em] 
\item[(U)] Uniform ellipticity: $\underline\nu|\xi|^2\le\xi^{\top} a_{\pi,c}\xi\le\overline\nu|\xi|^2$
and $(a_{\pi,c},b_{\pi,c})$ have VMO/BMO (bounded mean oscillation) bounds on $\mathcal D$ \citep{dong2009parabolic};
\item[(H)] The $Y$-block is uniformly elliptic with VMO bounds while the $X$-block may degenerate; MR and gradient bounds below
are applied on interior cylinders $Q_r\Subset\mathcal D$ \citep{lieberman1996second}.
\end{itemize}
\item There exists $R_0>0$ such that $\sup_{0<r\le R_0}\omega_a(r)\le\eta_0<\bar\eta(d,p,q)$ as in \citet{dong2009parabolic}.
\item Costate and curvature bands on $\mathcal D$: $\underline\lambda\le\lambda\le\bar\lambda$ and $\underline\kappa_x\le -\partial_x\lambda\le \bar\kappa_x$.
\item $U(c)=c^{1-\gamma}/(1-\gamma)$ with $\gamma>0$, $\gamma\neq 1$, and coefficients obey linear growth/boundedness.
\end{enumerate}
\end{assumption}

\begin{remark}[Ellipticity and interior estimates]
Regime (U) can be enforced by a local lower bound on $\|\pi\|$ on $\mathcal D$.
Under (H), arguments apply on $Q_r\Subset\mathcal D$ with local constants; a finite cover yields global bounds \citep{lieberman1996second}.
\end{remark}

\subsubsection{Automatic Properties of the Model}

\begin{lemma}[Strong monotonicity of the portfolio FOC map]\label{lem:strong_monotonicity}
Define
\[
  \mathcal G_{\pi}(\pi;c,\lambda,\partial\lambda)
  \;:=\; \lambda(\mu-r\mathbf 1)+\sigma\,Z_{\lambda X}^{\top}
  \;=\; X\bigl(-\partial_x\lambda\bigr)\,\Sigma\,\pi \;+\; \lambda(\mu-r\mathbf 1)\;+\;\sigma\rho\sigma_Y^{\top}(\partial_Y\lambda)^{\top}.
\]
If $\Sigma\succ 0$ and Assumption~\ref{ass:baseline} holds, then $\partial_{\pi}\mathcal G_{\pi}=X(-\partial_x\lambda)\,\Sigma\succ 0$.
Hence for all $\pi_1,\pi_2$,
\[
  \big\langle \mathcal G_{\pi}(\pi_1)-\mathcal G_{\pi}(\pi_2),\,\pi_1-\pi_2\big\rangle
  \;\ge\; \underline\kappa\,\|\pi_1-\pi_2\|^2,\qquad
  \underline\kappa:=\underline x\,\underline\kappa_x\,\lambda_{\min}(\Sigma)>0,
\]
so $\mathcal G_{\pi}$ is strongly monotone and has a Lipschitz inverse \citep{browder1965nonlinear,minty1962monotone}.
\end{lemma}

\begin{lemma}[Maximal regularity for $V^{\pi,c}$]\label{lem:max_regularity}
Under Assumption~\ref{ass:baseline} and $\eta_0<\bar\eta(d,p,q)$, both $V^{\pi,c}$ and $V^{\ast}$ lie in $W^{2,1}_{q,p}(Q_T)$ and
\[
  \|D^2V^{\pi,c}\|_{L^{q,p}}+\|\partial_t V^{\pi,c}\|_{L^{q,p}}
  \;\le\; C_{\mathrm{MR}}\Bigl(\|f_{\pi,c}\|_{L^{q,p}}+\|g\|_{W^{2-2/p}_q}\Bigr),
\]
with the same bound for $V^{\ast}$; under (H) the estimate holds on $Q_r\Subset\mathcal D$ \citep{dong2009parabolic}.
\end{lemma}

\subsubsection{Foundational Lemmas}

\begin{lemma}[FOC reduction: coefficient gap]\label{lem:foc_reduction_publishable}
Under Assumption~\ref{ass:baseline} and Lemma~\ref{lem:strong_monotonicity},
\[
  E_h \;\le\; L_G\,E_v \;+\; C_{\mathrm{FOC}}\,\varepsilon.
\]
\end{lemma}

\begin{proof}
By Lemma~\ref{lem:strong_monotonicity}, $\pi\mapsto\mathcal G_{\pi}(\pi;\lambda,\partial\lambda)$ has a Lipschitz inverse \citep{browder1965nonlinear,minty1962monotone}.
Thus $\|\pi-\pi^{\ast}\|\le c_1\bigl(\|\nabla V^{\pi,c}-\nabla V^{\ast}\|+\|\mathcal R_{\mathrm{FOC}}\|\bigr)$.
Lipschitz dependence of coefficients on controls yields the claim after absorbing constants.
\end{proof}

\begin{lemma}[Source gap]\label{lem:source_gap_publishable}
Under Assumption~\ref{ass:baseline} and HARA utility,
\[
  E_f \;\le\; C_f' E_v + C_f \,\varepsilon.
\]
\end{lemma}

\begin{proof}
On the costate band, $U$ and $(U')^{-1}$ are Lipschitz on the relevant ranges; apply a mean-value argument to the consumption FOC
(see PMP/HJB expositions in \citealp{yong2012stochastic}).
\end{proof}

\begin{lemma}[Gradient stability]\label{lem:gradient_stab_publishable}
With Lemma~\ref{lem:max_regularity} in force,
\[
  E_v \;\le\; C_S\,(E_h+E_f).
\]
\end{lemma}

\begin{proof}
Let $W:=V^{\pi,c}-V^{\ast}$. Then $W$ solves a linear parabolic PDE with source determined by $(h,f)$.
Parabolic MR with VMO coefficients yields $\|\nabla W\|_{L^{q,p}}\le C_S(\|h\|_{L^{q,p}}+\|f\|_{L^{q,p}})$ \citep{dong2009parabolic}.
\end{proof}

\begin{proposition}[Coupled error resolution]\label{prop:coupled_publishable}
There exists $C_v''>0$ such that $E_v\le C_v''\,\varepsilon$.
\end{proposition}

\begin{proof}
Combine the three lemmas to get $E_v\le\alpha E_v+\beta\varepsilon$, $\alpha=C_S(L_G+C_f')$.
For some slab length $\tau_\star>0$ (set by MR constants), on any slab of length $\tau\le\tau_\star$ one has $\alpha<1$, hence
$E_v\le \beta(1-\alpha)^{-1}\varepsilon$.
Concatenate $\lceil T/\tau_\star\rceil$ slabs. This is an application of the contraction principle.
\end{proof}

\subsubsection{Final Policy-Gap Theorem}

\begin{definition}[BPTT error]\label{def:bptt_error_publishable}
Let $\widehat\lambda$ estimate $\nabla V^{\pi,c}$ by BPTT.
For time step $\Delta t$ and batch size $M$,
\[
\delta_{\mathrm{BPTT}}
:=\|\widehat\lambda-\nabla V^{\pi,c}\|_{L^{q,p}}
\;\le\;
\begin{cases}
\kappa_{\mathrm{EM}}\,\Delta t^{1/2} + \kappa_2/\sqrt M, \\
\kappa_{\mathrm{Mil}}\,\Delta t \; + \kappa_2/\sqrt M \quad \text{(with first-order acceleration),}
\end{cases}
\]
in line with standard SDE error theory \citep{glasserman2004monte}.
\end{definition}

{
\noindent\textbf{Theorem \ref{thm:policy_gap} (Policy-Gap Bound).}
\itshape
Let $\pi^{\ast}$ be optimal and $\widehat\pi$ the P-PGDPO policy trained with time step $\Delta t$ and batch size $M$.
Under Assumption~\ref{ass:baseline},
\[
     \|\widehat\pi-\pi^{\ast}\|_{L^{q,p}}
     \;\le\;
     C_{\text{tot}} \bigl( \varepsilon + \delta_{\mathrm{BPTT}} \bigr).
\]
Equivalently,
\[
\|\widehat\pi-\pi^{\ast}\|_{L^{q,p}}
\;\le\;
\begin{cases}
C_{\text{tot}}\bigl(\varepsilon + \kappa_{\mathrm{EM}}\,\Delta t^{1/2} + \kappa_2/\sqrt M\bigr),\\[3pt]
C_{\text{tot}}\bigl(\varepsilon + \kappa_{\mathrm{Mil}}\,\Delta t + \kappa_2/\sqrt M\bigr).
\end{cases}
\]
}

\begin{proof}
By the triangle inequality,
\[
  \|\widehat\pi-\pi^{\ast}\|
  \;\le\; L_F\Bigl(\|\widehat\lambda-\nabla V^{\pi,c}\|+\|\nabla V^{\pi,c}-\nabla V^{\ast}\|\Bigr)
  \;\le\; L_F\bigl(\delta_{\mathrm{BPTT}}+E_v\bigr),
\]
and substitute Definition~\ref{def:bptt_error_publishable} and Proposition~\ref{prop:coupled_publishable}.
\end{proof}

The preceding assumptions, lemmas, and theorem complete the core logical chain for the policy-gap bound.
What follows is not part of the minimal existence–uniqueness–bound argument itself, but collects a number
of complementary considerations and implementation notes that help ensure the assumptions are met in practice
and the results can be applied robustly in computational or extended theoretical settings.
These points are offered as guidance for verification, numerical implementation, and consistency across
related parts of the manuscript.

\begin{enumerate}
\item \textbf{VMO smallness (how to guarantee).}
Two routes:
\begin{enumerate}
\item \emph{Top-level assumption.} State explicitly that the $x$-VMO modulus of leading coefficients satisfies $\sup_{0<r\le R_0}\omega_a(r)\le \eta_0<\bar\eta(d,p,q)$.
\item \emph{Verification from primitives.} If $y\mapsto \Sigma(y)$ and $y\mapsto \sigma_Y(y)$ are uniformly continuous with moduli $\omega_\Sigma(r),\omega_{\sigma_Y}(r)\to 0$ and the policy map is locally Lipschitz on the costate band, then for small $r$,
\[
\omega_a(r)\;\le\; C\bigl(\omega_\Sigma(r)+\omega_{\sigma_Y}(r)\bigr).
\]
If $\omega_\Sigma(r),\omega_{\sigma_Y}(r)\le C_0 r^{\alpha}$, choose $R_0\le (\eta_0/(C C_0))^{1/\alpha}$.
\end{enumerate}
See \citet{dong2009parabolic} for the VMO framework in parabolic equations.

\item \textbf{Norm coherence.}
All bounds are in $L^{q,p}(Q_T)$ / $W^{2,1}_{q,p}(Q_T)$. If elsewhere an equilibrium operator $T:X\to X^{\ast}$ is used, fix $X=W^{1,p}(\Omega)$ with
\[
\|u\|_{X} \;=\; \bigl(\|u\|_{L^p}^p+\|\nabla u\|_{L^p}^p\bigr)^{1/p},
\]
and state coercivity/monotonicity in this norm (avoid using only $\|u\|_{L^p}$).
See \citet{adams2003sobolev} for Sobolev space conventions.

\item \textbf{Operator interchanges (Fubini–Tonelli/DCT).}
With bounded controls and reflection on $C^{1,1}$ domains, moments are finite, so for $\varphi\in C^{1}$ of polynomial growth and integrable weight $\psi$,
\[
  \frac{d}{dt}\,\E\bigl[\psi(t)\,\varphi(t,X_t,Y_t)\bigr]
  \;=\; \E\bigl[\psi'(t)\,\varphi(t,X_t,Y_t)\bigr]
  \;+\; \E\bigl[\psi(t)\,(\partial_t + \cL)\,\varphi(t,X_t,Y_t)\bigr],
\]
and dominated convergence justifies interchanging $\int_0^T$ and $\E[\cdot]$.
See \citet{karatzas1991brownian} for reflected SDE background.

\item \textbf{Numerical convergence (strong vs.\ weak).}
For Euler–Maruyama with step $h$ and smooth coefficients:
\[
\bigl(\E\|Z_T^{\mathrm{EM}}-Z_T\|^2\bigr)^{1/2}=\cO(h^{1/2})\quad\text{(strong)},\qquad
\bigl|\E\phi(Z_T^{\mathrm{EM}})-\E\phi(Z_T)\bigr|=\cO(h)\quad\text{(weak)}.
\]
This matches Definition~\ref{def:bptt_error_publishable}; Milstein/RE gives first-order in the discretization term.
See \citet{kloeden1992numerical} for SDE numerical analysis.

\item \textbf{Time-slab vs.\ time-step.}
The contraction uses an \emph{analytic} slab length $\tau\le \tau_{\star}$ tied to PDE constants; it is independent of the numerical step $\Delta t$.

\item \textbf{Optional GE module (if used elsewhere).}
If you introduce $T:X\to X^{\ast}$ with
\[
\langle T(u),v\rangle
=\int_{\Omega}\!\Big(A\nabla u\cdot\nabla v + Buv + N(u)\,v\Big)\,dx,
\]
$A\succeq \underline\lambda I$, $B\ge 0$, and $N$ monotone with $N(u)u\ge c_0|u|^p-C_0$, then
\[
\langle T(u),u\rangle \;\ge\; c\,\|u\|_{W^{1,p}}^p - C,
\]
so coercivity condition holds in the \emph{same} $W^{1,p}$ norm.
See \citet{browder1965nonlinear} for monotone operator theory.
\end{enumerate}

\section{Classification of Parabolic HJB PDEs in Financial Models}\label{app:parabolic_financial_models}

Table \ref{tab:parabolic_classification} summarizes the classification of parabolic HJB PDEs in financial markets. Although the theoretical proofs for classes (H) and (D) are challenging, our P-PGDPO is applicable to the broader classes of parabolic PDEs from the financial and economic models. 

\begin{table}[h!]
\centering
\small
\caption{Classification of Parabolic HJB PDEs in Financial Models}
\label{tab:parabolic_classification}
\begin{tabular}{@{} p{3.5cm} p{4.5cm} p{7cm} @{}}
\toprule
Class & Diffusion Matrix $a(t,x,\mathbf{y})$ Property & SDE Coefficient Requirements and Dynamic Portfolio Examples \\
\midrule
\textbf{(U) Uniformly Elliptic} & Uniformly positive definite across the state space & The diffusion coefficients $\boldsymbol{\sigma}$ (for wealth $X_t$) and $\boldsymbol{\sigma}_Y$ (for state $\mathbf{Y}_t$) are such that the resulting diffusion matrix in the HJB is non-zero and bounded below by a positive constant everywhere (e.g., $\boldsymbol{\sigma}\Psi\boldsymbol{\sigma}^\top$ and $\operatorname{diag}(\boldsymbol{\sigma}_Y)\Phi\operatorname{diag}(\boldsymbol{\sigma}_Y)^\top$ uniformly positive definite). Drift coefficients ($\boldsymbol{\mu}, \boldsymbol{\mu}_Y$) are also sufficiently regular (e.g., Lipschitz). \\
($\implies$ Value Function / Costate very smooth) & & \textbf{Examples:} Merton's portfolio problem with constant coefficients ($\mu, \sigma > 0, r$ are constants); models with state-dependent volatility that remains strictly positive; the multi-asset OU model with one factor in Section \ref{sec:multiasset_multifactor_KO} (if asset volatilities $\sigma_i > 0$, factor volatility $\sigma_Y > 0$, and $\Psi$ is positive definite). \\[6pt]
\midrule
\textbf{(H) Hypo-elliptic} & May degenerate (even globally). Regularity is achieved if the system structure satisfies conditions like Hörmander's (i.e., Lie brackets of vector fields span the tangent space). The problem might be uniformly elliptic on compact subsets of interest. & Diffusion coefficients ($\boldsymbol{\sigma}$ or $\boldsymbol{\sigma}_Y$) can become zero at certain points within the domain. However, interaction with non-zero drift terms or other diffusion terms maintains regularity. \\
($\implies$ Value Function / Costate smooth in relevant regions) & & \textbf{Examples:} The multi-asset and multi-factor model in Section~\ref{sec:multiasset_multifactor_KO} (if some factor volatilities in $\boldsymbol{\sigma}_Y$ can be zero at certain states but the overall system structure, possibly involving $\kappa_Y$, ensures hypo-ellipticity); a Vasicek interest rate model affecting asset returns $\mu$ but not $\sigma$. \\[6pt]
\midrule
\textbf{(D) Degenerate at the Boundary} & Degenerates ($\to 0$) as the state approaches a boundary. & The diffusion coefficients $\boldsymbol{\sigma}$ or $\boldsymbol{\sigma}_Y$ (or the combined diffusion matrix) approach zero as a state variable (e.g., wealth, volatility, interest rate) approaches a boundary (typically 0). \\
($\implies$ Value Function / Costate regularity may break down near the boundary) & & \textbf{Examples:} CIR (Cox-Ingersoll-Ross) interest rate/volatility model (diffusion term like $\sigma_Y \sqrt{Y_t}$); Heston stochastic volatility model (asset diffusion term like $\sigma \sqrt{\nu_t}$). \\
\bottomrule
\end{tabular}

\end{table}

\section{Detailed Experimental Setup}
\label{app:experimental_setup_details}

This section provides a detailed setup for the numerical experiments presented in Section~\ref{sec:numerical_experiments}. The implementation code and detailed explanations can be found at \url{https://github.com/huhjeonggyu/PGDPO}.

\subsection{Parameter Values}
\label{app:param_values} 

The experiments consider the objective of maximizing the expected utility of terminal wealth \(X_T\), \(J = \mathbb{E}[U(X_T)]\), using a CRRA utility function \(U(x) = x^{1-\gamma}/(1-\gamma)\) with a relative risk aversion coefficient \(\gamma = 2.0\). The fixed time horizon is \(T = 1.5\).

Unless otherwise specified, the model parameters (for the model described in Section~\ref{sec:multiasset_multifactor_KO}) are generated based on the asset dimension \(n\) and factor dimension \(k\) using the following process (seeded for reproducibility, e.g., `seed=42'):
\begin{itemize}
    \item Risk-free rate: \(r = 0.03\).
    \item State process parameters (for the \(k\)-dimensional OU process \(\mathbf{Y}_t\)):
        \begin{itemize}
            \item Mean-reversion speeds: \(\boldsymbol{\kappa}_Y = \mathrm{diag}(2.0, 2.5, \dots, 2.0 + (k-1)0.5)\).
            \item Long-term means: \(\boldsymbol{\theta}_Y\) components drawn uniformly from \(U(0.2, 0.4)\).
            \item Factor volatilities: \(\boldsymbol{\sigma}_Y\) diagonal elements drawn uniformly from \(U(0.3, 0.5)\).
        \end{itemize}
    \item Asset parameters (for the \(n\) risky assets):
        \begin{itemize}
            \item Volatilities: \(\boldsymbol{\sigma} \in \mathbb{R}^n\) components drawn uniformly from \(U(0.1, 0.5)\).
            \item Factor loadings (\(\boldsymbol{\alpha} \in \mathbb{R}^{n \times k}\)): For each asset \(i = 1, \dots, n\), the \(k\)-dimensional row vector of factor loadings \(\boldsymbol{\alpha}_i^\top = (\alpha_{i1}, \dots, \alpha_{ik})\) is drawn independently from a Dirichlet distribution. This is done using `scipy.stats.dirichlet` with a concentration parameter vector consisting of \(k\) ones (\(\mathbf{1}_k = [1.0, \dots, 1.0]\)). Consequently, for each asset \(i\), all loadings are non-negative (\(\alpha_{ij} \ge 0\)) and they sum to unity (\(\sum_{j=1}^k \alpha_{ij} = 1\)). (Note: These loadings define the risk premium structure \(\boldsymbol{\mu}(t, \mathbf{Y}_t) = r\mathbf{1} + \mathrm{diag}(\boldsymbol{\sigma}) \boldsymbol{\alpha} \mathbf{Y}_t\)).
        \end{itemize}
    \item Correlation structure: The generation process aims to produce realistic correlation parameters while ensuring the overall block correlation matrix remains positive definite, which is crucial for stable simulations. This involves: 
        \begin{itemize}
            \item Asset correlation \(\Psi\) (\(n \times n\)) generated based on a latent factor structure ensuring validity.
            \item Factor correlation \(\Phi_Y\) (\(k \times k\)) generated as a random correlation matrix.
            \item Cross-correlation \(\boldsymbol{\rho}_Y\) (\(n \times k\)) components drawn uniformly from \(U(-0.2, 0.2)\).
            \item The full \((n+k) \times (n+k)\) block correlation matrix \(\begin{pmatrix} \Psi & \rho_Y \\ \rho_Y^\top & \Phi_Y \end{pmatrix}\) is constructed and numerically adjusted (e.g., via minimal diagonal loading) if necessary to ensure positive definiteness (tolerance \(10^{-6}\)).
        \end{itemize}
\end{itemize}
The specific values generated by the seed are used consistently across experiments for given \(n\) and \(k\).

\subsection{PG-DPO Algorithm Details}\label{app:PGDPO_algorithm}

Each iteration \(j\) of the algorithm involves the following steps:

\begin{enumerate}
 \item[Step 1.] \textbf{Sample Initial States:} Draw a mini-batch of \(M\) initial states \(\{(t_0^{(i)}, x_0^{(i)}, \mathbf{y}_0^{(i)})\}_{i=1}^M\) from the distribution \(\eta\) over the domain \(\mathcal{D}\).

 \item[Step 2.] \textbf{Simulate Forward Paths:} For each sampled initial state \(i\), simulate a discrete-time trajectory \(\{(X_k^{(i)}, \mathbf{Y}_k^{(i)})\}_{k=0}^N\) forward from \(t_k = t_0^{(i)}\) to \(t_N = T\), using the current policy networks \((\boldsymbol{\pi}_\theta, C_\phi)\). The state evolution follows:
 \[
 \begin{aligned}
  \mathbf{Y}_{k+1}^{(i)} &= \mathbf{Y}_k^{(i)} + \boldsymbol{\mu}_Y(t_k^{(i)}, \mathbf{Y}_k^{(i)}) \Delta t^{(i)} + \boldsymbol{\sigma}_Y(t_k^{(i)}, \mathbf{Y}_k^{(i)}) \Delta\mathbf{W}_k^{(i),Y}, \\
  X_{k+1}^{(i)} &= X_k^{(i)} + \left[ X_k^{(i)} \left( (1 - \boldsymbol{\pi}_k^{(i)\top}\mathbf{1}) r(t_k^{(i)}, \mathbf{Y}_k^{(i)}) + \boldsymbol{\pi}_k^{(i)\top} \boldsymbol{\mu}(t_k^{(i)}, \mathbf{Y}_k^{(i)}) \right) - C_k^{(i)} \right] \Delta t^{(i)} \\
  &\quad + X_k^{(i)} \boldsymbol{\pi}_k^{(i)\top} \boldsymbol{\sigma}(t_k^{(i)}, \mathbf{Y}_k^{(i)}) \Delta\mathbf{W}_k^{(i),X},
 \end{aligned}
 \]
 where \(\Delta t^{(i)} = (T - t_0^{(i)})/N\), \(\boldsymbol{\sigma}_Y(t_k^{(i)}, \mathbf{Y}_k^{(i)})\) is the vector of diffusion coefficients applied element-wise to the components of $\Delta\mathbf{W}_k^{(i),Y}$, and the  correlated Brownian increments \((\Delta\mathbf{W}_k^{(i),X}, \Delta\mathbf{W}_k^{(i),Y})\) are derived from the Cholesky factor \(L\) of the  covariance matrix \(\boldsymbol{\Omega}\).

 \item[Step 3.] \textbf{Calculate Realized Reward per Path:} For each path \(i\), compute the realized cumulative reward, explicitly noting its dependence on the starting state and policy parameters, denoted as \(J^{(i)}(t_0^{(i)}, x_0^{(i)}, \mathbf{y}_0^{(i)}; \theta, \phi)\):
 \[
 J^{(i)}(t_0^{(i)}, x_0^{(i)}, \mathbf{y}_0^{(i)}; \theta, \phi) = \sum_{k=0}^{N-1} e^{-\delta (t_k^{(i)} - t_0^{(i)})} U\bigl(C_\phi(t_k^{(i)}, X_k^{(i)}, \mathbf{Y}_k^{(i)})\bigr) \Delta t^{(i)} + \kappa e^{-\delta (T - t_0^{(i)})} U(X_N^{(i)}).
 \]
 This value \(J^{(i)}(t_0^{(i)}, \dots)\) is a stochastic sample related to the conditional expected utility \(J(t_0^{(i)}, \dots)\) defined in Eq.~\eqref{eq:policy_value_function}.

 \item[Step 4.] \textbf{Apply BPTT for Policy Gradient Estimate:} Treat the simulation process for path \(i\) as a computational graph. Apply automatic differentiation (BPTT) to compute the gradient of the realized reward \(J^{(i)}(t_0^{(i)}, \dots)\) with respect to the policy parameters \((\theta, \phi)\):
 \[
 \nabla_{(\theta, \phi)} J^{(i)}(t_0^{(i)}, x_0^{(i)}, \mathbf{y}_0^{(i)}; \theta, \phi).
 \]
 This gradient serves as an unbiased estimate related to the gradient of the conditional expected utility \(J(t_0^{(i)}, \dots)\).

 \item[Step 5.] \textbf{Estimate Gradient of \(\widetilde{J}\):} Obtain an unbiased estimate for the gradient of the overall extended objective \(\widetilde{J}\) by averaging the individual path gradients over the mini-batch:
 \[
 \nabla \widetilde{J}(\theta, \phi) \approx \nabla_{\text{batch}} = \frac{1}{M} \sum_{i=1}^M \nabla_{(\theta, \phi)} J^{(i)}(t_0^{(i)}, x_0^{(i)}, \mathbf{y}_0^{(i)}; \theta, \phi).
 \]

 \item[Step 6.] \textbf{Update Policy Parameters:} Update the parameters \((\theta, \phi)\) using this averaged gradient estimate via stochastic gradient ascent:
 \[
 (\theta, \phi)_{j+1} \leftarrow (\theta, \phi)_j + \alpha_j \, \nabla_{\text{batch}},
 \]
 where \(\alpha_j\) is the learning rate.
\end{enumerate}

By iterating steps from 1 to 6, the baseline PG-DPO algorithm gradually optimizes the policy \((\boldsymbol{\pi}_\theta, C_\phi)\) to maximize the extended value function \(\widetilde{J}\).

\subsection{PG-DPO Implementation Details}
\label{app:pgdpo_impl} 

We implement the Baseline PG-DPO algorithm using PyTorch.
\begin{itemize}
    \item \textbf{Policy Network (\(\boldsymbol{\pi}_\theta\)):} The policy network  consists of an input layer taking the current wealth \(W_t\), current time \(t\), and state vector \(\mathbf{Y}_t\) (total \(2+k\) inputs), followed by three hidden layers with 200 units each and LeakyReLU activation, and an output layer producing the \(n\)-dimensional portfolio weight vector \(\boldsymbol{\pi}_\theta\).
    \item \textbf{Optimizer:} Adam optimizer with a learning rate of \(1.0 \times 10^{-5}\).
    \item \textbf{Training:}
        \begin{itemize}
            \item Total epochs: \(10,000\).
            \item Batch size (\(M\)): \(1,000\).
            \item Time discretization steps (\(N\)): \(20\).
            \item Initial state sampling (\(\eta\)): For each trajectory in a batch, the initial time \(t_0\) is drawn uniformly from \([0, T)\), initial wealth \(W_0\) from \(U(0.1, 3.0)\). The initial state variables \(\mathbf{Y}_0 = (Y_{0,1}, \dots, Y_{0,k})^\top\) are sampled such that each component \(Y_{0,i}\) is drawn independently and uniformly from a single common range \([Y_{\min}, Y_{\max}]\). This range is determined by the minimum lower bound and the maximum upper bound observed across all individual factors' typical \( \pm 3\sigma \) intervals:
            \[ Y_{\min} = \min_{i=1,\dots,k}(\theta_{Y,i} - 3\sigma_{Y,ii}), \quad Y_{\max} = \max_{i=1,\dots,k}(\theta_{Y,i} + 3\sigma_{Y,ii}), \]
            where \(\theta_{Y,i}\) is the long-term mean and \(\sigma_{Y,ii}\) is the instantaneous volatility parameter for the \(i\)-th factor. The simulation then runs forward from \(t_0\) to the fixed horizon \(T\) using \(N=20\) time steps, with step size \(\Delta t = (T - t_0) / N\). This approach ensures the policy is trained effectively across the entire time interval \([0, T]\) and a relevant hypercube in the state space defined by the collective factor dynamics.
            \item Variance Reduction: Antithetic variates used.
            \item Wealth constraint: Simulated wealth \(W_t\) clamped at lower bound \(0.1\).
        \end{itemize}
    \item \textbf{Costate Estimation (for Two-Stage Evaluation):} During periodic evaluation, costate \(\lambda\) and derivatives \(\partial_x \lambda, \partial_{\mathbf{Y}} \lambda\) are computed via BPTT (`torch.autograd.grad') applied to Monte Carlo estimates of the expected terminal utility starting from evaluation points \((t_k, W_k, \mathbf{Y}_k)\).
\end{itemize}

\subsection{Benchmark Solution and Evaluation}
\label{app:benchmark_eval} 

As mentioned, the chosen multi-factor OU model allows for a semi-analytical benchmark solution. It is important to note that these experiments focus solely on the terminal wealth objective (\(J = \mathbb{E}[U(X_T)]\)), and thus do not include optimization or a neural network for the consumption policy \(C_\phi\). This simplification is made specifically to enable direct comparison with the semi-analytical benchmark solution derived from the HJB equation, which is readily available for the terminal wealth problem in this affine setting. 
\begin{itemize}
    \item \textbf{Benchmark Computation:} The system of ODEs (including a matrix Riccati equation) associated with the HJB equation for this terminal wealth problem is solved numerically backwards from the fixed horizon \(T\) using \texttt{scipy.integrate.solve\_ivp} (\texttt{Radau} method) to obtain the coefficients determining the value function and the optimal policy \(\boldsymbol{\pi}^*(t, W_t, \mathbf{Y}_t)\); see \citet{kim1996dynamic, liu2007portfolio} for the derivation of these equations.
    \item \textbf{Interpolation:} The benchmark policy (total, myopic, and hedging components) is pre-computed on a grid over \((t, W_t, \mathbf{Y}_t)\) space and stored using `scipy.interpolate.RegularGrid Interpolator' for efficient lookup.
    \item \textbf{Evaluation Metrics:} Performance is evaluated via the Root Mean Squared Error (RMSE) between the policies generated by (i) Baseline PG-DPO (\(\boldsymbol{\pi}_\theta\)) and the benchmark (\(\boldsymbol{\pi}^*\)), and (ii) Two-Stage PG-DPO (using Eq.~\eqref{eq:pi_star} with BPTT-estimated costates) and the benchmark (\(\boldsymbol{\pi}^*\)). RMSEs are computed for total, myopic, and hedging components.
    \item \textbf{Visualization:} Contour plots compare the learned policies and the benchmark over the state space (typically time \(t\) vs. one state variable \(Y_k\)), alongside error plots and costate visualizations.
\end{itemize}
All computations are performed using PyTorch on an NVIDIA A6000 GPU. 

\subsection{Deep BSDE Implementation Details}
\label{app:deep_bsde_impl}

The Deep BSDE method, serving as a benchmark, was implemented based on the standard framework by \citet{han2018solving} and \citet{E2017Deep}. This approach reformulates the problem's PDE as a backward stochastic differential equation (BSDE) and approximates its solution components using neural networks. Our PyTorch implementation addresses the multi-asset, multi-factor portfolio optimization problem from Section~\ref{sec:multiasset_multifactor_KO}.

\begin{itemize}
\item \textbf{Network Architectures:} Two main feedforward neural networks with Tanh activations are utilized:
\begin{itemize}
\item \textbf{V0Net:} Approximates the value function at the initial time of a path. Its output, denoted \(\hat{V}_{t_0}\), is an approximation of the true value, \(V(t_0, X_{t_0}, \mathbf{Y}_{t_0})\). It takes the initial state \((t_0, \log X_{t_0}, \mathbf{Y}_{t_0})\) as input and produces a scalar output (with a final negative softplus transformation). It typically uses 2 hidden layers of 64 units.
\item \textbf{ZNet:} Approximates \(\boldsymbol{Z}_t\), which relates to the scaled gradient of the value function concerning the underlying Brownian motions, at each time step \(t\). It takes the current state \((t, \log X_t, \mathbf{Y}_t)\) as input and outputs an \((n+k)\)-dimensional vector, representing \(\boldsymbol{Z}^X_t \in \mathbb{R}^n\) and \(\boldsymbol{Z}^Y_t \in \mathbb{R}^k\). It typically uses 3 hidden layers of 128 units.
\end{itemize}

\item \textbf{Policy Derivation:} The optimal portfolio \(\boldsymbol{\pi}_t\) is not a direct network output but is dynamically constructed at each step. For CRRA utility, the FOC-derived policy simplifies significantly. The simplification stems from the value function's known separable form in wealth (i.e., \(V \propto X_t^{1-\gamma}\)). This property allows the general intertemporal hedging demand term, \(-\frac{1}{X_t V_{xx}}\Sigma^{-1}(\boldsymbol{\sigma\rho\sigma_Y})V_{x\mathbf{y}}\), to be analytically transformed into the more tractable expression \(\frac{1}{\gamma} \Sigma^{-1} \left( \boldsymbol{\sigma} \boldsymbol{\rho} \right) \left(\boldsymbol{\sigma}_Y \frac{\nabla_{\mathbf{y}}V}{V}\right)\).

The Deep BSDE implementation then creates a policy by combining an analytical myopic demand with a network-based approximation of this tractable hedging term. Specifically, the core theoretical component \(\boldsymbol{\sigma}_Y \frac{\nabla_{\mathbf{y}}V}{V}\) is approximated by the ratio of neural network outputs, \(\frac{\boldsymbol{Z}^Y_t}{\hat{V}_t}\).\footnote{In our BSDE parameterization, the network \texttt{ZNet} learns \(\boldsymbol{Z}^Y_t \approx \nabla_{\mathbf{y}}V \cdot \boldsymbol{\sigma}_Y\), where \(\boldsymbol{\sigma}_Y\) is the diagonal matrix of factor volatilities. The ratio in the policy formula thus correctly approximates the required theoretical term \(\frac{\nabla_{\mathbf{y}}V \cdot \boldsymbol{\sigma}_Y}{V}\).} This leads to the final implemented formula:
\begin{equation}
\boldsymbol{\pi}_t = \underbrace{\frac{1}{\gamma} \Sigma^{-1} \left( \boldsymbol{\sigma} \mathbf{A} \mathbf{Y}_t \right)}_{\text{Myopic Demand}} \quad+ \underbrace{\frac{1}{\gamma} \Sigma^{-1} \left( \boldsymbol{\sigma} \boldsymbol{\rho} \frac{\boldsymbol{Z}^Y_t}{\hat{V}_t} \right)}_{\text{Intertemporal Hedging Demand}}.
\end{equation}
This structure is explicitly designed to capture both demand components, but as shown in our numerical results, the accuracy of this network-based approximation can be a practical point of failure.

\item \textbf{Loss Function:} The training objective is to accurately match the BSDE's terminal condition. For each sample path starting at state \((t_0, X_{t_0}, \mathbf{Y}_{t_0})\), the `V0Net' provides an initial value estimate, \(\hat{V}_{t_0}\). This single initial estimate is then \textbf{evolved forward to the terminal time \(T\) using the discrete-time BSDE formula}, which incorporates the outputs of the `ZNet' at each intermediate time step. This process yields a \textbf{simulated terminal value} for that specific path, which we denote \(\hat{V}_T^{\text{sim}}\). The loss is then calculated as the mean squared error between this simulated terminal value and the \textbf{true terminal condition}—the utility of the wealth \(U(X_T)\) realized at the end of that same path:
$$ \mathcal{L} = \mathbb{E}\left[ \left( \hat{V}_T^{\text{sim}} - U(X_T) \right)^2 \right]. $$
The parameters of both `V0Net' and `ZNet' are optimized to minimize this terminal loss, effectively forcing the initial prediction \(\hat{V}_{t_0}\) to be consistent with the terminal utility value through the lens of the BSDE dynamics.

\item \textbf{Training Details:} The networks are trained for \(10,000\) epochs using the Adam optimizer (initial learning rate \(3 \times 10^{-5}\), with a multi-step decay scheduler) and a batch size of \(1,024\). Each BSDE path is discretized into \(N_{\text{BSDE}}=80\) time steps. Gradient clipping (norm 2.0) and bounding of portfolio weights are applied to enhance training stability.

\item \textbf{Initial State Sampling:} For each training path, the initial time \(t_0\) is sampled uniformly from the discrete time steps over \([0, T-\Delta t]\). The initial wealth \(X_0\) is drawn uniformly from \([0.1, 3.0]\), and initial factor states \(\mathbf{Y}_0\) are sampled uniformly from a hyperrectangle representing their typical dynamic range (e.g., based on \(\pm 3\) stationary standard deviations around their long-term means).

\end{itemize}
The implementation uses PyTorch and is executed on the same NVIDIA A6000 GPU hardware as the PG-DPO experiments to ensure a fair comparison.

\section{Additional Numerical Results - Convergence Speed, Accuracy, and Policy Error Analysis}
\label{app:convergence_accuracy}

\begin{figure}[t!]
  \centering
  \begin{subfigure}[b]{0.8\textwidth}
    \centering
    \includegraphics[width=\textwidth]{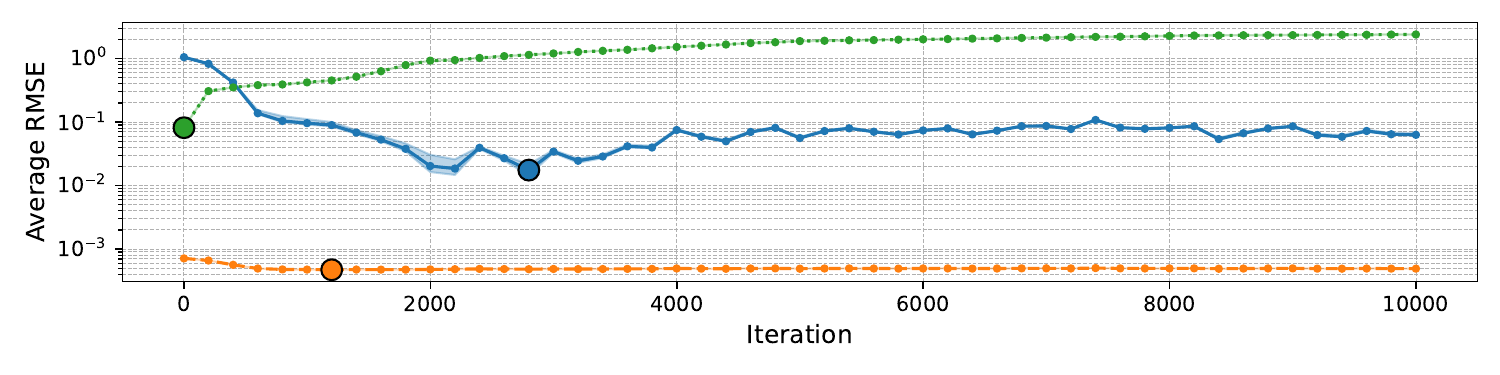} 
    \caption{$n=1,\;k=10$}
    \label{fig:rmse_1_10_updated}
  \end{subfigure}

  \vspace{1.2em}

  \begin{subfigure}[b]{0.8\textwidth}
    \centering
    \includegraphics[width=\textwidth]{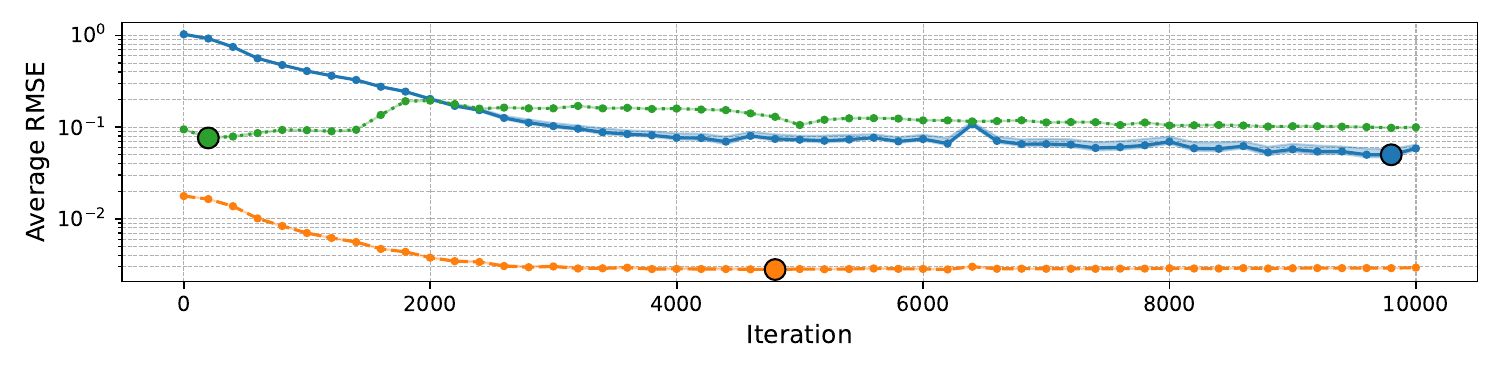} 
    \caption{$n=10,\;k=10$}
    \label{fig:rmse_10_10_updated}
  \end{subfigure}

  \vspace{1.2em}

  \begin{subfigure}[b]{0.8\textwidth}
    \centering
    \includegraphics[width=\textwidth]{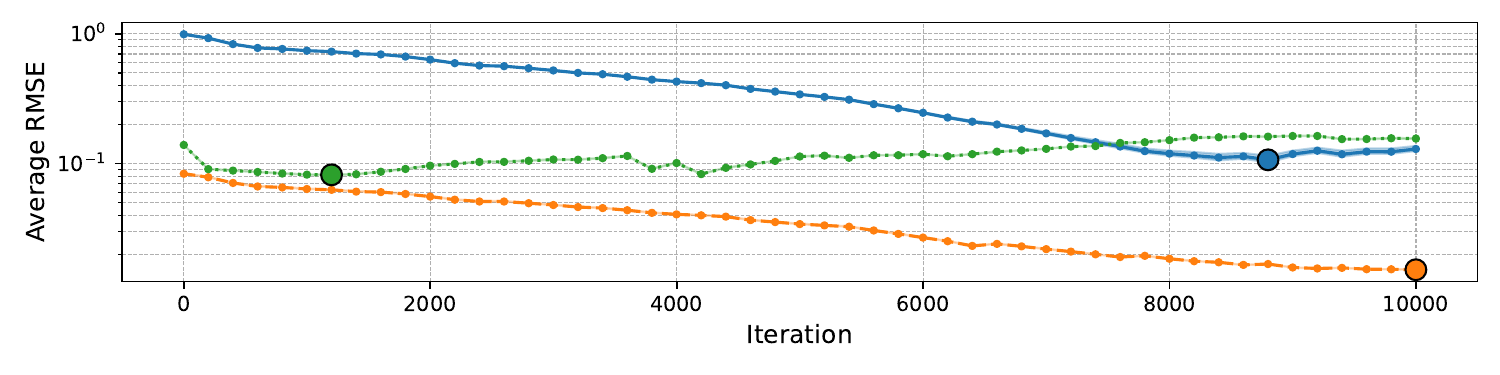} 
    \caption{$n=50,\;k=10$}
    \label{fig:rmse_50_10_updated}
  \end{subfigure}

  \caption{Average policy RMSE comparison between Baseline PG-DPO (e.g., blue), Projected PG-DPO (P-PGDPO) (e.g., orange), and Deep BSDE (e.g., green) relative to the benchmark solution across different asset dimensions ($n$) with $k=10$ factors. Values averaged over $k=10$ evaluation `slices.' For each slice $j=1,\dots,10$, the evaluation (e.g., for RMSE calculation or visualization) focuses on state factor $Y_j$, while other factors $Y_l, l \neq j,$ are held constant at their respective long-term means $\theta_{Y,l}$. Note the logarithmic scale on the $y$-axis.}
  \label{fig:rmse_plots_updated}
\end{figure}

This section quantitatively compares the convergence speed and final accuracy of the proposed PG-DPO framework—including both the Baseline and Projected PG-DPO (P-PGDPO) variants—against the Deep BSDE method, when applied to derive control policies. We evaluate performance across various problem dimensions, focusing on how the number of assets (\( n \)) and the number of state factors (\( k \)) affect policy accuracy.

Figure~\ref{fig:rmse_plots_updated} shows the average Root Mean Squared Error (RMSE) between the learned or derived policies and the semi-analytical benchmark policy. The RMSE is plotted against the number of training iterations on a logarithmic scale for experiments with \( k = 10 \) state variables and varying asset counts (\( n = 1, 10, 50 \)).
For P-PGDPO, the reported RMSE corresponds to the total portfolio policy \( \boldsymbol{\pi}^{\mathrm{PMP}} \) (Eq.~\eqref{eq:pi_star}), which combines the myopic and intertemporal hedging components. As discussed in Section~\ref{sec:multi_asset_numerical}, the myopic error is typically negligible—on the order of \( O(10^{-7}) \). Thus, the total RMSE for P-PGDPO is effectively governed by the accuracy of the \textit{hedging component}, which depends on estimating the costate gradient \( \partial_{\mathbf{Y}} \lambda_t^* \) via BPTT.


Figures \ref{fig:policy_error} and \ref{fig:decomposition error} display the error maps across models. In particular, Figure~\ref{fig:policy_error} illustrates the policy errors for the first asset ($i=1$) across three models-baseline PG-DPO, P-PGDPO, and Deep BSDE- in the high-dimensional setting with $n=50$ risky assets and $k=10$ state variables. The figure plots the difference between the total portfolio policy (i.e., the sum of the myopic and intertemporal hedging demands) generated by each method and the corresponding analytic benchmark solution. Panel (a) shows the Baseline PG-DPO, where substantial deviations from the benchmark are visible, particularly near maturity and for extreme values of $Y_1$. Panel (b) reports the results for P-PGDPO, which achieves significantly lower errors overall. The projection step based on stabilized costates produces policies that are visually close to the benchmark, with errors an order of magnitude smaller than in the baseline. Panel (c) presents the Deep BSDE method, which captures the broad structure of the policy but exhibits notable inaccuracies in regions where intertemporal hedging demand is critical. 

\begin{figure}[t!]
 \centering
 \begin{subfigure}[b]{0.32\textwidth}
  \centering
  \includegraphics[width=\textwidth]{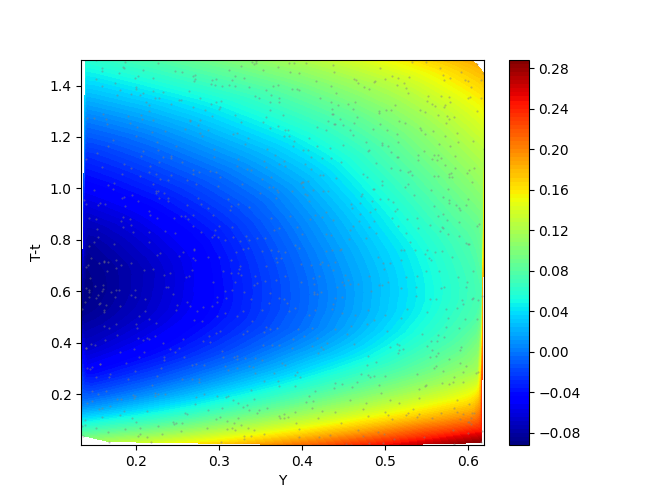}
  \caption{Baseline PG-DPO Error}
  \label{fig:benchmark_policy_total_row_2}
 \end{subfigure}
 \hfill
 \begin{subfigure}[b]{0.32\textwidth}
  \centering
  \includegraphics[width=\textwidth]{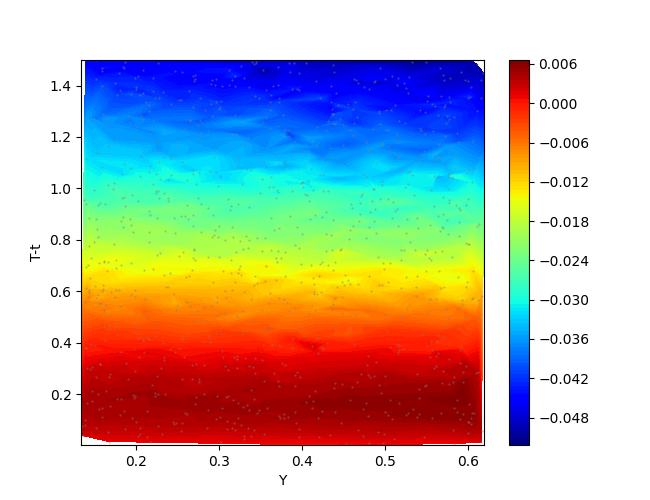}
  \caption{P-PGDPO Error}
  \label{fig:n50k10_2pg_err} 
 \end{subfigure}
 \hfill
 \begin{subfigure}[b]{0.32\textwidth}
  \centering
  \includegraphics[width=\textwidth]{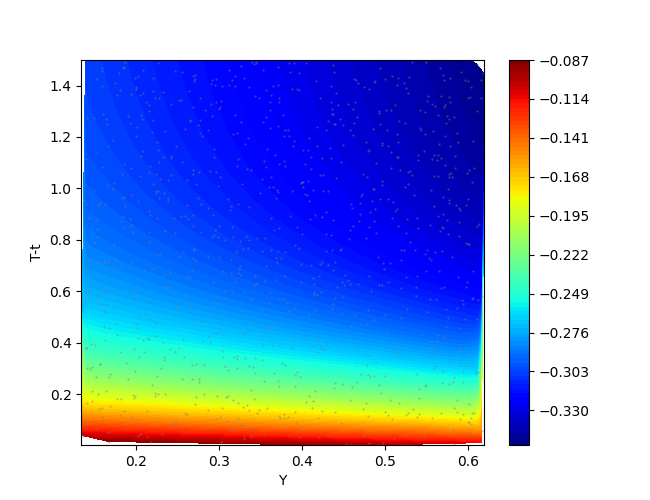}
  \caption{Deep BSDE Error}
  \label{fig:n50k10_bsde_err}
 \end{subfigure}
 \caption{Policy errors for the first asset ($i=1$) across models in the case where $n=50, k=10$. Errors are plotted against $Y_1$ at a representative time-to-maturity $T-t$ (e.g., $t=(T-T_0)/2$ if $T_0$ is initial time, assuming $T=1.5$ is the horizon from $t=0$; thus, for instance, at $t=0.75$ or time-to-maturity $0.75$). Other state factors $Y_{j \neq 1}$ are held at their long-term means $\theta_{Y,j}$.}
 \label{fig:policy_error}
\end{figure}

\begin{figure}[b!]
 \centering
  \hfill
\begin{subfigure}[b]{0.4\textwidth}
  \centering
  \includegraphics[width=\textwidth]{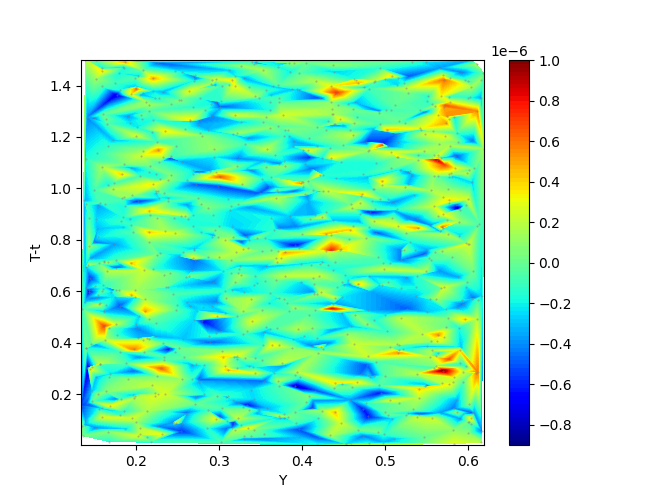}
  \caption{Myopic Error}
  \label{fig:n50k10_myopic_err_new}
 \end{subfigure}
 \hfill
 \begin{subfigure}[b]{0.4\textwidth}
  \centering
  \includegraphics[width=\textwidth]{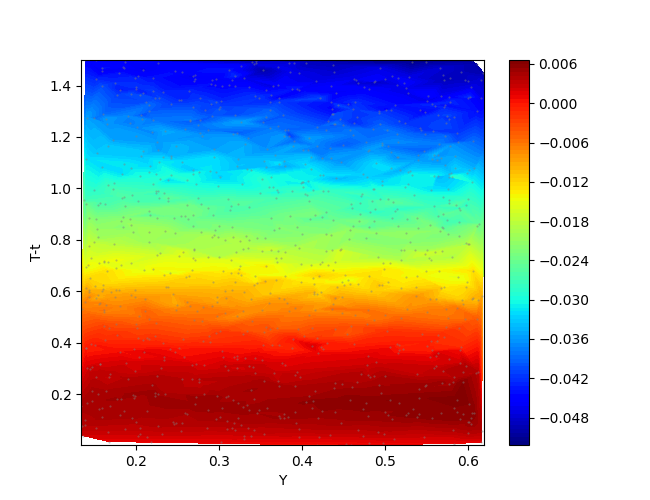}
  \caption{Hedging Error}
  \label{fig:n50k10_hedging_err_new}
 \end{subfigure}
  \hfill
 \caption{Decomposition errors of the Projected PG-DPO (P-PGDPO) policy for the first asset ($i=1$) in the $n=50, k=10$ case (results from iteration 6,400 as plotted), plotted against $Y_1$ at a representative time-to-maturity $T-t$. The first figure shows the error for myopic demand and the second figure shows the error of intertemporal hedging demand. Other state factors $Y_{j \neq 1}$ are held at their long-term means $\theta_{Y,j}$.}
 \label{fig:decomposition error}
\end{figure}

In addition to the total portfolio errors, Figure \ref{fig:decomposition error} provides a decomposition of the policy errors into the myopic and intertemporal hedging components for the first asset ($i=1$) under the setting with $n=50$ risky assets and $k=10$ state variables. The left panel shows the error map for the myopic demand, while the right panel shows the error for the intertemporal hedging demand, both plotted against the first variable $Y_1$ and time-to-maturity $(T-t)$, with all other state variables fixed as their long-run means. The results demonstrate that the myopic demand is recovered with residual errors on the order of $10^{-6}$, confirming the robustness of the costate-driven projection step for the static component of the portfolio. By contrast, the hedging demand exhibits larger and more structured errors, concentrated in regions where horizon effects are most pronounced. This reflects the intrinsic difficulty of estimating the cross-derivatives of the value function that drive intertemporal hedging behavior, particularly in high-dimensional environments.

In sum, Figures \ref{fig:policy_error} and \ref{fig:decomposition error} illustrate that while P-PGDPO achieves highly accurate recovery of both myopic and hedging components, the primary source of residual policy error originates from the hedging term, underscoring the importance of our methodological focus on stabilizing and projecting costate estimates to capture horizon-dependent dynamics.

\end{document}